\documentclass[a4paper,UKenglish,cleveref, autoref]{lipics-v2021}

\newcommand{\ERCagreement}{{\begin{minipage}{.56\textwidth}This paper is part of a project that has received funding from the European Research Council (ERC) under the European Union's Horizon 2020 research and innovation programme (grant agreement No 810115 -- {\sc Dynasnet}). \end{minipage}\hfill\begin{minipage}{.33\textwidth}\includegraphics[width=\textwidth]{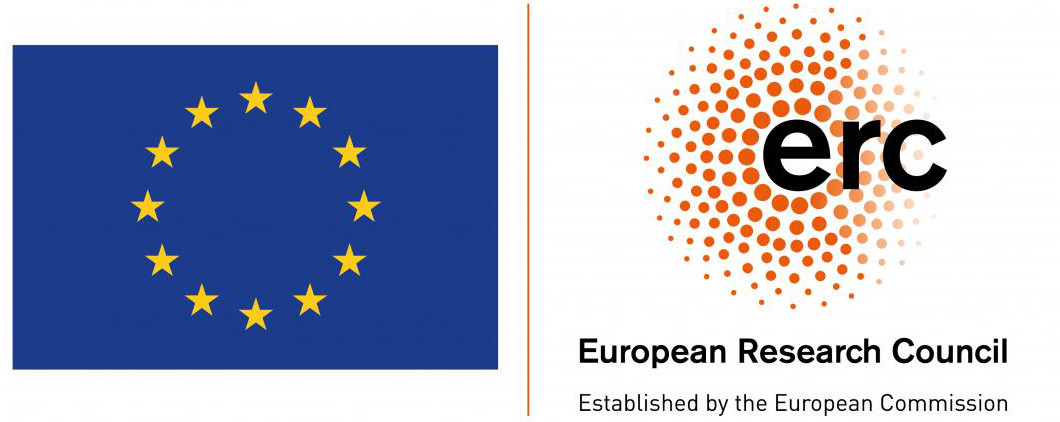}\end{minipage}\hfill}}

\title{Twin-width V: linear minors, modular counting, and matrix multiplication}
\titlerunning{Twin-width V: linear minors, modular counting, and matrix multiplication}

\author{\'{E}douard Bonnet}{Univ Lyon, CNRS, ENS de Lyon, Université Claude Bernard Lyon 1, LIP UMR5668, France \and \url{http://perso.ens-lyon.fr/edouard.bonnet/}}{edouard.bonnet@ens-lyon.fr}{https://orcid.org/0000-0002-1653-5822}{}
\author{Ugo Giocanti}{Univ Lyon, CNRS, ENS de Lyon, Universit\'{e} Claude Bernard Lyon 1, LIP UMR5668, France}{ugo.giocanti@ens-lyon.fr}{}{}
\author{Patrice Ossona de Mendez}{Centre d'Analyse et de Mathématique Sociales CNRS UMR 8557, France \and and Computer Science Institute of Charles University (IUUK), Praha, Czech Republic \and \url{http://cams.ehess.fr/patrice-ossona-de-mendez/} }{pom@ehess.fr}{https://orcid.org/0000-0003-0724-3729}{}
\author{St\'{e}phan Thomass\'{e}}{Univ Lyon, CNRS, ENS de Lyon, Universit\'{e} Claude Bernard Lyon 1, LIP UMR5668, France}{stephan.thomasse@ens-lyon.fr}{}{}

\authorrunning{\'E. Bonnet, U. Giocanti, P. {Ossona de Mendez}, S. Thomassé}

\Copyright{Édouard Bonnet, Ugo Giocanti, Patrice {Ossona de Mendez}, Stéphan Thomassé}



\category{}

\relatedversion{}

\supplement{}

\funding{\ERCagreement This paper was supported by the ANR projects (French National Research Agency) TWIN-WIDTH (ANR-21-CE48-0014-01) and Digraphs (ANR-19-CE48-0013-01).}

\acknowledgements{}

\nolinenumbers 

\hideLIPIcs  

\EventEditors{John Q. Open and Joan R. Access}
\EventNoEds{2}
\EventLongTitle{42nd Conference on Very Important Topics (CVIT 2016)}
\EventShortTitle{CVIT 2016}
\EventAcronym{CVIT}
\EventYear{2016}
\EventDate{December 24--27, 2016}
\EventLocation{Little Whinging, United Kingdom}
\EventLogo{}
\SeriesVolume{42}
\ArticleNo{23}

\usepackage[utf8]{inputenc}  

\usepackage[T1]{fontenc}
\usepackage{lmodern}

\usepackage[colorinlistoftodos,bordercolor=orange,backgroundcolor=orange!20,linecolor=orange,textsize=normalsize]{todonotes}

\usepackage[ruled,vlined,linesnumbered]{algorithm2e}

\usepackage{amsmath}  
\usepackage{amssymb}     
\usepackage{bbm}
\usepackage{tkz-graph}
\usepackage{complexity}

\usepackage{comment}

\usepackage{bm}
\usepackage{hyperref}
\usepackage{mathrsfs}
\usepackage{mathtools} 
\usepackage{paralist}
\usepackage{fixmath}
\usepackage{interval}
\usepackage{dsfont}

\newcommand\pref[1]{($\ref{#1}$)}

\Crefname{figure}{Figure}{Figures}

\crefformat{equation}{\textup{#2(#1)#3}}
\crefrangeformat{equation}{\textup{#3(#1)#4--#5(#2)#6}}
\crefmultiformat{equation}{\textup{#2(#1)#3}}{ and \textup{#2(#1)#3}}
{, \textup{#2(#1)#3}}{, and \textup{#2(#1)#3}}
\crefrangemultiformat{equation}{\textup{#3(#1)#4--#5(#2)#6}}%
{ and \textup{#3(#1)#4--#5(#2)#6}}{, \textup{#3(#1)#4--#5(#2)#6}}{, and \textup{#3(#1)#4--#5(#2)#6}}

\Crefformat{equation}{#2Equation~\textup{(#1)}#3}
\Crefrangeformat{equation}{Equations~\textup{#3(#1)#4--#5(#2)#6}}
\Crefmultiformat{equation}{Equations~\textup{#2(#1)#3}}{ and \textup{#2(#1)#3}}
{, \textup{#2(#1)#3}}{, and \textup{#2(#1)#3}}
\Crefrangemultiformat{equation}{Equations~\textup{#3(#1)#4--#5(#2)#6}}%
{ and \textup{#3(#1)#4--#5(#2)#6}}{, \textup{#3(#1)#4--#5(#2)#6}}{, and \textup{#3(#1)#4--#5(#2)#6}}

\usepackage{xspace}
\usepackage{tikz}

\usetikzlibrary{arrows}
\usetikzlibrary{patterns}
\usetikzlibrary{calc}
\usetikzlibrary{fit}
\usetikzlibrary{shapes}
\usetikzlibrary{positioning}
\usetikzlibrary{math}
\usetikzlibrary{external}

\usepackage{ifthen}

\usepackage[scr=boondox,scrscaled=1.05]{mathalfa}

\makeatletter
\newtheorem*{rep@theorem}{\rep@title}
\newcommand{\newreptheorem}[2]{%
\newenvironment{rep#1}[1]{%
 \def\rep@title{#2 \ref{##1}}%
 \begin{rep@theorem}}%
 {\end{rep@theorem}}}
\makeatother

\newreptheorem{theorem}{Theorem}

\newenvironment{proofofclaim}{\noindent \textsc{Proof of the claim:}}{\hfill$\Diamond$\medskip}

\renewcommand{\geq}{\geqslant}
\renewcommand{\leq}{\leqslant}
\renewcommand{\preceq}{\preccurlyeq}

\newcommand{\tree}{\mathcal{T}}
\newcommand{\bb}{\mathcal{B}}

\newcommand{\CC}{\mathcal C}
\newcommand{\DD}{\mathcal D}

\newcommand{\set}[1]{\ensuremath{\{#1\}}} 
\newcommand{\setof}[2]{\set{#1\mid#2}} 
\newcommand{\from}{\colon} 
\newcommand{\str}[1]{\mathbf{#1}} 
\newcommand{\degr}{\mathrm{deg}}

\DeclareMathOperator{\mt}{mt}

\DeclareMathOperator{\reduct}{\text{\sf Reduct}}

\newcommand*\sg[1]{\left \{ #1 \right \}}

\theoremstyle{definition}

\theoremstyle{remark}

\newcommand{\fomc}{\text{\rm FO+MOD}\xspace}
\newcommand{\fo}{\text{\rm FO}\xspace}
\DeclareMathOperator{\minclos}{Clos_{\text{pm}}}
\DeclareMathOperator{\linmin}{Clos_{\text{lm}}}
\DeclareMathOperator{\tww}{tww}
\newcommand{\Mall}{\mathcal M_{\text{\rm all}}}

\newcommand{\infpar}{\leq_{\text{pm}}}
\newcommand{\suppar}{\geq_{\text{pm}}}
\newcommand{\inflin}{\leq_{\text{lm}}}

\newcommand{\strict}{parity-minor free\xspace}
\newcommand{\lmf}{linear-minor free\xspace}
\newcommand{\lm}{linear minor\xspace}
\newcommand{\lms}{linear minors\xspace}

\newcommand{\pmc}{\textsc{Parity Minor Containment}\xspace}
\newcommand{\imc}{\textsc{Interval Minor Containment}\xspace}
\newcommand{\gmc}{\textsc{$k$-Grid Minor}\xspace}
\newcommand{\mmc}{\textsc{$k$-Mixed Minor}\xspace}

\newcommand{\Oo}{\mathcal O}

\newcommand{\row}{\text{rows}}
\newcommand{\col}{\text{cols}}

\begin{document}

\maketitle

\begin{abstract}
  We continue developing the theory around the twin-width of totally ordered binary structures (or equivalently, matrices over a finite alphabet), initiated in the previous paper of the series.
  We first introduce the notion of parity and linear minors of a matrix, which consists of iteratively replacing consecutive rows or consecutive columns with a linear combination of them.
  We show that a matrix class (i.e., a set of matrices closed under taking submatrices) has bounded twin-width if and only if its linear-minor closure does not contain all matrices.
  We observe that the fixed-parameter tractable (\FPT) algorithm for first-order model checking on structures given with an $\Oo(1)$-sequence (certificate of bounded twin-width) and the fact that first-order transductions of bounded twin-width classes have bounded twin-width, both established in \emph{Twin-width~I}, extend to first-order logic with modular counting quantifiers.
  We make explicit a win-win argument obtained as a by-product of \emph{Twin-width~IV}, and somewhat similar to bidimensionality, that we call rank-bidimensionality.
  This generalizes the seminal work of Guillemot and Marx [SODA '14], which builds on the Marcus-Tardos theorem [JCTA '04].
  It works on general matrices (not only on classes of bounded twin-width) and, for example, yields \FPT~algorithms deciding if a small matrix is a parity or a linear minor of another matrix given in input, or exactly computing the grid or mixed number of a given matrix (i.e., the maximum integer $k$ such that the row set and the column set of the matrix can be partitioned into $k$ intervals, with each of the $k^2$ defined cells containing a non-zero entry, or two distinct rows and two distinct columns, respectively).

  Armed with the above-mentioned extension to modular counting, we show that the twin-width of the product of two conformal matrices $A, B$ over a finite field is bounded by a function of the twin-width of $A$, of $B$, and of the size of the field.
  Furthermore, if $A$ and $B$ are $n \times n$ matrices of twin-width~$d$ over $\mathbb F_q$, we show that $AB$ can be computed in time $\Oo_{d,q}(n^2 \log n)$.
  
  We finally present an \emph{ad hoc} algorithm to efficiently multiply two matrices of bounded twin-width, with a single-exponential dependence in the twin-width bound.
  More precisely, pipelined to observations and results of Pilipczuk et al.~[STACS~'22], we obtain the following.
  If the inputs are given in a compact tree-like form (witnessing twin-width at most~$d$), called twin-decomposition of width~$d$, then two $n \times n$ matrices $A, B$ over $\mathbb F_2$ can be multiplied in time $4^{d+o(d)}n$, in the sense that a twin-decomposition of their product $AB$, with width $2^{d+o(d)}$, is output within that time, and each entry of $AB$ can be queried in time $\Oo_d(\log \log n)$.
  Furthermore, for every $\varepsilon > 0$, the query time can be brought to constant time $\Oo(1/\varepsilon)$ if the running time is increased to near-linear $4^{d+o(d)}n^{1+\varepsilon}$.
  Notably, the running time is sublinear (essentially square root) in the number of (non-zero) entries.
\end{abstract}

\section{Introduction}\label{sec:intro}


Since its introduction, the treewidth of a graph has proved to be a particularly important concept in graph theory, both in finite model theory~\cite{gradel2007finite}, in algorithmic design (see for instance the textbook of Cygan et al.~\cite[Chapter 7]{complexity}) and in structural analysis (see the Graph Minors series of Robertson and Seymour~\cite{RobertsonS83}).
This invariant is strongly related to the concept of graph minor.
Recall that a minor of a graph is a graph obtained by a succession of edge contractions and vertex or edge deletions.
The treewidth of a graph is monotone with respect to this operation, in the sense that the treewidth of a minor of a graph~$G$ cannot be larger than the treewidth of~$G$.
By a classical theorem by Robertson and Seymour~\cite{ROBERTSON198692}, a class of graphs has bounded treewidth if and only if its \emph{minor closure} (that is, the set of all the  minors of graphs in the class) does not contain all grids.
In particular, a~graph with huge treewidth admits a large square grid as a minor. This result, as well as its subsequent qualitative improvements (See~\cite{ChuzhoyT21}, for instance), is the basis of the so-called \emph{bidimensionality} algorithmic technique, a win-win argument leveraging low treewidth or the existence of a large grid minor~\cite{bidim}.

This paper is the fifth of a series dedicated to a novel invariant of binary structures, the \emph{twin-width} (See~\cref{sec:prelim} for formal definitions).
This invariant appeared to be particularly relevant for the study of \emph{ordered} binary structures, and especially matrices over a finite alphabet \cite{twin-width4}. 
Some of our results will only need the matrix entries to belong to a finite alphabet, while some will require the entries to belong to a~finite field.
A \emph{submatrix} of a matrix is obtained by deleting some rows and columns.
Most of our results concern (infinite) sets of matrices.
In our framework, it will be natural to consider sets of matrices closed under the operation of taking a submatrix.
Sets of matrices with this property are called \emph{matrix classes}, analogously to permutation classes, which are classes of permutations closed under taking subpermutations. 
The alphabet or field being fixed (and having at least two elements), a matrix class is said to be \emph{proper} if it does not include all matrices with entries in the prescribed alphabet or field.

The notion of a \emph{rank Latin division} has been introduced in~\cite{twin-width4} (see definition in \cref{sec:strict}).
It consists of a~regular partition of the rows and columns delimiting blocks that either have constant entries or have full rank, in a globally controlled way, where full-rank blocks draw a~universal permutation (see~\cref{fig:latin}). 
Just as the grids acts as a witness of a large treewidth, the rank Latin divisions witness a large twin-width: a matrix has either small twin-width or has a submatrix with a~large rank Latin division.
This is effective: in \FPT~time, either a contraction sequence of the matrix (witnessing that the twin-width is low) is output or a~large rank Latin division is found in a~submatrix (witnessing that the twin-width is high).

In this paper, we introduce an operation that plays a somewhat analogous role with respect to the twin-width of ordered binary structures that taking a minor plays with treewidth.
Applied to 0,1-matrices, this operation consists of a succession of row or column deletions, and replacements of two consecutive rows or columns by their entry-wise sum (modulo $2$).  Any 0,1-matrix obtained this way is a  \emph{parity minor} of the original matrix. 
More generally, when applied to  matrices over a finite field $\mathbb F_p$, this operation consists of~a succession of replacements of two consecutive rows or columns by a~linear combination of these (over $\mathbb F_p$), and any matrix over $\mathbb F_p$ obtained this way is a~\emph{linear minor} of the original matrix. 
We say that a matrix class \emph{excludes} a matrix $M$ as a~linear minor if $M$ is not a linear minor of a~matrix in the class.
As expected, every matrix is a linear minor of any matrix having a~sufficiently large rank Latin division, thus classes with unbounded twin-width do not exclude any matrix as a~linear minor (\Cref{lem:strict-01}). 
It appears that this necessary condition is also sufficient.
 
 \begin{theorem}\label{thm:intro-pm}
 	A matrix class  over a finite field has bounded twin-width if and only if it~excludes some matrix as a linear minor.
 \end{theorem}

However, our proof of \Cref{thm:intro-pm} involves some (finite) model theoretic arguments. 

 From a model theoretical point of view, a matrix over a finite alphabet of size $p$ is  seen as  a structure with two linearly ordered sets of elements, the row and column index sets, and $p$ binary predicates expressing the presence of a particular symbol at a specific entry. 
 The logical formulas we will consider will allow distinguishing row and column indices, comparing indices of a same sort, and testing whether the entry of the matrix defined by two indices contains a given symbol. 

 It appears that the twin-width of ordered structured behaves very nicely with respect to first-order logic and (as we shall see) its modulo-counting extension.
 This situation is reminiscent of the relation of treewidth (and cliquewidth) with monadic second-order logic \cite{Courcelle1} and its modulo-counting extension \cite{Courcelle00}.

 Indeed, it follows from the results proved  in the first paper of the series~\cite{twin-width1}, that 
 first-order model checking (that is: the problem of  deciding whether a first-order (FO) sentence $\varphi$ is satisfied on a structure) is fixed-parameter tractable on matrices over a~finite alphabet, when parametrized by $\varphi$, the size of the alphabet, and the twin-width of the matrix, provided that some so-called \emph{$d$-sequence} witnessing the upper bound on the twin-width is given together with the matrix.
 Here we observe that this result extends to the more expressive first-order logic with modulo-counting (FO+MOD), which is the logic obtained by adding to the standard first-order constructions new quantifiers $\exists^{i[p]}$, 
 where ``$\exists^{i[p]}x\enspace\varphi(x)$'' expresses that the number of witnesses $x$ for the formula $\varphi$ is congruent to $i$ modulo $p$.

Logical formulas also allow defining new structures from an original structure.  This is the essence of the notion of transduction. A \emph{transduction} of binary structures $\mathsf T$ first color the elements of a given binary structure $\mathbf A$ in all possible ways, thus constructing a set of colored structures. Then,  each of these colored structures gives rise to a new binary structure by means of fixed logical formulas, thus constructing a set $\mathsf T(\mathbf A)$ of derived structures, the transduction of $\mathbf A$ by $\mathsf T$.
A set $\mathcal D$ of structures is a transduction of a set $\mathcal C$ of structures if there exists a transduction $\mathsf T$ with $\mathcal D\subseteq \mathsf T(\mathcal C)$.
A set $\mathcal C$ of structures is \emph{monadically dependent} if the set of all finite graphs is not a transduction of $\mathcal C$. (While this actually follows from~\cite{BS1985monadic}, we will take here this characteristic property as a definition of monadic dependence.) Note that it has been recently proved \cite{MonNIP} that, for hereditary classes of structures (like matrix classes)   monadic dependence coincides with the classical notion of dependence (or NIP), which is one of the most fundamental dividing line in model theory; a proof of such a collapse in the particular case of hereditary classes of ordered graphs was previously shown in \cite{twin-width4}.

In~\cite{twin-width1}, it was proved that for every FO-transduction $\mathsf T$ of binary structures, the maximum twin-width of a structure in $\mathsf T(\mathbf A)$ is bounded by a function (depending on $\mathsf T$) of the twin-width of $\mathbf A$. We also extend this result to FO+MOD-transductions (\Cref{thm:fomc-closure}).
As an example, there is a transduction $\mathsf L_p$ such that for every matrix $M$ over $\mathbb F_p$, the set $\mathsf L_p(M)$ is exactly the set of all linear minors of $M$ (\Cref{lem:min-clos-transduction}). Thus, the closure by linear minors of a~matrix class with bounded twin-width also has bounded twin-width, from which \Cref{thm:intro-pm} follows.

 Together with the results established in~\cite{twin-width4}, this leads to the following equivalence, where the equivalence with the properties in bold is proved in the current paper.
 
 \begin{theorem}
 	Given a matrix class $\mathcal M$ over a finite field, the following are equivalent.
 	\begin{compactenum}[$(i)$]
 		\item  $\mathcal M$ has bounded twin-width;
 		\item  \textbf{$\mathcal M$ excludes a linear minor;}
 		\item $\mathcal M$ is monadically dependent;
 		\item   every matrix class that is an \fo-transduction of $\mathcal M$ is proper;
 		\item  \textbf{every  matrix class that is an \fomc-transduction of $\mathcal M$ is proper;}
 		\item $\mathcal M$ is small (i.e. the number of $n\times n$ matrices in $\mathcal M$ is at most $2^{O(n)}$);
 		\item  \textbf{every \fomc-transduction of $\mathcal M$ is small;}
 		\item $\mathcal M$ is subfactorial (i.e. the number of $n\times n$ matrices in $\mathcal M$ is less than $n!$, for sufficiently large $n$).
 	\end{compactenum}
 	Assuming that \FPT~$\neq$ \AW$[*]$, those conditions are further equivalent to:
 	\begin{compactenum}[$(i)$]
\setcounter{enumi}{8}
 		\item {\rm FO}-model checking is {\FPT} on $\mathcal M$;
 		\item \textbf{\fomc-model checking is {\FPT} on $\mathcal M$}.
 	\end{compactenum}
 \end{theorem}
 
We now consider some consequences of these results.
 
 We call \emph{rank-bidimensional} a parameterized problem defined on matrices whenever the presence of a large rank Latin division in a submatrix incurs an (easy) \FPT~algorithm.
 Thus, we get the following.
 
 \begin{theorem}\label{thm:intro-rank-bidim}
 	Every \fomc-definable rank-bidimensional problem is in \FPT.
 \end{theorem}

 From~\cref{thm:intro-rank-bidim}, we obtain \FPT~algorithms for deciding if a (small) matrix is a \lm (or parity minor) of another matrix, for exactly computing the grid number, mixed number, and grid rank of a matrix (see~\Cref{sec:prelim} for definitions).
 
 \medskip
 
 Next we show that, over a finite field, the square $M^2$ of a matrix $M$ with bounded twin-width has bounded twin-width, by expressing the squaring operation as an \fomc-transduction.
 From the characterization in terms of large rank Latin division of submatrices, it follows that if two matrices $A$ and $B$ have small twin-width, then so does the matrix $\left(\begin{smallmatrix} 0&A \\ B&0 \end{smallmatrix}\right)$. As $\left(\begin{smallmatrix} 0&A \\ B&0 \end{smallmatrix}\right)^2=\left(\begin{smallmatrix} AB&0 \\ 0&BA \end{smallmatrix}\right)$, we deduce:
 
 \begin{theorem}\label{thm:intro-bdtww-product}
 	There is a computable function $f: \mathbb N^2 \to \mathbb N$ such that the following holds.
 	Let $A$ and $B$ be two conformal matrices over a finite field $\mathbb F_q$, both of twin-width at most~$d$.
 	Then the twin-width of the product $AB$ is at most~$f(d,q)$.
 \end{theorem}
 Note that, by similar arguments, the sum of two conformal matrices over $\mathbb F_q$, both of twin-width at most~$d$, has twin-width at most $f'(d,q)$, for some computable function $f'$.

 We now consider the problem from an algorithmic point of view.
 From a computational perspective, the data structures used to encode matrices over a finite field are crucial.
 Encoding matrices as bipartite binary structures allows using the machinery developed for (ordered) graphs.
 In this setting, natural witnesses for twin-width boundedness are \mbox{\emph{$d$-sequences}} (or, \emph{contraction sequences}), which we mentioned earlier when  discussing first-order model checking complexity on classes with bounded twin-width (see \Cref{subsec:contraction-seq} for a formal definition).
 A naive implementation of the algorithm presented in~\cite[Theorem 2]{twin-width4} runs in time $\exp(\exp(\Oo(d^2 \log d))) n^3$, and outputs a $2^{\Oo(d^4)}$-sequence if the twin-width is indeed at most~$d$.  We show how to bring the dependence in $n$ down to $\Oo_d(n^2 \log n)$ (\Cref{thm:approx-tww-quad}).
 
 Gajarský et al.~\cite{Gajarsky22}, building on Pilipczuk et al.~\cite{PilipczukSZ22},  showed that given an $n$-vertex graph (or binary structure) $G$ with a~$d$-sequence and a first-order formula $\varphi(x_1, \ldots, x_k)$, one can compute in time $\Oo_{d,\varphi}(n^{1+\varepsilon})$ (resp.~$\Oo_{d,\varphi}(n)$) a data structure that answers for any query $v_1, \ldots, v_k \in V(G)$ whether $G \models \varphi(v_1, \ldots, v_k)$ in time $\Oo_{d,\varphi}(1/\varepsilon)$ (resp.~$\Oo_{d,\varphi}(\log \log n)$).
 This result can be extended to \fomc.
 Then, the squaring operation can be performed by means of a simple interpretation, which gives a near-linear representation of $M^2$ (in the domain size, that is, sublinear in the number of matrix entries) where entries can be queried in constant time.
 
 We thus obtain an algorithm that takes two matrices of bounded twin-width (\emph{without} witnesses) and outputs their product in quasilinear time in the number of entries. 
 
 \begin{theorem}\label{thm:intro-square}
 	Given two $n \times n$ matrices $A$ and $B$ over $\mathbb F_q$, both of twin-width at most~$d$, there is an algorithm to compute their product $AB$ in time $\Oo_{d,q}(n^2 \log n)$. 
 \end{theorem}

However, this algorithm is not practical due to the acute dependence on the twin-width bound.
We thus place ourselves in a setting where inputs are already in compact form (witnessing low twin-width).
The use of an adapted internal representation is a classical technique of digital computing (Fourier transform, redundant representation of numbers, etc.). Likewise, it appears that convenient representations of matrices of bounded twin-width for 
matrix computations are~\emph{twin-decompositions}~\cite{twin-width3,twin-width&permutations}.
Informally, a twin-decomposition is a~tree whose leaves are bijectively mapped to the domain of the structure (here, to the row and column indices), and internal nodes are ordered and naturally correspond to contractions.
The binary relations (here, the entries) are encoded by additional edges joining pairs of nodes of the tree, and respecting some specific rules. Every binary structure with bounded twin-width has a twin-decomposition with linearly many extra edges, hence the twin-decomposition forms a degenerate graph.
The \emph{width} of the twin-decomposition is related to this degeneracy (see~\cref{sec:prelim} for precise definitions).
Notice that a~twin-decomposition of constant width takes quasilinear space to describe a~set of binary relations with possibly quadratically many pairs.
We show that a~twin-decomposition can be computed from a~contraction sequence in quadratic time and observe that, conversely, a~contraction sequence can be computed from a~twin-decomposition in linear time.

Our last contribution is an \emph{ad hoc} efficient matrix multiplication algorithm for matrices over $\mathbb F_p$ with bounded twin-width, which we state here in the case of matrices over $\mathbb F_2$.

\begin{theorem}[\cref{thm:intro-square-adhoc}+\cite{PilipczukSZ22}]\label{thm:intro-square-adhoc-chan}
	Let $A$ and $B$ be two $n \times n$ matrices over $\mathbb F_2$ given in the form of twin-decompositions of width at most~$d$.
	For every $\varepsilon > 0$, there is a~$4^{d+o(d)}n^{1+\varepsilon}$-time algorithm that outputs a twin-decomposition of the product $AB$ of width $2^{d+o(d)}$ and a data structure of size $2^{d+o(d)}n^{1+\varepsilon}$ such that querying an entry of $AB$ takes time $\Oo(1/\varepsilon)$.
\end{theorem}

This result is based on an admittedly  technical  algorithm computing a~twin-decomposition of the square $M^2$ of a matrix $M$ from a~twin-decomposition of $M$, which extends to a matrix multiplication algorithm for matrices each represented by a twin-decomposition.

\begin{theorem}\label{thm:intro-square-adhoc}
	Let $A$ and $B$ be two $n \times n$ matrices over $\mathbb F_2$ given in the form of twin-decompositions of width at most~$d$.
	There is a~$4^{d+o(d)}n$-time algorithm that outputs a~twin-decomposition of the product $AB$ of width $2^{d+o(d)}$.
\end{theorem}

The entries of $AB$ can then be queried in time essentially the height of the twin-decomposition, which can be made logarithmic.
However, by computing the data structure introduced by Pilipczuk et al.~\cite{PilipczukSZ22} in $\Oo(d' n^{1+\varepsilon})$ time and space where $d'$ upperbounds the width of a~twin-decomposition of the matrix, the entry queries can be performed in $\Oo(1/\varepsilon)$ time. 
\cref{thm:intro-square-adhoc,thm:intro-square-adhoc-chan} carry over on any finite field $\mathbb F_q$ with running time $q^{2d+o(d)}n$ and $2^{\Oo_q(d)}n^{1+\varepsilon}$, respectively.

An intriguing question concerns the existence of such results over infinite fields (starting with $\mathbb Q$).
We do not have a direct definition of twin-width of matrices over $\mathbb Q$ based on contraction sequences. However linear-minor freeness naturally carries to infinite fields, and thus, it is natural to consider that a class of matrices over $\mathbb Q$ has bounded twin-width if its closure under linear minors is not the set of all matrices. This can be equivalently stated via the notion of \emph{grid rank} of a matrix $M$, i.e., the largest $k$ for which there is a $k\times k$ subdivision of $M$ in which every block has rank at least $k$. Note that if a matrix has grid rank $k$, then any linear minor has grid rank at most $k$. Indeed, one can even show that a class of matrices has bounded grid rank if and only if it does not contain some matrix as a linear minor. We believe that computing the product of two matrices over $\mathbb Q$ with bounded grid rank should be done in almost quadratic time, however, we lack a structural decomposition as in the finite field case.

\medskip

There is a vast literature on computing matrix multiplication, or other natural primitives of linear algebra, on classes of structured matrices.
We give a few references on rank-structured matrices (see for instance~\cite{Eidelman99,Hackbusch99,Borm03,Vandebril05,Sa18}) and matrices of bounded treewidth.

A~square matrix has quasiseparable order~$s$ if all its submatrices that are completely above the main diagonal, or completely below it, have rank at most~$s$.
Note that on adjacency matrices this is equivalent to the \emph{linear rank-width} parameter (a dense analogue of pathwidth). 
Pernet~\cite{Pernet16} shows that multiplying two $n \times n$ matrices with quasiseparable order~$s$ can be done in time $\Oo(s^{\omega-2} n^2)$, where $\omega$ is the exponent of matrix multiplication, or $\Oo(s^3 n)$ if the matrices are given in a suitable compact form~\cite{PernetS18}.
The closely-related semiseparable matrices also have efficient multiplication algorithms~\cite{Xia10}.
So-called\footnote{These classes involve a hierarchical (dichotomic) block decomposition into blocks that are either of low rank (if far from the main diagonal) or recursively in the class (otherwise); see~\cite{Hackbusch99,Borm03} for definitions.} {$\mathscr H$-matrices} (for hierarchical matrices) and $\mathscr H^2$-matrices admit almost linear-time algorithms for vector-matrix multiplication~\cite{Borm03}. 

One can naturally extend the treewidth graph parameter to 0,1-matrices $M$ by considering the treewidth of the bipartite graph whose biadjacency matrix is $M$.
Fomin et al.~\cite{FominLSPW18} show how to compute the determinant, the rank, and to solve a linear system defined by an $n \times n$ matrix of treewidth $k$, in time $k^{\Oo(1)}n$.
It was recently shown by Dong et al.~\cite{Dong21} how to solve linear programs in expected almost linear-time on matrices of bounded treewidth.

\section{Preliminaries}\label{sec:prelim}

We denote by $[i,j]$ the set of integers $\{i,i+1,\ldots,j-1,j\}$, and $[i]$ is a short-hand for $[1,i]$.
We use the standard graph-theoretic notations: $V(G)$, $E(G)$, $N_G[S]$, $N_G(S)$, $G[S]$ respectively denote the vertex set, edge set, closed neighborhood of $S$, open neighborhood of $S$, and subgraph induced by $S$, of a graph $G$.
Given a matrix $M$, we may interchangeably denote by $M_{x,y}$ or $M[x,y]$ the entry of $M$ at row $x$ and column $y$.

\subsection{Binary structures and matrices}\label{subsec:logic}

A relational \emph{signature} $\Sigma$ is a finite set of relation symbols~$R$, each with a specified arity $r\in\mathbb N$. 
\mbox{A \emph{$\Sigma$-structure}~$\mathbf A$} is defined by a set~$A$ (the \emph{domain} of $\mathbf A$) together with a relation $R^\mathbf A\subseteq A^{r}$ for each relation symbol \mbox{$R \in \Sigma$} with arity $r$.
The syntax and semantics of first-order formulas over $\Sigma$, or \emph{$\Sigma$-formulas} for brevity, are defined as usual.
We will augment first-order logic with modular counting quantifiers (see~\cref{sec:fomc} for definitions).
A~\emph{binary structure} is a $\Sigma$-structure such that every relation symbol of $\Sigma$ has arity at most~2.
An \emph{ordered binary structure} is a structure $\str A$ over a signature $\Sigma$ consisting of unary and binary relation symbols which includes the symbol $<$, defining in $\str A$ a total order on the domain of $\str A$.
 
A matrix $M$ over a finite alphabet $A$ with rows $R$ and columns $C$ is viewed as an ordered binary structure with domain $R \uplus C$, equipped with the following relations:
 \begin{compactitem}
   \item a~unary relation $R$ interpreted as the set of rows,
   \item an antisymmetric binary relation $<$ which defines a total order on $R \uplus C$, extending the total orders on the rows and columns of $M$ in such a way that the rows precede the columns,
   \item one binary relation $E_a$, for each $a \in A$, where $E_a(r,c)$ holds if and only if $r$ is a row, $c$ is not (hence is a column), and $a$ is the entry of $M$ at row $r$ and column $c$.
 \end{compactitem}
 Matrices over a finite alphabet are actually as general as unconstrained binary structures.
 Indeed, one can encode any binary structure as an adjacency matrix $M$, where the values $M[x,y]$ are in one-to-one correspondence with the \emph{atomic type of $(x,y)$}, that is, the subset of unary relations containing $x$, the subset of unary relations containing $y$, and the subset of binary relations containing the (ordered) pair $(u,v)$.
 The reader might think of matrices over finite alphabets, binary structures, or edge-colored graphs as different representations of the same objects.
 Depending on the context, we will sometimes prefer one representation over another. 

\subsection{Contraction sequences and twin-width}\label{subsec:contraction-seq}

A \emph{trigraph $G$} has vertex set $V(G)$, (black) edge set $E(G)$, and red edge set $R(G)$, with $E(G)$ and $R(G)$ being disjoint.
The \emph{set of neighbors $N_G(v)$} of a vertex $v$ in a trigraph $G$ consists of all the vertices adjacent to $v$ by a black or red edge.
A $d$-trigraph is a trigraph $G$ such that the \emph{red graph} $(V(G),R(G))$ has degree at most~$d$.
In that case, we also say that the trigraph has \emph{red degree} at most~$d$.
A (vertex) \emph{contraction} or \emph{identification} in a trigraph~$G$ consists of merging two (non-necessarily adjacent) vertices $u$ and $v$ into a single vertex $z$, and updating the edges of $G$ in the following way.
Every vertex of the symmetric difference $N_G(u) \triangle N_G(v)$ is linked to $z$ by a red edge.
Every vertex $x$ of the intersection $N_G(u) \cap N_G(v)$ is linked to $z$ by a black edge if both $ux \in E(G)$ and $vx \in E(G)$, and by a red edge otherwise.
The rest of the edges (not incident to $u$ or $v$) remain unchanged.
We insist that the vertices $u$ and $v$ (together with the edges incident to these vertices) are removed from the trigraph. 
See \cref{fig:contraction} for an illustration.
\begin{figure}[h!]
\begin{tikzpicture}
\def\v{2}
\def\t{6}
\def\s{0.6}

\draw[thick, -stealth] (3.25,\v /2) -- (4.25,\v/2) ;

\foreach \i/\j in {-5/u_1,-4/u_2,-3/x_1,-2/x_2,-1/x_3,0/x_4,1/x_5,2/x_6,3/x_7,4/v_1,5/v_2}{
  \node[draw,circle,inner sep=0.03cm] (n\i) at (\s * \i,\v) {$\j$} ; 
}

\node[draw,circle] (u) at (-1.5,0) {u} ;
\node[draw,circle] (v) at (1.5,0) {v} ;

\foreach \i in {-4,-3,-2,0,2,3}{
  \draw (u) -- (n\i) ;
}
\foreach \i in {-5,-1,1}{
  \draw[thick, red] (u) -- (n\i) ;
}
\foreach \i in {-3,...,0,2,3,...,5}{
  \draw (v) -- (n\i) ;
}
\foreach \i in {1}{
  \draw[thick, red] (v) -- (n\i) ;
}

\begin{scope}[xshift=7.5cm]
\node[draw,circle] (uv) at (0,0) {z} ;
\foreach \i/\j in {-5/u_1,-4/u_2,-3/x_1,-2/x_2,-1/x_3,0/x_4,1/x_5,2/x_6,3/x_7,4/v_1,5/v_2}{
  \node[draw,circle,inner sep=0.03cm] (m\i) at (\s * \i,\v) {$\j$} ; 
}

\foreach \i in {-3,-2,0,2,3}{
  \draw (uv) -- (m\i) ;
}
\foreach \i in {-5,-4,-1,1,4,5}{
  \draw[thick, red] (uv) -- (m\i) ;
}
\end{scope}
\end{tikzpicture}
\caption{Contraction of vertices $u$ and $v$, and how the edges of the trigraph are updated.}
\label{fig:contraction}
\end{figure}
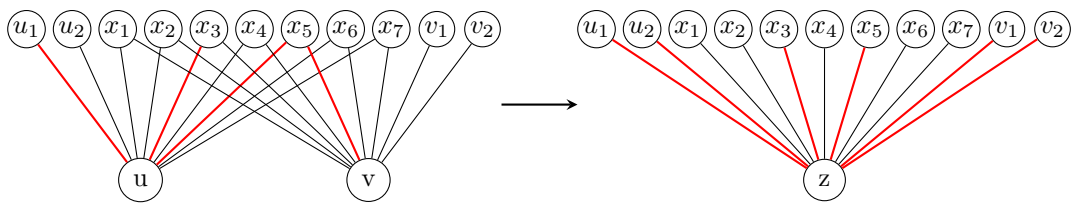

A \emph{$d$-sequence} (or \emph{contraction sequence}) is a sequence of \mbox{$d$-trigraphs} $G_n, G_{n-1}, \ldots, G_1$, where $G_n = G$, $G_1=K_1$ is the graph on a single vertex, and $G_{i-1}$ is obtained from $G_i$ by performing a single contraction of two (non-necessarily adjacent) vertices.
We observe that $G_i$ has precisely $i$ vertices, for every $i \in [n]$.
The twin-width of~$G$, denoted by $\tww(G)$, is the minimum integer~$d$ such that $G$ admits a~$d$-sequence. See \cref{fig:contraction-sequence} for an illustration.
\begin{figure}[h!]
  \centering
  \resizebox{400pt}{!}{
  \begin{tikzpicture}[
      vertex/.style={circle, draw, minimum size=0.68cm}
    ]
    \def\s{1.2}
    \foreach \i/\j/\l in {0/0/a,0/1/b,0/2/c,1/0/d,1/1/e,1/2/f,2/1/g}{
      \node[vertex] (\l) at (\i * \s,\j * \s) {$\l$} ;
    }
    \foreach \i/\j in {a/b,a/d,a/f,b/c,b/d,b/e,b/f,c/e,c/f,d/e,d/g,e/g,f/g}{
      \draw (\i) -- (\j) ;
    }

    \begin{scope}[xshift=3 * \s cm]
    \foreach \i/\j/\l in {0/0/a,0/1/b,0/2/c,1/0/d,2/1/g}{
      \node[vertex] (\l) at (\i * \s,\j * \s) {$\l$} ;
    }
    \foreach \i/\j/\l in {1/1/e,1/2/f}{
      \node[vertex,opacity=0.2] (\l) at (\i * \s,\j * \s) {$\l$} ;
    }
    \node[draw,rounded corners,inner sep=0.01cm,fit=(e) (f)] (ef) {ef} ;
    \foreach \i/\j in {a/b,a/d,b/c,b/d,b/ef,c/ef,c/ef,d/g,ef/g,ef/g}{
      \draw (\i) -- (\j) ;
    }
    \foreach \i/\j in {a/ef,d/ef}{
      \draw[red, very thick] (\i) -- (\j) ;
    }
    \end{scope}

    \begin{scope}[xshift=6 * \s cm]
    \foreach \i/\j/\l in {0/1/b,0/2/c,2/1/g,1/1/ef}{
      \node[vertex] (\l) at (\i * \s,\j * \s) {$\l$} ;
    }
    \foreach \i/\j/\l in {0/0/a,1/0/d}{
      \node[vertex,opacity=0.2] (\l) at (\i * \s,\j * \s) {$\l$} ;
    }
    \draw[opacity=0.2] (a) -- (d) ;
    \node[draw,rounded corners,inner sep=0.01cm,fit=(a) (d)] (ad) {ad} ;
    \foreach \i/\j in {ad/b,b/c,b/ad,b/ef,c/ef,c/ef,ef/g,ef/g}{
      \draw (\i) -- (\j) ;
    }
    \foreach \i/\j in {ad/ef,ad/g}{
      \draw[red, very thick] (\i) -- (\j) ;
    }
    \end{scope}

    \begin{scope}[xshift=9 * \s cm]
    \foreach \i/\j/\l in {0/2/c,2/1/g,0.5/0/ad}{
      \node[vertex] (\l) at (\i * \s,\j * \s) {$\l$} ;
    }
    \foreach \i/\j/\l in {0/1/b,1/1/ef}{
      \node[vertex,opacity=0.2] (\l) at (\i * \s,\j * \s) {$\l$} ;
    }
    \draw[opacity=0.2] (b) -- (ef) ;
    \node[draw,rounded corners,inner sep=0.01cm,fit=(b) (ef)] (bef) {bef} ;
    \foreach \i/\j in {ad/bef,bef/c,bef/ad,c/bef,c/bef,bef/g}{
      \draw (\i) -- (\j) ;
    }
    \foreach \i/\j in {ad/bef,ad/g,bef/g}{
      \draw[red, very thick] (\i) -- (\j) ;
    }
    \end{scope}

    \begin{scope}[xshift=11.7 * \s cm]
    \foreach \i/\j/\l in {0/2/c}{
      \node[vertex] (\l) at (\i * \s,\j * \s) {$\l$} ;
    }
     \foreach \i/\j/\l in {0.5/0/adg,0.5/1.1/bef}{
      \node[vertex] (\l) at (\i * \s,\j * \s) {\footnotesize{\l}} ;
    }
    \foreach \i/\j in {c/bef}{
      \draw (\i) -- (\j) ;
    }
    \foreach \i/\j in {adg/bef}{
      \draw[red, very thick] (\i) -- (\j) ;
    }
    \end{scope}

    \begin{scope}[xshift=13.7 * \s cm]
    \foreach \i/\j/\l in {0.5/0/adg,0.5/1.1/bcef}{
      \node[vertex] (\l) at (\i * \s,\j * \s) {\footnotesize{\l}} ;
    }
    \foreach \i/\j in {adg/bcef}{
      \draw[red, very thick] (\i) -- (\j) ;
    }
    \end{scope}

    \begin{scope}[xshift=15 * \s cm]
    \foreach \i/\j/\l in {1/0.75/abcdefg}{
      \node[vertex] (\l) at (\i * \s,\j * \s) {\tiny{\l}} ;
    }
    \end{scope}
    
  \end{tikzpicture}
  }
  \caption{A 2-sequence witnessing that the initial graph has twin-width at most~2.}
  \label{fig:contraction-sequence}
\end{figure}
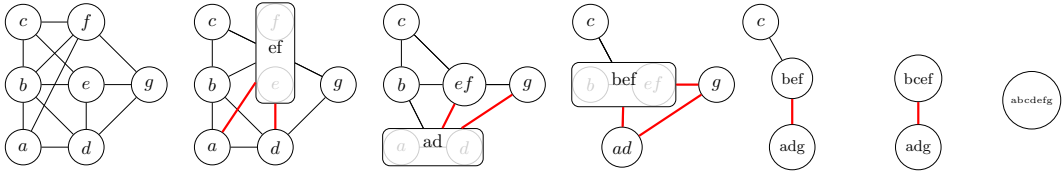

Twin-width can be generalized from graphs to binary structures in some (functionally) equivalent ways~\cite{twin-width1,twin-width2}.
Here we choose the following definition.

On general binary structures, red edges exist between two vertices $x,y \in V(G_i)$ whenever there are up to four vertices $u \neq v,u' \neq v' \in V(G)$ such that $u$ and $u'$ (which might be the same vertex) were contracted (together with possibly other vertices) to form $x$, similarly $v$ and $v'$ were contracted to form $y$, and the atomic types of $(u,v)$ and of $(u',v')$, or of $(v,u)$ and of $(v',u')$, are distinct.
If instead, all such pairs $(u,v)$ have the same atomic type, this shared atomic type labels the edge $xy$ in $G_i$.
Contraction sequences, $d$-sequences, and twin-width are then similarly defined.

In particular, we now have a definition of twin-width for the matrices of~\cref{subsec:logic}.

\subsection{Matrix divisions}

We will often denote by $R$ (resp.~$C$) the sets of row (resp.~column) indices of a matrix $M$.
For $X \subseteq R$ and $Y \subseteq C$, we denote by $M[X,Y]$ the submatrix of $M$ consisting of the entries at rows in $X$ and columns in $Y$.

A~\emph{$(k,\ell)$-division} of a matrix $M$ is a pair of partitions $(\mathcal R=\{R_1, \ldots, R_k\}, \mathcal C=\{C_1, \ldots, C_\ell\})$ of $R$ and $C$, respectively, such that every $R_i$ corresponds to consecutive rows, and every $C_j$, to consecutive columns.
A~\emph{$k$-division} is a $(k,k)$-division.
We may call a set of consecutive row or column indices an~\emph{interval}.
The \emph{grid number}, \emph{mixed number}, \emph{grid rank}, respectively, of a matrix is the largest integer $k$ such that $M$ has a $k$-division $(\{R_1, \ldots, R_k\}, \{C_1, \ldots, C_k\})$ for which, for every $i,j \in [k]$, 
\begin{compactitem}
\item $M[R_i,C_j]$ has a non-zero entry, 
\item $M[R_i,C_j]$ has at least two distinct rows and at least two distinct columns, 
\item $M[R_i,C_j]$ has at least $k$ distinct rows or at least $k$ distinct columns, 
\end{compactitem}
respectively.

The divisions are called \emph{$k$-grid minor}, \emph{$k$-mixed minor}, and \emph{rank-$k$ division}, respectively.
An \emph{interval minor} of a matrix $M$ is a matrix $N$ with $k$ rows and $\ell$ columns such that $M$ has a $(k,\ell)$-division $(\{R_1, \ldots, R_k\}, \{C_1, \ldots, C_\ell\})$ such that for every $i \in [k]$ and every $j \in [\ell]$, $N_{i,j}$ is an entry of $M[R_i,C_j]$.

We call \textsc{$k$-Grid Minor}, \textsc{$k$-Mixed Minor}, \textsc{Rank-$k$ Division}, \textsc{Interval Minor Containment}, respectively, the computational problems of deciding if an input matrix has grid number at least $k$, mixed number at least $k$, grid rank at least $k$, and if a matrix is an interval minor of another matrix.

\medskip

The following celebrated theorem by Marcus and Tardos, solving the Füredi-Hajnal, and hence~\cite{Klazar00} the Stanley-Wilf conjectures, is crucial to the theory of twin-width.  

\begin{theorem}[\cite{MarcusT04}]\label{thm:marcus-tardos}
  There is a function $\mt \from \mathbb N \to \mathbb N$ such that for every natural $k$, every $n \times m$ matrix with at least $\mt(k) \max(n,m)$ non-zero entries admits a $k$-grid minor. 
\end{theorem}
The current best known bound is $\mt(k) = \frac{8}{3} (k+1)^2 2^{4k}=2^{\Oo(k)}$~\cite{Cibulka16}.

\subsection{Computing contraction sequences}\label{sec:comput_seq}

Efficiently (in polynomial time, \FPT~time, or even slice-wise polynomial time) approximating the twin-width of a binary structure remains a challenging open question.
However, such an algorithm is known for totally ordered binary structures, or matrices over a finite alphabet.

\begin{theorem}[\cite{twin-width4}]\label{thm:approx-tww-unspec}
  Given an $n \times m$ matrix $M$ over a finite alphabet $A$, and an integer $d$, there is an $\exp(\exp(\Oo_{|A|}(d^2 \log d)))(n+m)^3$-time algorithm that 
  \begin{compactitem}
  \item either correctly reports that the twin-width of $M$ is larger than $d$,
  \item or outputs a $|A|^{\Oo(d^4)}$-sequence for $M$.
  \end{compactitem}
\end{theorem}

The exponent of $(n+m)$ in the algorithm running time was not made explicit in \cite{twin-width4}, but it can be observed that a direct implementation would indeed take cubic time.
We will now see that we can bring this down to an almost quadratic dependence in the number of rows and columns.
For the sake of convenience, we will work with symmetric matrices.
For most applications, this is not really a limitation.
To compute something on a non-symmetric matrix $M$, we will instead consider the symmetric $s(M):=\left( \begin{smallmatrix} 0&M \\ M^T&0 \end{smallmatrix} \right)$.
Our main application concerns matrix multiplication.
To compute the product $MN$ of two square\footnote{If not, one can pad them with rows or columns of 0 entries.} matrices, we can compute $$s(M) s(N^T) = \left( \begin{smallmatrix} 0&M \\ M^T&0 \end{smallmatrix} \right) \left( \begin{smallmatrix} 0&N^T \\ N&0 \end{smallmatrix} \right) = \left( \begin{smallmatrix} MN & 0 \\ 0 & M^TN^T \end{smallmatrix} \right).$$
That the resulting matrix is not necessarily symmetric is irrelevant.
We read off $MN$ from its top-left block.

We need the following definitions.

A~$k$-division of a symmetric matrix is \emph{symmetric} if for every $i \in [k]$, the $i$-th row part and the $i$-th column part contain the same row and column indices.
Let $M$ be a matrix and $\mathcal D=(\mathcal R, \mathcal C)$ be a $(k,\ell)$-division of $M$ (with row set $R$, and column set $C$).
A~row part $R_i \in \mathcal R$ is \emph{$r$-poor} in $\mathcal D$ if there is a union $X$ of at most $r$ column parts of $\mathcal C$ such that $M[R_i,C \setminus X]$ contains less than $r$ distinct rows.
An~$r$-poor column part is defined symmetrically.
In~\cite{twin-width4}, a~division~$\mathcal D$ is said \emph{$r$-rich} if none of its parts are $r$-poor.

A~\emph{symmetric fusion} of a matrix $M$ with a symmetric $k$-division $(\mathcal R, \mathcal C)$ consists of merging two consecutive row parts of $\mathcal R$, and the two corresponding column parts in $\mathcal C$.
It results in a symmetric $(k-1)$-division of $M$.
The notions of \emph{$(k,\ell)$-partition}, \emph{$k$-partition}, and \emph{symmetric merge} are obtained from \emph{$(k,\ell)$-division}, \emph{$k$-division}, and \emph{symmetric fusion}, respectively, by relaxing the condition that parts should consist of consecutive rows or columns.

As a preparatory step, we give an explicit lower bound for the grid rank to imply twin-width larger than~$d$. 
\begin{lemma}[\cite{twin-width4}]\label{lem:hgr-htww}
  Let $d$ be a non-negative integer, and $k := 2d(d+1)+1$.
  If $M$ admits a~rank-$k$ division, then $M$ has twin-width larger than~$d$.
\end{lemma}
\begin{proof}
  One can observe that a~rank-$k$ division is a $(k-1)$-rich division.
  Hence we assume that $M$ has a $2d(d+1)$-rich division.
  The fact that $M$ has then twin-width larger than~$d$ is the content of~\cite[Lemma 19]{twin-width4-arxiv}.
\end{proof}
    
Still following the previous paper of the series~\cite{twin-width4}, we show the following.

\begin{theorem}\label{thm:approx-tww-quad}
  Given an $n \times n$ symmetric matrix $M$ over a finite alphabet $A$, and an integer~$d$, there is an algorithm running in time $\Oo_{d,|A|}(n^2 \log n)$ that 
  \begin{compactitem}
  \item either correctly reports that the twin-width of $M$ is larger than $d$,
  \item or outputs an~$\exp(\exp(\exp(\Oo_{|A|}(d^4))))$-sequence for $M$.
  \end{compactitem}
\end{theorem}

\begin{proof}
Let $M$ be the input, as described, and set $k := 2d(d+1)+1$, and $r := |A|^{k \cdot \mt(k^2)}$.
Let $\mathcal D=(\mathcal R, \mathcal C)$ be the unique (symmetric) $n$-division of $M$, and $\mathcal P=(\mathcal R^{\mathcal P}, \mathcal C^{\mathcal P})$ be the unique (symmetric) $n$-partition of $M$.
Thus, initially, $\mathcal D = \mathcal P$.
Trivially, every part of $\mathcal D$ is $r$-poor in~$\mathcal D$.

We need the following, shown in the previous paper of the series.
\begin{claim}[{\cite[Theorem 2, Lemma 20]{twin-width4-arxiv}}]\label{lem:check-poor}
  Checking if a row (resp. column) part $R$ of size $h$ is $r$-poor in $\mathcal D$ can be done in time $\exp(\exp_{|A|}(r \log r)) \cdot h n$.

  If $R$ is indeed $r$-poor, within the same time, a set $Y$ of at most $r$ column (resp. row) parts can be found, and we output a partition of $R$ into the at most $r-1$ parts corresponding to the equivalent classes of the relation \emph{being equal rows (resp. columns) outside $Y$}.
  
  For convenience, if $R$ is not $r$-poor, we output $\{R\}$.
\end{claim}
\begin{proofofclaim}
  This is direct from~\cite[Lemma 20]{twin-width4-arxiv}.
  The running time is indeed linear in the number of entries spanned by~$R$, that is $h n$. 
\end{proofofclaim}

\textbf{Algorithm.}
While $\mathcal D$ has more than one row part, we do the following.

For every pair of consecutive row parts $R_i, R_{i+1}$ in $\mathcal R$ with $i$ odd, we check if $R_i \cup R_{i+1}$ is $r$-poor in $(\mathcal R \setminus \{R_i,R_{i+1}\} \cup \{R_i \cup R_{i+1}\}, \mathcal C)$.\\
$(A)$ If at least half of the pairs can be merged into an $r$-poor part, we make the corresponding symmetric fusions in $\mathcal D$.
We obtain a new symmetric $s$-division $\mathcal D'=(\mathcal R', \mathcal C')$ where indices are rearranged from 1 to $s$, and a new symmetric partition $\mathcal D'^{\mathcal P}=(\mathcal R'^{\mathcal P}, \mathcal C'^{\mathcal P})$, which is the union of the partitions output by~\cref{lem:check-poor}.\\
$(B)$ If, instead, less than half of the pairs can be merged, we will prove (slightly adapting \cite[Section 5]{twin-width4-arxiv}) that the grid rank of $M$ is large, which in turns implies that its twin-width is more than~$d$, by~\cref{lem:hgr-htww}.
We thus exit the while loop.

\medskip

\textbf{Running time.}
Case $(A)$ reduces by a factor of at least $1/4$ the current number of row parts in $\mathcal D$.
Hence it can only be done $\Oo(\log n)$ times.
By~\cref{lem:check-poor}, each such iteration takes time $$\exp(\exp(r \log r)) \cdot \left(\sum\limits_{i~\text{odd natural}} |R_i \cup R_{i+1}| \cdot n\right) = \Oo_d(n^2).$$
In this upper bound, $i$ does not exceed the size of the current division, and the equality holds since the submatrices we consider partition $M$.
The total running time of $\exp(\exp(\exp(\exp_{|A|}(d^4 \log d)))) \cdot n^2 \log n=\Oo_d(n^2 \log n)$ follows since case $(B)$ stops the algorithm.

\medskip

\textbf{Correctness.}
It was shown in~\cite{twin-width4} that if $(B)$ never occurs, the sequence of computed partitions $\mathcal P$ are coarser and coarser, and essentially form an~$O_{d,|A|}(1)$-sequence: one can arbitrarily merge parts to bridge one partition to the next.
Following the end of the proof of~\cite[Theorem 2, Case 2.]{twin-width4-arxiv}, we obtain an~$|A|^{\Oo(r^2)}$-sequence, hence an~$\exp(\exp(\exp(\Oo_{|A|}(d^4))))$-sequence, for~$M$.
We thus focus on the case when $(B)$ happens.

What was playing the role of case $(B)$ in~\cite{twin-width4} (Case 1.) was more constrained: It stated that no two consecutive parts can be fused into an $r$-poor part.
Then, one could directly show that, if this holds, the twin-width of $M$ has to be large.
Now, we only know that a positive fraction of the consecutive pairs cannot be fused into an $r$-poor part.
We show that this still implies a large grid rank (following \cite[Section 5]{twin-width4-arxiv}).

Let $\mathcal D^B=(\mathcal R^B=\{R^B_1,R^B_2,\ldots,R^B_{2s-1},R^B_{2s}\}, \mathcal C^B=\{C^B_1,C^B_2,\ldots,C^B_{2s-1},C^B_{2s}\})$ be the current division when $(B)$ occurs.
Note that, for convenience, we assumed that the number of row parts in $\mathcal D^B$ is even.
We consider the symmetric division $\mathcal D=(\mathcal R,\mathcal C)$ with $$\mathcal R=\{R_1=R^B_1 \cup R^B_2, R_2=R^B_3 \cup R^B_4, \ldots, R_s = R^B_{2s-1},R^B_{2s}\},$$
$$\mathcal C=\{C_1=C^B_1 \cup C^B_2, C_2=C^B_3 \cup C^B_4, \ldots, C_s = C^B_{2s-1},C^B_{2s}\}.$$
At least half of the $R_i$s (and symmetrically, of the $C_i$s) are \emph{not} $r$-poor.

Following~\cite[Theorem 22]{twin-width4}, in each column part of $\mathcal C$, we color in red the cells of $\mathcal D$ with at least $k$ distinct rows or $k$ distinct columns.
We then color blue the cells $X$ that feature a row vector which is not in a~non-red cell, in the same column part of $X$ and below it.
In the at least $s/2$ column parts that are not $r$-poor, one can show that at least $\mt(k^2)$ cells are colored~\cite[Theorem 22]{twin-width4}.
We symmetrically fuse consecutive blocks of $r$-poor parts with an adjacent part which is not $r$-poor.
Observe that the cells that are containing a formerly blue cell are now blue or red, and the cells containing a formerly red cell are still red.
We now obtain a new symmetric division with at least $s/2$ column parts (or equivalently, row parts), where every column part contains at least $\mt(k^2)$ colored cells.
Thus, we can finish the proof as in~\cite[Theorem 22]{twin-width4}, and obtain a rank-$k$ division of $M$.
Hence, by~\cref{lem:hgr-htww}, the twin-width of $M$ is larger than~$d$. 
\end{proof}

\subsection{Twin-decompositions}\label{sec:twin_dec}

A twin-decomposition of a graph~$G$ also uses the framework of a rooted carving decomposition, i.e., a rooted binary tree whose leaves are in one-to-one correspondence with the vertices of~$G$.
In the case of twin-decompositions though, the internal nodes of the rooted binary tree are totally ordered and the \emph{width} is quite different from how carving-width is defined.

A rooted binary tree with $n-1$ internal nodes bijectively labeled by $\ell$ on $[n-1]$ is said~\emph{ranked} if whenever $u$ and $v$ are two distinct internal nodes such that $v$ is a descendant of $u$, then $\ell(u) < \ell(v)$ holds.
For every $i \in [n]$, the \emph{$i$-th border} of a ranked tree $\mathcal T$ with $n-1$ internal nodes is the set of maximal subtrees whose roots have label at least $i$, where the leaves all have labels $+\infty$ (or equivalently $n$).
We denote by $B_i(\mathcal T)$ the $i$-th border of $\mathcal T$.
Note that an $i$-th border consists of exactly $i$ subtrees.
We denote by $r(\mathcal T)$ the root of $\mathcal T$.

A~\emph{twin-decomposition} of an $n$-vertex graph $G$ is a pair $(\mathcal T, \mathcal B)$ where
\begin{enumerate}[(a)]
\item $\mathcal T$ is a rooted binary \textbf{t}ree, ranked by $\ell$ on $[n-1]$, whose leaves are in one-to-one correspondence with $V(G)$, and 
\item $\mathcal B$ (for \textbf{b}icliques) is a set of edges over $V(\mathcal T)$, such that:
\end{enumerate}
\begin{enumerate}
\item $\mathcal B$ partitions the edge set of $G$, where an edge between $u, v \in V(T)$ is interpreted as the biclique of $G$ linking every leaf in the subtree of $\mathcal T$ rooted at $u$, to every leaf in the subtree of $\mathcal T$ rooted at $v$, and 
\item No edge of $\mathcal B$ \emph{crosses} an $i$-th border, i.e., links a node in a subtree $\mathcal T'$ of $B_i(\mathcal T)$ but \emph{not} the root of $\mathcal T'$ to a vertex outside every subtree of $B_i(\mathcal T)$.
\end{enumerate}

One can retrieve the sequence of trigraphs $G=G_n, \ldots, G_1$ from the twin-decomposition in the following way.
The vertex set of $G_i$ corresponds to the subtrees of $B_i(\mathcal T)$,
\begin{compactitem}
\item with a black edge between $\mathcal T_1 \in B_i(\mathcal T)$ and $\mathcal T_2 \in B_i(\mathcal T)$ whenever there is an edge in $\mathcal B$ between $r(\mathcal T_1)$ or one of its ancestors (in $\mathcal T$) and $r(\mathcal T_2)$ or one of its ancestors, or whenever for every leaf $u$ of $\mathcal T_1$ and every leaf $v$ of $\mathcal T_2$, there is an edge between the path from $u$ to $r(\mathcal T_1)$ and the path from $v$ to $r(\mathcal T_2)$,
\item and a red edge between $\mathcal T_1 \in B_i(\mathcal T)$ and $\mathcal T_2 \in B_i(\mathcal T)$ whenever this does not hold but yet there is an edge $uv \in \mathcal B$ with $u \in V(\mathcal T_1)$ and $v \in V(\mathcal T_2)$, and a~pair $(u',v')$ such that $u'$ is a leaf of $\mathcal T_1$, $v'$ is a leaf of $\mathcal T_2$, and there is no edge between the path from $u'$ to $r(\mathcal T_1)$ and the path from $v'$ to $r(\mathcal T_2)$.
\end{compactitem}

The width of the twin-decomposition $(\mathcal T, \mathcal B)$ is again defined as the maximum red degree among every vertex of every trigraph $G_i$ (as previously defined).
See~\cref{fig:twin-decomposition} for an illustration of a twin-decomposition corresponding to a particular contraction sequence.

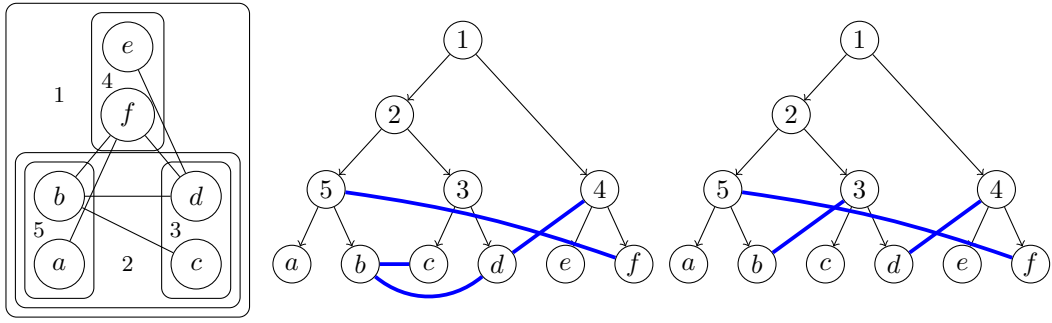
\begin{figure}[h!]
  \centering
  \begin{tikzpicture}[scale=0.9]
    \foreach \i/\j/\l in {0/0/a, 0/1/b, 2/0/c, 2/1/d, 1/3.2/e, 1/2.2/f}{
      \node[draw,circle,minimum size=0.66cm] (\l) at (\i,\j) {$\l$} ;
    }
    \foreach \i/\j in {b/c,b/d,d/e,d/f,f/a,f/b}{
      \draw (\i) -- (\j) ;
    }
    \foreach \i/\j/\l in {a/b/ab, c/d/cd, e/f/ef, ab/cd/abcd, abcd/ef/all}{
      \node[draw,rounded corners,fit=(\i) (\j)] (\l) {} ;
    }
    \foreach \i/\j/\s in {1/0/2,0/2.5/1,-0.3/0.5/5,1.7/0.5/3,0.7/2.7/4}{
      \node at (\i,\j) {\footnotesize{\s}} ;
    }

    \begin{scope}[xshift=2.4cm]
      \foreach \i/\l in {1/a,2/b,3/c,4/d,5/e,6/f}{
        \node[draw,circle,inner sep=0.04cm,minimum size=0.5cm] (v\l) at (\i,0) {$\l$} ;
      }
      \foreach \i/\j/\l in {1.5/1.1/5,3.5/1.1/3,5.5/1.1/4, 2.5/2.2/2, 3.5/3.3/1}{
        \node[draw,circle,inner sep=0.04cm,minimum size=0.5cm] (v\l) at (\i,\j) {$\l$} ;
      }
    \foreach \i/\j in {a/5,b/5, c/3,d/3, e/4,f/4, 5/2,3/2, 2/1, 4/1}{
      \draw[<-] (v\i) -- (v\j) ;
    }
    \foreach \i/\j/\b in {b/c/0,b/d/-40,d/4/0,f/5/-5.5}{
      \draw[line width=0.05cm,blue] (v\i) to [bend left=\b] (v\j) ;
    }
    \end{scope}

    \begin{scope}[xshift=8.2cm]
      \foreach \i/\l in {1/a,2/b,3/c,4/d,5/e,6/f}{
        \node[draw,circle,inner sep=0.04cm,minimum size=0.5cm] (v\l) at (\i,0) {$\l$} ;
      }
      \foreach \i/\j/\l in {1.5/1.1/5,3.5/1.1/3,5.5/1.1/4, 2.5/2.2/2, 3.5/3.3/1}{
        \node[draw,circle,inner sep=0.04cm,minimum size=0.5cm] (v\l) at (\i,\j) {$\l$} ;
      }
    \foreach \i/\j in {a/5,b/5, c/3,d/3, e/4,f/4, 5/2,3/2, 2/1, 4/1}{
      \draw[<-] (v\i) -- (v\j) ;
    }
    \foreach \i/\j/\b in {b/3/0,d/4/0,f/5/-5.5}{
      \draw[line width=0.05cm,blue] (v\i) to [bend left=\b] (v\j) ;
    }
    \end{scope}
\end{tikzpicture}
  \caption{Left: a graph $G$ with a contraction sequence (or partition sequence), where trigraph $G_i$ is obtained after performing the contraction labeled $i$. Center: the twin-decomposition corresponding to this contraction sequence, with the edges of $\mathcal B$ in blue. Right: a ranked tree $\mathcal T$ and a partition $\mathcal B$ of the edges of $G$ that does \emph{not} make for a twin-decomposition, since the edge $b3$ crosses $B_5(\mathcal T)$ (and $B_4(\mathcal T)$).}
\label{fig:twin-decomposition}
\end{figure}

Suppose $(\mathcal T,\mathcal B)$ is a twin-decomposition of a graph with width at most~$d$.
If one orients every edge $uv \in \mathcal B$ as the arc $(u,v)$ whenever the parent of $u$ has a larger label than the parent of $v$ (keeping the edge undirected if $u$ and $v$ shares the same parent), then every node $z$ of $V(\mathcal T)$ has at most $d$ out-neighbors.
Indeed at the $i$-th border, with $i$ being the label of the parent $p(z)$ of $z$, the subtree of $B_i(\mathcal T)$ rooted at $p(z)$ has red degree at most one per out-neighbor of~$z$.
This is because two neighbors of $z$ (in $\mathcal B$) cannot be in a descendant-ancestor relationship, as we imposed $\mathcal B$ to \emph{partition} the edge set of~$G$.   
Since~$z$ has at most one undirected edge incident to it, the graph made by $\mathcal B$ on $V(\mathcal T)$ is $d+1$-degenerate.

Thus the mere definition of twin-decomposition imposes that $|\mathcal B|=\Oo(n)$.

For every twin-decomposition $(\mathcal T, \mathcal B)$, we say that $\mathcal B$ is \emph{lifted up} if for every $uv\in \mathcal B$ such that $\ell(p(u))\geq \ell(p(v))$, either $p(v)=p(u)$ or if $u'\in V(\mathcal T)\backslash \sg{u}$ is such that $p(u')=p(u)$, then $u'v\notin \mathcal B$. For example, the edge set of the twin-decomposition in the center of \cref{fig:twin-decomposition} is lifted up.

\begin{remark}
\label{rem: high}
For every twin-decomposition $(\mathcal T, \mathcal B)$ of width $d$, one can compute in time $\Oo(dn)$ a set $\mathcal{B'}\subseteq E(\mathcal{T})$ of size at most $|\bb|$ such that $(\mathcal T, \mathcal{B'})$ is a twin-decomposition of the same sequence where $\mathcal{B'}$ is lifted up.
Indeed, for each $i$ ranging from $n-1$ to $1$, if $z=\ell^{-1}(i)$ and $u,u'\in V(\mathcal T)$ are such that $p(u)=p(u')=z$, for each $v \in V(\mathcal T)$ such that $uv,u'v\in \mathcal B$, we replace $\mathcal B$ by $(\mathcal B \setminus \sg{uv,u'v}) \cup \sg{zv}$.
\end{remark}

\begin{remark}
  If $(\mathcal T, \mathcal B)$ is a twin-decomposition where $\mathcal B$ is lifted up, then the edges of $\bb$ are exactly the black edges ``disappearing'' at some point of the contraction sequence associated to $\tree$.
  More precisely, $e=uv \in \bb$ if and only if there exists an $i$ such that $u$ and $v$ are distinct vertices of $V(G_i)$, $uv\in E(G_i)$ and either $u$ and $v$ are identified in $G_{i-1}$ or $u'v'\in R(G_{i-1})$, where $u'$ and $v'$ denote the vertices of $G_{i-1}$ that contain respectively $u$ and $v$.
\end{remark}

As previously observed, twin-decompositions permit to retrieve their corresponding contraction sequences, i.e., adjacencies and red adjacencies in the $G_i$s.
However, given a twin-decomposition $(\mathcal T, \mathcal B)$ of $G$ such that $|\mathcal B|=\Oo(n)$, we should not store all the $R(G_i)$ for each $i\in [n-1]$ as it may take $\Omega(n^2)$ space; too much for our algorithmic purposes when we want a linear complexity in $n$ (see \cref{sec:product}).
To do so, we will need to compute $R(G_i)$ on the fly.

Assume that $G_n, \ldots, G_1$ is a $d$-sequence of $G$, with corresponding twin-decomposition $(\mathcal T, \mathcal B)$.
Given $(\mathcal T, \mathcal B)$, we show how to compute all the red neighborhoods $\{W~\in V(G_i):~UW \in R(G_i)\}$ for every $i \in [n-1], U \in V(G_i)$ in $\Oo(dn)$ time.
Note that we may identify the vertices of $V(G_i)$ with their corresponding subsets of vertices of $V(G)$ (and thus denote them with uppercase symbols).
For each $i\in [n-1], U\in V(G_i)$, we let $L_U^i:=\sg{W_1, \ldots, W_k}$ be the $k\leq d$ red neighbors of $U$ in $G_i$.
We say that the $L_U^i$ can be \emph{dynamically computed} if there exists an algorithm which runs in $n$ steps, such that at the end of step $i$ (where $i$ ranges from $n-1$ to $1$), it only stores in memory the values of $L_U^i$ for each $U\in V(G_i)$, and every step takes time $\Oo_d(1)$.

\begin{lemma}
 \label{lem: listeadj}
 Given a twin-decomposition $(\mathcal T, \mathcal B)$ of width $d$ of an $n$-vertex graph $G$, one can dynamically compute in time $\Oo(dn)$ the adjacency lists $L^i_U$ of each node of $V(\mathcal T)$.
\end{lemma}
\begin{proof}
By \cref{rem: high} we may assume that $\mathcal B$ is lifted up.
We simply need to process the nodes of $\mathcal T$ by decreasing values of their labels and update in constant time the lists $L^i_U$ according to the information given by $\mathcal B$.
The formal algorithm is given by \cref{alg:listeadj}.
Since the width of the twin-decomposition is at most~$d$, at each step, the number of considered edges of $\mathcal B$ is at most $d$, which gives the desired complexity.
In the following algorithm, the set $\mathcal V$ at step $i$ correspond to $V(G_i)$ and the list $L_U^i$ is given by the content of $L_U$ at the end of step $i$.
   \begin{algorithm}[ht]
  \DontPrintSemicolon
  \SetKwInOut{Input}{Input}
  \SetKwInOut{Output}{Output}
  \Input{~~A twin-decomposition $(\mathcal T, \mathcal B)$ of width $d$.}
  $\mathcal V \leftarrow \bigcup_{u \in V(G)} \{\{u\}\}$\;
  \For{$U\in \mathcal V$}{
    $L_U\leftarrow \sg{}$\;
  }
 \For{$i = n-1 \rightarrow 1$}{
   Let $Z$ be the node of $\mathcal T$ labeled by $i$ and $U,V$ denote its two children in $\mathcal T$. \;
   $\mathcal V \leftarrow \mathcal V\backslash \sg{U,V}\cup \sg{Z}$\;
   $L_Z\leftarrow L_U \cup L_V \backslash \sg{U,V}$\;
   \For{$W\in \mathcal V \backslash \sg{Z}$ such that $VW\in \mathcal B$ xor $UW\in \mathcal B$}{
        $L_Z\leftarrow L_Z\cup \sg{W}$\;
   }
 }
 \caption{On-the-fly computation of the red graphs based on a twin-decomposition}
 \label{alg:listeadj}
 \end{algorithm}
\end{proof}

For the sake of simplicity, we described twin-decompositions for graphs but they readily generalize to binary structures.
Here, the representation of binary structures with edge-colored graphs will be the most convenient. 

We thus see a binary structure as a graph $G$ with an edge-labeling function $\nu:E(G)\to \mathbb F_q$.
A~twin-decomposition of $G$ is still a pair $(\mathcal T, \mathcal B)$, only now, edges of $\mathcal B$ are labeled over $\mathbb F_q$.
For every $\ell \in \mathbb F_q$, let $G_\ell$ denote the subgraph of $G$ induced by the set of edges labeled $\ell$, and by $\mathcal B_\ell$, the subset of $\bb$ of edges labeled $\ell$.
Then $(\mathcal T, \bb_\ell)$ is a twin-decomposition of $G_\ell$.
We say that $\mathcal B$ is lifted up if $\bb_\ell$ is lifted up for every $\ell \in \mathbb F_q$.

\subsection{Computing a twin-decomposition from a contraction sequence}\label{subsec:seq-to-tww-dec}

There is an easy linear (in the input size, i.e., in the number of edges) algorithm that, given a $d$-sequence, computes a corresponding twin-decomposition (of same width~$d$).
The \emph{list of triples} of a contraction sequence $G_n, \ldots, G_i, \ldots, G_1$ is $(u_n,v_n,z_n), \ldots, (u_i,v_i,z_i), \ldots, (u_2,v_2,z_2)$ such that the contraction of $u_i$ and $v_i$ in $G_i$ into a new vertex $z_i$ results in $G_{i-1}$.  

\begin{theorem}\label{thm:to-tww-dec}
  A twin-decomposition of width $d$ of an $n$-vertex graph $G$ given with the list of triples of a $d$-sequence can be computed in time $\Oo(n^2)$. 
\end{theorem}
\begin{proof}
  Let $(u_n,v_n,z_n), \ldots, (u_2,v_2,z_2)$ be the list of triples of the given $d$-sequence of~$G$.
  We want to build a twin-decomposition $(\mathcal T,\mathcal B)$ of width $d$ for $G$.
  We initialize $\mathcal T$ by creating one leaf for each vertex of $G$, each leaf being labeled by the corresponding vertex.
  We initialize $\mathcal B$ as the empty set.
  
  We process the list of triples from $(u_n,v_n,z_n)$ to $(u_2,v_2,z_2)$.
  At the $i$-th iteration (for some $i \in [2,n]$) we explicitly perform the contraction of $u_i$ and $v_i$ in $G_i$ (into $z_i$), and obtain the new trigraph $G_{i-1}$.
  In $\mathcal T$, we add a common parent $z_i$ to $u_i$ and $v_i$, and set $\ell(z_i)=i-1$.
  For every $w \in V(G_{i-1}) \setminus \{z_i\}$ such that $\{u_iw,v_iw\} \cap E(G_i) \neq \emptyset$ and $z_iw \in R(G_{i-1})$, we add the edge $u_iw$ to $\mathcal B$ if $u_iw \in E(G_i)$, or we add the edge $v_iw$ to $\mathcal B$ if $v_iw \in E(G_i)$.
  One can observe that the $i$-th step takes $\Oo(n)$ time, hence the overall running time is quadratic.
\end{proof}

Actually, on an $n$-vertex $m$-edge graph the twin-decomposition of a $d$-sequence of~$G$ can be computed in time $\Oo(dn+m)$. 

\begin{observation}\label{obs:twin-dec-to-seq}
  There is an $\Oo(n)$-time algorithm that inputs a twin-decomposition $(\mathcal T, \mathcal B)$ of width $d$ of an $n$-vertex graph $G$, and outputs the list of triples of a~$d$-sequence of $G$.
\end{observation}
\begin{proof}
  We add identifying labels, say \emph{tags}, to the nodes of $\mathcal T$ (in addition to the labeling function $\ell$).
  The tags of the leaves match the one-to-one correspondence between $V(G)$ and the leaves of $V(T)$.
  For every $i$ from $n$ down to 2, append the triple $(u_i,v_i,z_i)$ where $u_i$ and $v_i$ are the tags of the two children of the node $v = \ell^{-1}(i-1)$, and tag $v$ with the fresh identifier $z_i$. 
\end{proof}

A~consequence of the second paper of the series is that, given twin-decompositions of bounded width, one can find twin-decompositions of (larger) bounded width, where in addition the tree $\mathcal T$ has logarithmic depth.  

\begin{theorem}\emph{\cite[see Lemma 23 and Proposition 22]{twin-width2}}
  For every integer $d$, there is a larger integer $D$ such that every $n$-vertex graph of twin-width $d$ admits a twin-decomposition of width at most~$D$ and depth $\Oo_d(\log n)$. 
\end{theorem}

Given a twin-decomposition $(\mathcal T, \mathcal B)$ of $G$ with width $d$ and depth $h$, the presence of an edge between two vertices of $G$ can be decided in time $\Oo(d h)$.
This yields a~linear-space representation of $G$ with edge queries in logarithmic time.
The following stronger result was shown by Pilipczuk et al.

\begin{theorem}[\cite{PilipczukSZ22}]\label{thm:compact-rep}
 Given a twin-decomposition of width~$d$ of an $n$-vertex graph $G$, and any $\varepsilon > 0$, there is a data structure of size $\Oo(d n^{1+\varepsilon})$ \emph{(}resp.~$\Oo_d(n)$\emph{)}, computable in time $\Oo(d n^{1+\varepsilon})$ \emph{(}resp.~$\Oo_d(n \log n \log \log n)$\emph{)}, that supports edge queries of $G$ in time $\Oo(1/\varepsilon)$ \emph{(}resp.~$\Oo_d(\log \log n)$\emph{)}.
\end{theorem}

All the results mentioned in this section extend from graphs to binary structures.

%

\subsection{Parameterized complexity of model checking}\label{subsec:comp-compl}

First-order (\FO) matrix model checking asks, given a matrix $M$ (or a totally ordered binary structure~$\mathcal S$) and a first-order sentence $\varphi$ (i.e., a formula without any free variable), if $M \models \varphi$ holds, that is, if $\varphi$ is true in $M$.
\FO~model checking is fixed-parameter tractable (\FPT) on a matrix class $\mathcal M$, with respect to the sentence size and the input matrix, if there exists a constant~$c$ and a computable function $f$, such that $M \models \varphi$ can be decided in time $f(|\varphi|)\,(m+n)^c$, for every $n \times m$ matrix $M \in \mathcal M$ and \FO~sentence $\varphi$.

\FO~model checking of general (unordered) graphs is $\AW[*]$-complete~\cite{Downey96}, and thus very unlikely to be \FPT.
Indeed $\FPT \neq \AW[*]$ is a much weaker assumption than the already widely-believed Exponential Time Hypothesis~\cite{Impagliazzo01}, and if false, would in particular imply the existence of a subexponential algorithm solving \textsc{3-SAT}.
\FO~model checking of general binary structures of bounded twin-width given with an $\Oo(1)$-sequence can even be solved in linear \FPT~time $f(|\varphi|)\,|U|$, where $U$ is the domain of the structure~\cite{twin-width1}.

Gajarský et al.~\cite{Gajarsky22} reproved that result using a different, and more standard formalism.
Building on~\cref{thm:compact-rep}, they also presented an algorithm that inputs a binary structure given with an $\Oo(1)$-sequence and a formula with some free variables, and after some linear-time processing, can answer queries of the form \emph{does the given tuple satisfy the formula in the structure} in doubly-logarithmic time.  

\begin{theorem}[\cite{Gajarsky22}]\label{thm:gaj-linear}
  Given a binary $\Sigma$-structure $\mathbf A$ on a domain of size $n$, a~$d$-sequence of $\mathbf A$, and a first-order $\Sigma$-formula $\varphi(x_1,\ldots,x_k)$, there is a data structure computable in time $\Oo_{d,\varphi}(n)$ that given any query of the form $v_1, \ldots, v_k \in A$ reports in time $\Oo_{d,\varphi}(\log \log n)$ whether $\mathbf A \models \varphi(v_1, \ldots, v_k)$ holds. 
\end{theorem}

Up to increasing the preprocessing time to near-linear, the queries can even be met in constant time. 

\begin{theorem}[\cite{Gajarsky22}]\label{thm:gaj-quasilinear}
  For every $\varepsilon > 0$, given a binary $\Sigma$-structure $\mathbf A$ on a domain of size~$n$, a~$d$-sequence of $\mathbf A$, and a first-order $\Sigma$-formula $\varphi(x_1,\ldots,x_k)$, there is a data structure computable in time $\Oo_{d,\varphi}(n^{1+\varepsilon})$ that given any query of the form $v_1, \ldots, v_k \in A$ reports in time $\Oo_{d,\varphi}(1/\varepsilon)$ whether $\mathbf A \models \varphi(v_1, \ldots, v_k)$ holds. 
\end{theorem}

In classes of ordered binary structures or matrices, one need not require that the contraction sequence is given in input.
Indeed there is an \FPT~approximation algorithm, that takes a matrix $M$ of twin-width $d$, and outputs a $g(d)$-sequence of $M$ in time $h(d)\,|M|^{\Oo(1)}$~\cite{twin-width4}.
Hence, \FO~matrix model checking can be solved in \FPT~time $f(|\varphi|)\,|M|^{\Oo(1)}$~\cite{twin-width4} in classes of bounded twin-width.

\subsection{Interpretations and transductions}\label{subsec:fmt}

Let $\Sigma,\Gamma$ be relational signatures. 
A~\emph{(simple) interpretation} $\mathsf I\from \Sigma\to\Gamma$ 
consists of the following $\Sigma$-formulas: a \emph{domain} formula~$\nu(x)$,
and for each relation symbol $R\in\Gamma$ of arity $r$, a formula $\rho_R(x_1,\ldots, x_r)$.
If~$\mathbf A$ is a $\Sigma$-structure, the $\Gamma$-structure $\mathsf I(\mathbf A)$ has domain $\nu(\mathbf A)=\{v\in A:\mathbf A\models\nu(v)\}$ and the interpretation of a relation symbol $R\in\Sigma$ of arity $r$ is $\rho_R(\mathbf A)\cap \nu(\mathbf A)^{r}$, that is:
\[
R^{\mathsf I(\mathbf A)}=\{(v_1,\dots,v_{r})\in \nu(\mathbf A)^{r}:\mathbf A\models \rho_R(v_1,\dots,v_r)\}.
\]
Note that \cref{thm:gaj-linear,thm:gaj-quasilinear} can be seen as efficiently computing the structure $\mathsf I(\mathbf A)$ given the formulas $\rho_R(x_1,\ldots, x_r)$, when $\mathbf A$ has bounded twin-width and is given with a contraction sequence.
If $\CC$ is a~set of $\Sigma$-structures then denote $\mathsf I(\CC)=\setof{\mathsf I(\str A)}{\str A\in\CC}$.
A~class $\CC$ \emph{interprets} a class $\DD$ if there is an interpretation $\mathsf I$ such that $\mathsf I(\CC)\supseteq\DD$.

Let $\Sigma\subseteq \Sigma^+$ be relational signatures.
The \emph{$\Sigma$-reduct} of a $\Sigma^+$-structure $\mathbf A$ is the structure obtained from $\mathbf A$ by ``forgetting'' all the relations not in $\Sigma$; we denote this interpretation as $\reduct_{\Sigma}\from \Sigma^+\to\Sigma$.
A~class $\CC$ of {$\Sigma$-structures} \emph{transduces} a class $\DD$ if there is a class $\CC^+$ of {$\Sigma^+$-structures},
where $\Sigma^+$ is the union of $\Sigma$ and some unary relation symbols such that $\reduct_{\Sigma}(\CC^+)=\CC$ and $\CC^+$ interprets~$\DD$.
Here, the unary relation symbols and the underlying interpretation are called~\emph{transduction} (or \FO-transduction).

As it will be enough for our purposes, we will here take the following characterization as the definition of monadic dependence.
 \begin{theorem}[Baldwin and Shelah \cite{BS1985monadic}]\label{thm:baldwin-shelah}
 	A class $\mathcal C$ of $\Sigma$-structures is monadically dependent if and only if $\mathcal C$ does not transduce the set of all finite graphs.
 \end{theorem}

 Importantly transductions preserve the boundedness of twin-width.
 \begin{theorem}[\cite{twin-width1}]\label{thm:transductions-tww}
   Let $\mathcal C$ be a set of $\Sigma$-structures of bounded twin-width, and $\mathsf T \from \Sigma \to \Sigma$ be a~transduction.
  Then $\mathsf T(\mathcal C)$ has bounded twin-width.
 \end{theorem}

 We will deal with \fomc-transductions, which are defined as \FO-transductions but with first-order logic augmented with modular counting quantifiers.
 We postpone the relevant definitions to~\cref{sec:fomc}.

 \subsection{Organization}

 In~\cref{sec:parity-minors}, we introduce the parity and \lms.
 In~\cref{sec:fomc}, we define the logic \fomc and show that parity and \lms can be expressed with \fomc-transductions.
 In~\cref{sec:strict}, we use the two previous sections to show~\cref{thm:intro-pm}, that is, the equivalence between bounded twin-width and linear-minor freeness.
 In~\cref{sec:divisions}, we introduce the rank bidimensionality, and classify several matrix problems involving a division of the rows and columns as being fixed-parameter tractable.
 In~\cref{sec:product}, we use the extensions of~\cref{sec:fomc} to show that bounded twin-width is preserved by matrix product, and give an almost quadratic algorithm to multiply two matrices of bounded twin-width (\cref{thm:intro-square}).
 In~\cref{sec:efficient}, we present a linear-time (i.e., possibly sublinear in the number of non-zero entries) algorithm when further given a twin-decomposition of the two matrices to multiply.

\section{Parity and linear minors}\label{sec:parity-minors}

We now introduce the notion of \emph{parity minor} for matrices over a~finite field.
Let $M$ be a~matrix with entries in a field $\mathbb F$.
A \emph{deletion operation} (or simply \emph{deletion}) in $M$ consists of deleting a row or a column.
A \emph{sum operation} (or simply \emph{sum}) in $M$ consists of replacing any pair of \emph{consecutive} rows $r_i,r_{i+1}$ or columns $c_j,c_{j+1}$ by their pointwise sum in $\mathbb F$.
That is, $r_{i+1}$ is deleted and $r_i$ is replaced by $r'_i$ with $r'_i[k] = r_i[k]+r_{i+1}[k]$, for every column index $k$, where $+$ is the addition in $\mathbb F$.
We say that a matrix $N$ is a \emph{parity minor} of $M$, denoted by $N \infpar M$, if $N$ can be obtained from $M$ by a sequence of deletions and sums.

\begin{observation}\label{obs:pm-del-sum}
If $N$ is a parity minor of $M$, then $N$ can be obtained from $M$ by performing all the deletions before performing all the sums.
\end{observation}
\begin{proof}
  Any (single) deletion performed just after a (single) sum can be equivalently performed just before the sum.
  If the deletion is precisely on the row or column of the sum, one can equivalently remove the two consecutive rows or columns and not perform the sum.
\end{proof}
On the contrary, as we impose the sums to be on consecutive rows or columns, one \emph{cannot} necessarily perform all the sums before all the deletions. 

\Cref{obs:pm-del-sum} implies the following reformulation.
A \emph{parity minor} of a matrix $M$ is any $n \times m$ matrix obtained by summing up every cell of an $(n,m)$-division of a submatrix of $M$.
Indeed, after deleting the rows and columns not part of the submatrix, one can obtain the minor by summing every row and column part of the $(n,m)$-division into a single row and column.
This equivalent definition justifies the term of ``parity minor'';
Over the binary field~$\mathbb F_2$, the minor operation boils down to dividing a submatrix and keeping from each cell its parity of 1-entries.
See~\cref{fig:parity-minor} for an illustration.

\begin{figure}[ht]
  \centering
  \begin{tikzpicture}
    \def\s{0.5}
    \def\hb{\s/2}
    \def\vb{\s/2}
    \def\he{8.5 * \s}
    \def\ve{7.5 * \s}
    \def\op{80}
    \foreach \i/\j/\v in {1/1/0,1/2/0,1/3/1,1/4/1,1/5/0,1/6/0,1/7/1, 2/1/0,2/2/0,2/3/0,2/4/1,2/5/0,2/6/1,2/7/0, 3/1/1,3/2/1,3/3/0,3/4/0,3/5/0,3/6/1,3/7/1, 4/1/1,4/2/1,4/3/1,4/4/0,4/5/0,4/6/0,4/7/1, 5/1/0/\op,5/2/1/\op,5/3/1,5/4/1,5/5/0,5/6/0,5/7/0, 6/1/0,6/2/1,6/3/0,6/4/0,6/5/0,6/6/1,6/7/0, 7/1/0,7/2/0,7/3/1,7/4/1,7/5/0,7/6/0,7/7/1, 8/1/1,8/2/0,8/3/0,8/4/1,8/5/1,8/6/1,8/7/1}{
      \node (e\i\j) at (\s * \i,\s * \j) {$\v$} ;
    }
    \draw (\hb-0.05,\vb) -- (\hb-0.05,\ve) --++(0.1,0) ;
    \draw (\hb-0.05,\vb) --++(0.1,0) ;
    \draw (\he+0.05,\vb) -- (\he+0.05,\ve) --++(-0.1,0) ;
    \draw (\he+0.05,\vb) --++(-0.1,0) ;
    \node at (-0.3,4 * \s) {$M =$};
    
    \begin{scope}[xshift=10 * \s cm]
    \foreach \i/\j/\v/\o in {1/1/0,1/2/0,1/3/1/\op,1/4/1,1/5/0,1/6/0/\op,1/7/1, 2/1/0,2/2/0,2/3/0/\op,2/4/1,2/5/0,2/6/1/\op,2/7/0, 3/1/1/\op,3/2/1/\op,3/3/0/\op,3/4/0/\op,3/5/0/\op,3/6/1/\op,3/7/1/\op, 4/1/1,4/2/1,4/3/1/\op,4/4/0,4/5/0,4/6/0/\op,4/7/1, 5/1/0/\op,5/2/1/\op,5/3/1/\op,5/4/1/\op,5/5/0/\op,5/6/0/\op,5/7/0/\op, 6/1/0,6/2/1,6/3/0/\op,6/4/0,6/5/0,6/6/1/\op,6/7/0, 7/1/0,7/2/0,7/3/1/\op,7/4/1,7/5/0,7/6/0/\op,7/7/1, 8/1/1,8/2/0,8/3/0/\op,8/4/1,8/5/1,8/6/1/\op,8/7/1}{
      \node (e\i\j) at (\s * \i,\s * \j) {\textcolor{white!\o!black}{$\v$}} ;
    }
    \draw (\hb-0.05,\vb) -- (\hb-0.05,\ve) --++(0.1,0) ;
    \draw (\hb-0.05,\vb) --++(0.1,0) ;
    \draw (\he+0.05,\vb) -- (\he+0.05,\ve) --++(-0.1,0) ;
    \draw (\he+0.05,\vb) --++(-0.1,0) ;
    \foreach \i in {2,6}{
      \draw[very thick] (\i * \s + \s/2,\vb) -- (\i * \s + \s/2,\ve) ;
    }
    \foreach \j in {2,4,6}{
      \draw[very thick] (\hb,\j * \s + \s/2) -- (\he,\j * \s + \s/2) ;
    }
    \end{scope}

    \begin{scope}[xshift=21.5 * \s cm]
    \foreach \i/\j/\v in {1/1/0,1/2/0,1/3/0,1/4/1, 2/1/1,2/2/0,2/3/0,2/4/1, 3/1/1,3/2/0,3/3/1,3/4/0}{
      \node (e\i\j) at (\s * \i,\s * \j + 1.5 *\s) {$\v$} ;
    }
    \draw (\hb-0.05,2 * \s) -- (\hb-0.05,6 * \s) --++(0.1,0) ;
    \draw (\hb-0.05,2 * \s) --++(0.1,0) ;
    \draw (3.5 * \s +0.05,2 * \s) -- (3.5 * \s+0.05,6*\s) --++(-0.1,0) ;
    \draw (3.5 * \s +0.05,2 * \s) --++(-0.1,0) ;

    \node at (-0.3,4 * \s) {$N =$};
    \end{scope}
  \end{tikzpicture}
  \caption{A parity minor, equivalently \lm, $N$ of a matrix $M$ over $\mathbb F_2$.
  In the middle, the deleted rows and columns of $M$ are in light gray, while the $(4,3)$-division of the remaining submatrix giving rise to $N$ is represented with solid black lines.}
  \label{fig:parity-minor}
\end{figure}

We now define a more permissive notion of \emph{\lm} by replacing the sum operation by a weighted sum, or linear combination.
A \emph{weighted sum operation} (or simply \emph{weighted sum}) in $M$ consists of replacing any pair of consecutive rows $r_i,r_{i+1}$ (or columns $c_j,c_{j+1}$) by $\alpha r_i + \beta r_{i+1}$ (or $\alpha c_j + \beta c_{j+1}$) for some chosen $\alpha, \beta \in \mathbb F$.
Again it means that $r_{i+1}$ is deleted and $r_i$ is replaced by $r'_i$ with $r'_i[k] = \alpha r_i[k] + \beta r_{i+1}[k]$, for every column index $k$, where operations are performed in~$\mathbb F$.
Choosing $\alpha = 1_\mathbb F$ and $\beta = 0_{\mathbb F}$ emulates the deletion of $r_{i+1}$, while $\alpha = 0_\mathbb F$ and $\beta = 1_{\mathbb F}$ corresponds to the deletion of $r_i$.
Thus we no longer need to add the deletion operations.
We say that a matrix $N$ is a \lm of $M$, denoted by $N \inflin M$, if $N$ can be obtained from $M$ by a sequence of weighted sums.
Let us call \emph{$\mathbb F$-weighting of $M$} (or simply \emph{weighting of $M$} if $\mathbb F$ is clear from the context) any mapping $w: \row(M) \cup \col(M) \to \mathbb F$.
Equivalently a \lm of a matrix $M$ over a finite field $\mathbb F$ is any $n \times m$ matrix $N$ obtained from an $\mathbb F$-weighting $w$ and $(n,m)$-division $\mathcal D=(\{R_1, \ldots, R_n\}, \{C_1, \ldots, C_m\})$ of $M$ by replacing every cell $M[R_i,C_j]$ of $\mathcal D$ by the single entry $\sum_{r \in R_i, c \in C_j} w(r)w(c)M_{r,c}$, that is, setting $N_{i,j} = \sum_{r \in R_i, c \in C_j} w(r)w(c)M_{r,c}$.
We can indeed rewrite the entries of $N$ that way, since every finite field is commutative~\cite{Maclagan1905}.

The second definition actually works for any commutative field, while the first one accommodates any field.
Yet another way to see a \lm over a commutative field is that one first divides (i.e., partitions into intervals) the column set of $M$, picks one column vector in the span of each column part of $M$ to form the matrix $M'$, divides now the row set of $M'$, and finally picks one row vector in the span of each row part of $M'$ to obtain $N$ a \lm of $M$.
Of course the row division may precede the column division instead. 
We observe that over the binary field the only possible weighted sums are deletions and (simple) sums.

\begin{observation}\label{obs:f2}
  Over $\mathbb F_2$, parity minors and \lms coincide.
\end{observation}

We recall that a matrix class is a set of matrices which is closed under taking submatrices.
The \emph{parity-minor closure} of $\mathcal M$, denoted by $\minclos(\mathcal M)$, is the matrix class of all matrices $N$ which are parity minors of some $M$ in $\mathcal M$.
Similarly, the \emph{linear-minor closure} of $\mathcal M$, denoted by $\linmin(\mathcal M)$, is the matrix class of all matrices $N$ which are \lms of some $M$ in $\mathcal M$.
We say that $\mathcal M$ is \emph{\strict} (resp.~\emph{\lmf}) if its parity-minor closure (resp.~linear-minor closure) is not the set of all $\mathbb F$-matrices.
We may denote the set of all $\mathbb F$-matrices by $\mathcal M(\mathbb F)$, or by $\Mall$ if the field $\mathbb F$ is clear from the context.
Thus $\mathcal M$ is \lmf over~$\mathbb F$ if $\linmin(\mathcal M) := \{N~:~N \inflin M,~M \in \mathcal M\} \neq \mathcal M(\mathbb F)$.
We also denote by $\mathcal M_n(\mathbb F)$ the set of $n \times n$ matrices over $\mathbb F$, and by $\mathcal M_{\square}(\mathbb F)$ the set of square matrices over $\mathbb F$.

\section{First-order logic with modular counting}\label{sec:fomc}


First-order logic with modular counting, denoted by \fomc, augments \fo with a new kind of quantifiers, $\exists^{i[p]}$ for some fixed integer $p$, and any $i \in [0,p-1]$.
The semantics is extended as follows. 
For any structure $\mathcal M$ with domain $M$, $\mathcal M \models \exists^{i[p]}x~\varphi(x)$ holds whenever $|\{a \in M $ $|$ $\mathcal M \models \varphi(a)\}| \equiv i \mod p$.
Informally, one can now express that the number of witnesses for a formula $\varphi(x)$ is equal to $i$ modulo $p$.
We may use $\exists^e$ for $\exists^{0[2]}$, and $\exists^o$ for $\exists^{1[2]}$.
As for \fo-sentences, an \fomc-sentence is said \emph{prenex} if it is formed by a succession of quantifiers followed by a quantifier-free formula.
Every sentence of quantifier rank $\ell$ is logically equivalent to a prenex sentence of quantifier rank $h(\ell)$ for some tower function $h$.

The following theorem (\cref{thm:fomc-mc}) was proven in~\cite{twin-width1} with \fo-sentences instead of \fomc-sentences, and follows from that paper with a small adaptation.
It can also be observed from the work of Gajarský et al.~\cite{Gajarsky22}.
Here we choose to adapt the proof from the first paper of the series~\cite{twin-width1}.
In both cases, making the proof self-contained would require quite a lot of background and overlap with either one of those papers, and in the end, would not bring anything new.
A~reader interested in the proof will first have to read Section 7 of~\cite{twin-width1}, and in particular, get familiar with morphism-trees (similar to Ehrenfeucht-Fraïssé game trees) and the dynamic-programming algorithm handling them.

\begin{theorem}[follows with some small adjustments from \cite{twin-width1}]\label{thm:fomc-mc}
Given a $d$-sequence of a binary structure $\mathcal S$, and a prenex \fomc-sentence~$\varphi$ of quantifier rank $\ell$, one can decide $\mathcal S \models \varphi$ in time $f(\ell,d) n$ for some computable non-elementary function $f$.
\end{theorem}
\begin{proof}
  As mentioned above, this essentially follows from the FO model-checking algorithm presented in \cite[Section~7]{twin-width1}.
  However we need to slightly modify the way we reduce the morphism-trees to account for the modular counting quantifiers.
  We do not reduce two equivalent nodes, but rather, if we find $p+1$ pairwise equivalent nodes, we only keep one such node.
  One can observe that the size of a reduct of any depth-$\ell$ morphism-tree is still bounded by a function of $\ell$ (since $p$ is an absolute constant).

  This is the only modification to the dynamic-programming updates of the theories local to the red graph.
  At every node of $MT'_\ell(G,\mathcal P_i,X)$, the existence, universality, and number of witnesses modulo $p$ among the parts with local root $X$ is preserved (regardless of the prenex sentence of quantifier depth $\ell$).
\end{proof}  

A related result by Kuske and Schweikardt~\cite{Kuske18} asserts that \fomc model checking is \FPT~in classes of locally bounded treewidth.  

Another useful fact is that FO-transductions of bounded twin-width classes have themselves bounded twin-width (see \cite[Section 8]{twin-width1}).
This can be then generalized to \fomc-transductions. 

\begin{theorem}\label{thm:fomc-closure}
  Let $\mathcal C$ be a set of binary structures with bounded twin-width, and $\mathsf{T}$ be an \fomc-transduction.
  Then $\mathsf{T}(\mathcal C)$ has bounded twin-width.
\end{theorem}
\begin{proof}
  This follows from \cite[Section 8]{twin-width1} with the change of the previous theorem.
\end{proof}

As an illustrative example, we give a simple consequence of \cref{thm:fomc-closure}: shallow vertex minors of bounded twin-width graphs have bounded twin-width.
The \emph{local complementation} at a vertex $v$ of a graph $G$ consists of replacing, in $G$, the induced subgraph $G[N(v)]$ by its complement. 
A~graph $H$ is a~\emph{depth-1 vertex minor} of a~graph $G$ if there exists an independent set $I$ of $G$ and a subset $R \subseteq G$, such that $H$ can be obtained from $G$ by local complementation at the vertices of $I$, then deletion of the vertices in $R$.
Note that the result does not depend on the order in which the local complementations are performed as $I$ is an independent set.
A~\emph{depth-$d$ vertex minor of~$G$} is inductively defined as a depth-1 vertex minor of a depth-$(d-1)$ vertex minor of $G$. 

\begin{proposition}
	For every pair of integers $d, t$, there is an integer $t'$ such that every depth-$d$ vertex minor of a graph with twin-width at most $t$ has twin-width at most $t'$.
\end{proposition}
\begin{proof}
  We shall just show that the \emph{depth-1} vertex minors of graphs with bounded twin-width have bounded twin-width.
  The class $\mathcal D_t$ of all depth-1 vertex minors of graphs with twin-width at most $t$ is an \fomc-transduction of the set of all graphs with twin-width at most $t$.
  
  The transduction works as follows: considering $I$ and $R$ as unary predicates, the new domain is defined as the set of vertices not in $R$, and the new adjacency relation is defined by the formula
  $\neg(x=y)\wedge \bigl(E(x,y)\leftrightarrow (\exists^e z\ I(z)\wedge E(x,z)\wedge E(y,z))\bigr)$, that is: two vertices $u$ and $v$ have their adjacency complemented if they have an odd number of common neighbors in $I$.
\end{proof}

We now show that the parity-minor closure of a matrix class can be expressed as an \fomc-transduction.

\begin{lemma}\label{lem:min-clos-transduction}
  Let $\mathbb F$ be a finite field. 
  There is an \fomc-transduction $\mathsf{T}_{\text{pm}}$ such that, for every matrix class~$\mathcal M$ over $\mathbb F$, $\mathsf{T}_{\text{pm}}(\mathcal M) = \minclos(\mathcal M)$.
\end{lemma}
\begin{proof}
  We first show the lemma when $\mathbb F$ is the binary field.
  The proof of that fact already contains all the ideas of the general case without presenting ``unnecessary'' technicalities.
  We recall that our 0,1-matrices $M$ are $\tau$-structures with $\tau = (R,\prec,E)$.
  The universe is the union of the sets of row and column indices, $\prec$ is interpreted as a linear order on the indices (with the row preceding the column indices), and $M \models E(x,y)$ whenever there is a 1-entry at row $x$, column $y$ in $M$.
  We write $x \preceq y$ as a short-hand for $x = y \vee x \prec y$.

  The transduction $\mathsf T_{\text{pm}}$ uses two (non-deterministic) unary relations $U$ and $D$.
  We interpret~$U$ as the rows and columns kept to form the parity minor, and $D$ as the first row (resp.~column) of a row part (resp.~column part).
  Let $N$ be any parity minor of $M$.
  Crucially there is an even number of 1-entries in a cell of a division of $M$ if and only if the number of rows with an odd number of 1-entries is even.
  This can be expressed by two nested modular counting quantifiers, and $E^N(x,y)$ can be \fomc-defined as
  $$\text{Valid}(x,y)~\land~\exists^e r~\exists^o c~\text{Rectangle}(x,y,r,c)~\wedge~E^M(r,c),~\text{with}$$
  $$\text{Valid}(x,y) = R(x)~\land~\neg R(y)~\land~D(x)~\land~D(y),~\text{and}$$
  $$\text{Rectangle}(x,y,r,c) = R(r)~\wedge~\neg R(c)~\wedge~U(r)~\wedge~U(c)~\wedge~x \preceq r~\wedge~y \preceq c$$
  $$\wedge~\forall r'~(D(r')~\wedge~x \prec r') \rightarrow r \prec r'~\wedge~\forall c'~(D(c')~\wedge~y \prec c') \rightarrow c \prec c'.$$

  For the sake of legibility, we omitted the $M$ superscript for all the relations but $E^M$.
  The universe of $N$ is defined as the indices $x$ such that $\models D(x)$, and $\prec^N$ is naturally inherited from $\prec^M$.

  One can check that there is a perfect correspondence between the parity-minor operations and the formula $E^N(x,y)$.
  Rows or columns $x$ such that $\models \neg U(x)$ are removed.
  The lines of the division start just before every row or column $x$ such that $\models D(x)$.
  All the rows or columns $x$ such that $\models \neg \exists x'~D(x')~\land~x' \preceq x$ could be equivalently removed.  
  Therefore $T_{\text{pm}}(\mathcal C) = \minclos(\mathcal C)$ holds.

  \medskip

  We now deal with the general case.
  Let $p$ be the number of elements in the finite field $\mathbb F$.
  We thus consider $\tau_p$-structures with $\tau_p=(R,\prec,E_1,E_2,\ldots,E_{p-1})$, where $R$, $\prec$ are as before, while the relation symbols $E_i$ all have arity~2.
  Our theory contains the sentence
  $$\forall x,y~\left(R(x)~\wedge~\neg R(y)\right)~\rightarrow \bigwedge_{i \in [p-1]} \left( E_i(x,y) \rightarrow \bigwedge_{j \in [p-1] \setminus \{i\}} \neg E_j(x,y)\right).$$
  So there is a unique entry $i \in [0,p-1]$ at row $x$, column $y$; namely the unique $i \in [p]$ such that $\models E_i(x,y)$ if it exists, or 0 otherwise.

  We need to generalize the binary counting (``$\exists^e r~\exists^o c$'') of the sum within a rectangular zone.
  Let us denote by $\tilde i$ the element of $\mathbb F$ corresponding to $E_i$, and $\tilde 0 = 0_{\mathbb F}$.
  If $p$ is prime, then simply $\tilde i=i$.
  We can now define $E_i^N(x,y)$ as
  $$\text{Valid}(x,y)~\wedge~\bigvee_{\substack{a: [p-1] \times [0,p-1] \to [0,p-1] \\ \sum\limits_{\substack{j \in [p-1] \\ k \in [0,p-1]}} a(j,k)k \cdot \tilde{j} = \tilde i}}~\bigwedge_{\substack{j \in [p-1] \\ k \in [0,p-1]}}~\exists^{a(j,k)[p]}r~\exists^{k[p]}c~\text{Rectangle}(x,y,r,c)~\wedge~E^M_j(r,c).$$
  Let us observe that in the previous formula, ``$a(j,k)k \cdot \tilde{j}$'' does \emph{not} involve any multiplication in $\mathbb F$.
  The term $a(j,k)k$ is a scalar integer, thus $a(j,k)k \cdot \tilde{j}$ is a sum of up to $(p-1)^2$ occurrences of $\tilde{j}$.

  The integer $a(j,k)$ aims to match the number modulo $p$ of rows with $k$ modulo $p$ occurrences of $\tilde{j}$.
  Since the order of every element $\tilde{j}$ in $(\mathbb F,+)$ divides $p$,
  $$\bigwedge_{k \in [0,p-1]} \exists^{a(j,k)[p]}r~\exists^{k[p]}c~\text{Rectangle}(x,y,r,c)~\wedge~E^M_j(r,c)$$
  indeed holds when the number of $\tilde j$-entries in the eligible region is equal to $\sum_{k \in [0,p-1]}a(j,k)k$ modulo~$p$.
  Thus, if this holds for every $j \in [p-1]$, the sum of the entries in the eligible region is $\sum\limits_{j \in [p-1], k \in [0,p-1]} a(j,k)k \cdot \tilde{j}$, which happens to be $\tilde i$.
\end{proof}

Similarly we will now show that the linear-minor closure can be expressed by an \fomc-transduction.
This is based on the proof of the previous lemma and is only slightly more technical.

\begin{lemma}\label{lem:linmin-clos-transduction}
  Let $\mathbb F$ be a finite field. 
  There is an \fomc-transduction $\mathsf{T}_{\text{lm}}$ such that, for every matrix class~$\mathcal M$ over $\mathbb F$, $\mathsf{T}_{\text{lm}}(\mathcal M) = \linmin(\mathcal M)$.
\end{lemma}
\begin{proof}
  Let $p$ be the cardinality of $\mathbb F$.
  We still use the unary relation $D$ to encode the division.
  The main difference with the proof of~\cref{lem:min-clos-transduction} is that $\mathsf{T}_{\text{lm}}$ uses $p$ extra relations $U_0, \ldots, U_{p-1}$ instead of just $U$.
  Their interpretation is that $M \models U_i(x)$ holds whenever the weighting of row or column $x$ is $\tilde{i} \in \mathbb F$.
  We impose that the unary relations $U_i$ partition the set of row and column indices (i.e., the $\mathbb F$-weighting is indeed a mapping) with the sentence
  $$\forall x~\bigvee_{i \in [0,p-1]}U_i(x)~\wedge~\bigwedge_{i \in [0,p-1]} \left( U_i(x) \rightarrow \bigwedge_{j \in [0,p-1] \setminus \{i\}} \neg U_j(x)\right).$$

  Again $R^N$ and $\prec^N$ are naturally inherited from $R^M$ and $\prec^M$.
  We shall now define $E_i^N(x,y)$ as
  $$\text{Valid}(x,y)~\wedge~\bigvee_{\substack{a: [p-1] \times [0,p-1]^3 \to [0,p-1] \\ \sum\limits_{\substack{j \in [p-1] \\ k,g,h \in [0,p-1]}} a(j,k,g,h)k \cdot \tilde g \tilde h \tilde j = \tilde i}}~\bigwedge_{\substack{j \in [p-1] \\ k,g,h \in [0,p-1]}}~\exists^{a(j,k,g,h)[p]}r~\exists^{k[p]}c~\text{Rectangle}(x,y,r,c)$$
  $$\wedge~U_g(r)~\wedge~U_h(c)~\wedge~E^M_j(r,c),$$
  where $\text{Valid}(x,y)$ is defined as before while $\text{Rectangle}(x,y,r,c)$ is now simply
  $$R(r)~\wedge~\neg R(c)~\wedge~x \preceq r~\wedge~y \preceq c~\wedge~\forall r'~(D(r')~\wedge~x \prec r') \rightarrow r \prec r'~\wedge~\forall c'~(D(c')~\wedge~y \prec c') \rightarrow c \prec c'.$$
  Now $a(j,k,g,h)$ aims to match the number modulo $p$ of rows mapped by the $\mathbb F$-weighting to~$\tilde g$ with $k$ modulo $p$ occurrences of $\tilde j$ on columns mapped by the $\mathbb F$-weighting to $\tilde h$.
\end{proof}

\section{Equivalence of bounded twin-width and linear-minor freeness}\label{sec:strict}

As a corollary of the previous section, we obtain that bounded twin-width matrix classes over finite fields are \lmf, and in particular \strict.

\begin{lemma}\label{lem:bdtww-implies-strict}
  Let $\mathbb F$ be a finite field.
  For every matrix class $\mathcal M$ over $\mathbb F$ of bounded twin-width, $\mathcal M$ is \lmf.
\end{lemma}
\begin{proof}
  By~\cref{lem:linmin-clos-transduction}, there is an \fomc-transduction $T_{\text{lm}}$ with $T_{\text{lm}}(\mathcal M) = \linmin(\mathcal M)$.
  Thus by~\cref{thm:fomc-closure}, $\linmin(\mathcal M)$ has bounded twin-width, and cannot be the set of all $\mathbb F$-matrices.
\end{proof}

For the converse, we will need the notion of \emph{rank Latin divisions} previously introduced~\cite{twin-width4}.
For any integers $k \geqslant 2$ and $d \geqslant 1$, a~\emph{rank-$k$ Latin $d$-division} of a $kd^2 \times kd^2$ matrix $M$ is a regular $d$-division $\mathcal D$ of $M$ that can be refined into a regular $d^2$-division $((R_1, \ldots, R_{d^2}), (C_1, \ldots, C_{d^2}))$ such that
\begin{itemize}
\item $\forall i \in [d^2]$, $M[R_i,C_j]$ is constant for every $j \in [d^2]$ but one $j_i$ for which it has rank~$k$,
\item $\forall j \in [d^2]$, $M[R_i,C_j]$ is constant for every $i \in [d^2]$ but one $i_j$ for which it has rank~$k$,
\item and every cell of $\mathcal D$ contains exactly one $M[R_i,C_j]$ with rank~$k$.
\end{itemize}

See the left-hand side of~\cref{fig:latin} for an illustration.
We notice that every $M[R_i,C_j]$ is a $k \times k$ matrix, since the division $\mathcal D$ and its refinement are assumed regular.
In particular the non-constant submatrices $M[R_i,C_j]$ are full rank.
The definition we give in the previous paper of the series~\cite{twin-width4} is more restrictive: The full-rank submatrices $M[R_i,C_j]$ are canonical (Ramsey-minimal) and the pattern their positions draw within the matrix $M$ is fixed (following~\cref{fig:latin}).
As we do not need these additional properties, we relaxed the definition here.
It was previously shown that matrix classes with unbounded twin-width contain matrices with rank-$k$ Latin $d$-divisions for arbitrarily large values of $k$ and $d$.
For our purpose we will only need $k=2$ and $d$ diverging.

\begin{lemma}[\cite{twin-width4}]\label{lem:latin}
  Let $\mathcal M$ be a matrix class of unbounded twin-width over a finite field.
  Then for every $d$, there is a matrix $M \in \mathcal M$ with a rank-2 Latin $d$-division.
\end{lemma}

Equipped with that technical lemma, the following fact roughly boils down to~\cref{fig:latin}.
\begin{figure}[h!]
    \centering
    \begin{tikzpicture}[scale=.35]
      \def\e{0.05}
      \pgfmathsetseed{17} 

    \foreach \i in {0,...,8}{
       \foreach \j in {0,...,8}{
        \pgfmathsetmacro{\ii}{2*\i}
        \pgfmathsetmacro{\ip}{\ii+2}
        \pgfmathsetmacro{\jj}{2*\j}
        \pgfmathsetmacro{\jp}{\jj+2}
        \pgfmathsetmacro{\col}{ifthenelse(\j==mod(\i,3)*3+floor(\i/3),,ifthenelse(rnd >.5, , 1))}
        \node at (\ii+0.5,\jj+0.5) {\col} ;
        \node at (\ii+1.5,\jj+0.5) {\col} ;
        \node at (\ii+0.5,\jj+1.5) {\col} ;
        \node at (\ii+1.5,\jj+1.5) {\col} ;
       }
    }
    
    \foreach \a/\b in {0/1,1/0, 2/6,3/6,3/7, 4/12,5/13, 6/2,7/3, 8/8,9/8,9/9,10/15,11/14,11/15, 12/4,12/5,13/5, 14/10,15/11, 16/16,16/17,17/16}{
      \node at (\a+0.5,\b+0.5) {1} ;
    }

    \draw[line width=1.4pt, scale=6, color=black!10!yellow, opacity=100] (0, 0) grid (3, 3);

     \foreach \k in {0,...,8}{
        \pgfmathsetmacro{\i}{2*\k}
        \pgfmathsetmacro{\j}{2*((mod(\k,3)*3+floor(\k/3))}
        \draw[line width=1.6pt, red] (\i+\e,\j+\e) -- (\i+\e,\j+2-\e) -- (\i+2-\e,\j+2-\e) -- (\i+2-\e,\j+\e) -- cycle;      
     }
     \foreach \i in {2,4,8,10,14,16}{
       \draw[blue] (\i,0) -- (\i,18) ;
       \draw[blue] (0,\i) -- (18,\i) ;
     }

     \begin{scope}[xshift=23cm]
     \foreach \i in {0,3,6,9}{
       \draw[line width=1.4pt, color=black!10!yellow] (\i,0) -- (\i,18) ;
     }
     \foreach \i in {0,6,12,18}{
       \draw[line width=1.4pt, color=black!10!yellow] (0,\i) -- (9,\i) ;
     }
    \foreach \k in {0,...,8}{
        \pgfmathsetmacro{\i}{\k}
        \pgfmathsetmacro{\j}{2*((mod(\k,3)*3+floor(\k/3))}
        \draw[line width=1.6pt, red] (\i+\e,\j+\e) -- (\i+\e,\j+2-\e) -- (\i+1-\e,\j+2-\e) -- (\i+1-\e,\j+\e) -- cycle;      
     }
     
    \foreach \a/\b in {0/0,0/1, 1/7, 2/12,2/13, 3/2, 3/3, 4/9, 5/14, 6/4, 7/10,7/11, 8/17}{
      \node at (\a+0.5,\b+0.5) {1} ;
    }
    \foreach \i in {0.5,4.5,7.5,10.5,17.5}{
      \draw[thick] (0,\i) -- (9,\i) ;
    }
     \end{scope}

     \begin{scope}[xshift=36cm, yshift=7.5cm]
       \foreach \i in {0,...,3}{
       \draw[line width=1.4pt, color=black!10!yellow] (\i,0) -- (\i,3) ;
     }
     \foreach \i in {0,...,3}{
       \draw[line width=1.4pt, color=black!10!yellow] (0,\i) -- (3,\i) ;
     }
    \foreach \a/\b in {0/0,1/1,1/2,2/1}{
      \node at (\a+0.5,\b+0.5) {1} ;
    }
     \end{scope}

     \foreach \i in {20.6,34.2}{
       \node at (\i,9) {$\suppar$} ;
     }
    \end{tikzpicture}
    \caption{To the left, an example of a rank-2 Latin 3-division in $\mathbb F_2$.
      The division $\mathcal D$ is represented in yellow, its refinement in blue, and the rank-2 submatrices are highlighted by red boxes.
      Any $3 \times 3$ binary matrix, say, $\left( \begin{smallmatrix} 0&1&0 \\ 0&1&1 \\ 1&0&0 \end{smallmatrix} \right)$ can be obtained as parity minor.
      First sum the two columns in each vertical part of the blue division.
      Then correct the parity by removing at most one row per red box (middle matrix).
      Finally sum the cells of the yellow division (right matrix).
    }
    \label{fig:latin}
\end{figure}
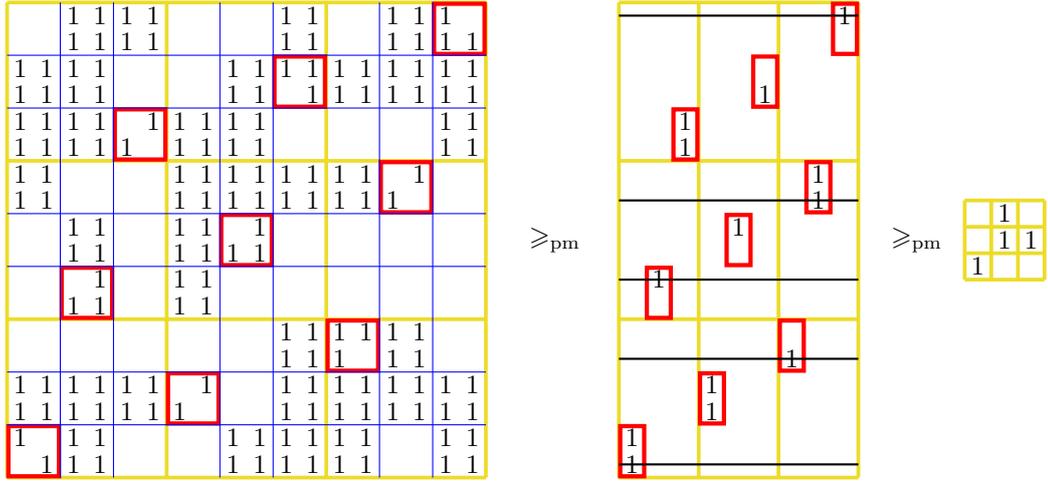

\begin{lemma}\label{lem:strict-01}
 Let $\mathbb F$ be a finite field and $\mathcal M$ be a matrix class of $\mathbb F$-matrices.
 If $\mathcal M$ is \lmf then it has bounded twin-width.
\end{lemma}

\begin{proof}
  We show the contrapositive.
  We first prove when $\mathbb F$ is the binary field.
  Let $N$ be any $n \times n$ $0,1$-matrix.
  By~\cref{obs:f2}, we want to show that $N$ is a \emph{parity} minor of some matrix in $\mathcal M$.
  Indeed observe that if all the square $0,1$-matrices are in the parity-minor closure $\minclos(\mathcal M)$, then all the $0,1$-matrices are. 

  By~\cref{lem:latin}, there is a matrix $M \in \mathcal M$ with a rank-2 Latin $n$-division $\mathcal D$.
  Let $((R_1, \ldots, R_{n^2}), (C_1, \ldots, C_{n^2}))$ be the refinement of $\mathcal D$ satisfying the properties of rank Latin divisions.
  We first sum (modulo 2) the two columns of $C_j$, renaming the resulting column $c_j$, for every $j \in [n^2]$.
  We call $M'$ the obtained $2n^2 \times n^2$ matrix, and still refer to $\mathcal D$ for the division induced by the initial $\mathcal D$ on $M'$. 
  For every $j \in [n^2]$, $M'_{R_i,c_j}$ is a $2 \times 1$ matrix equal to $\left( \begin{smallmatrix} 0 \\ 0 \end{smallmatrix} \right)$ for every $i \in [n^2]$ but one value $i_j$ for which it is necessarily different from $\left( \begin{smallmatrix} 0 \\ 0 \end{smallmatrix} \right)$.
  Indeed summing the two columns of a $2 \times 2$ of rank 2 cannot yield a zero vector, while for both $\left( \begin{smallmatrix} 0&0 \\ 0&0 \end{smallmatrix} \right)$ and $\left( \begin{smallmatrix} 1&1 \\ 1&1 \end{smallmatrix} \right)$ it does.

By definition of a rank Latin division, the non-zero $2 \times 1$ submatrices corresponding to the initial rank-2 submatrices are on pairwise disjoint sets of rows.
For each $i,j \in [n]$, in the $(i,j)$-cell of $\mathcal D$ (in $M'$), we can thus remove one row containing a 1 whenever the parity of the number of 1-entries within that cell does not match $N_{i,j}$.
After these deletions, we sum all the remaining rows and columns in each row, then column, part of $\mathcal D$, and obtain the matrix $N$.
See~\cref{fig:latin} for a visual depiction.

\medskip

The proof for the general case is not more complicated.
Instead of summing the two columns $c'_j,c''_j$ of $C_j$, we perform the weighted sum $c'_j - c''_j$ (by picking $\alpha = 1_{\mathbb F}$ and $\beta = -1_{\mathbb F}$) and again call the resulting column $c_j$, and the resulting matrix, $M'$.
This has the same effect as canceling every constant submatrices, while \emph{not} canceling every rank-2 submatrix.
Now for every $2 \times 1$ submatrix on column $c$ boxed in red in~\cref{fig:latin} within the $(i,j)$-cell of $\mathcal D$, traversed by, say, rows $r$ and $r'$, we perform the weighted sum $(N_{i,j} x^{-1})r+0_{\mathbb F}r'$ if $M'_{r,c} = x \neq 0_{\mathbb F}$ (it is not possible that $M'_{r,c} = M'_{r',c} = 0_{\mathbb F}$).
We finally get $N$ by (simply) summing every row (or column) part of $\mathcal D$ into a single row (or column).
\end{proof}

We can then add to the list of characterizations of bounded twin-width for matrix classes on finite fields (see~\cite[Theorem 1]{twin-width4}) a handful of new equivalent conditions.

\begin{theorem}\label{thm:equiv}
Given a matrix class $\mathcal M$ over a finite field, the following are equivalent.
\begin{enumerate}[$(i)$]
\item \label{it:bd-tww} $\mathcal M$ has bounded twin-width.
\item \label{it:strict} \textbf{$\mathcal M$ is \lmf.}
\item \label{it:m-nip} $\mathcal M$ is monadically dependent, i.e., for every \fo-transduction $\mathsf T$, $\mathsf T(\mathcal M) \neq \Mall$.
\item \label{it:m-nip-mc} \textbf{For every \fomc-transduction $\mathsf T$, $\mathsf T(\mathcal M) \neq \Mall$.}
\item \label{it:exp-speed} $\mathcal M$ is small.
\item \label{it:exp-speed-mc} \textbf{For every \fomc-transduction $\mathsf T$, $\mathsf T(\mathcal M)$ is small.}
\item \label{it:fact-speed} $\exists n_0 \in \mathbb N$, $|\mathcal M_n| \leqslant n!$, $\forall n \geqslant n_0$.
\item \label{it:fpt-mc} Given $M \in \mathcal M$ and $\varphi \in {\rm FO}[\tau]$, $M \models \varphi$ can be decided in $f(|\varphi|) \cdot |M|^{\Oo(1)}$.
\item \label{it:fpt-mcmc} \textbf{Given $M \in \mathcal M$ and $\varphi \in \fomc[\tau]$, $M \models \varphi$ can be decided in $f(|\varphi|) \cdot |M|^{\Oo(1)}$.}
\end{enumerate}
\end{theorem}

\begin{proof}
  All the items that are \emph{not} in bold were previously proven equivalent~\cite{twin-width4}.
  We shall now prove that~\pref{it:strict}, \pref{it:m-nip-mc}, \pref{it:exp-speed-mc}, and \pref{it:fpt-mcmc} are also equivalent to bounded twin-width.
  Let us note that the implication from~\pref{it:fpt-mc} (and from ~\pref{it:fpt-mcmc}) to the other items is conditional on the complexity-theoretic assumption \FPT~$\neq$ \AW$[*]$.
  The equivalence \pref{it:bd-tww} $\Leftrightarrow$ \pref{it:strict} is what \cref{lem:bdtww-implies-strict,lem:strict-01} show.

  \pref{it:m-nip-mc} directly implies \pref{it:m-nip} since an \fo-transduction is a particular \fomc-transduction, and \pref{it:exp-speed-mc} directly implies \pref{it:exp-speed} since the identity is an \fomc-transduction.
  For every matrix class $\mathcal M$ of bounded twin-width and \fomc-transduction~$\mathsf T$, $\mathsf T(\mathcal M)$ has bounded twin-width, by~\cref{thm:fomc-closure}.
  Thus~$\mathsf T(\mathcal M)$ is small and not equal to $\Mall$.
  This means that~\pref{it:bd-tww} implies \pref{it:m-nip-mc} and \pref{it:exp-speed-mc}.

  Finally \pref{it:fpt-mcmc} immediately implies \pref{it:fpt-mc}, while the converse is given by~\cref{thm:fomc-mc}.
  As for ordered binary structures there is a polytime algorithm to find $\Oo(1)$-sequences (when one exists) \cite[Theorem 2]{twin-width4}, we do \emph{not} need to require that the contraction sequence is provided in the input.
\end{proof}

\cref{thm:equiv} could be equivalently stated in terms of hereditary classes of ordered binary structures.
Let us discuss what the new equivalent conditions of~\cref{thm:equiv} (in bold) bring.

The equivalences \pref{it:fpt-mc} $\Leftrightarrow$ \pref{it:fpt-mcmc} and \pref{it:m-nip} $\Leftrightarrow$ \pref{it:m-nip-mc} answer a couple of questions in the special case of ordered binary structures.
It is indeed an intriguing open question whether \fomc model checking is strictly harder than \fo model checking; more precisely, whether there is a hereditary class of graphs (or binary structures) on which \fo model checking is fixed-parameter tractable while \fomc model checking is \W$[1]$-hard.
\cref{thm:equiv} shows that if any such class exists, it has to be found among the \emph{unordered} binary structures.
Similarly, we do not know if the monadic dependence of a hereditary set of binary structures is equivalent to the absence of an \fomc-transduction onto the set of all graphs.
Replacing \emph{\fomc-transduction} by \emph{\fo-transduction}, this is a theorem by Baldwin and Shelah~\cite{BS1985monadic}, which we actually took for the definition of monadic dependence.

Thus, it is again a question about the extra power (or lack thereof) of modular counting quantifiers.
The equivalence \pref{it:m-nip} $\Leftrightarrow$ \pref{it:m-nip-mc} settles this question when restricted to ordered binary structures.  
Note that there are many properties that are expressible in \fomc but not in \fo, the simplest being that \emph{the universe has an even number of elements} (over the empty signature).
However, it is unclear if these properties can be used for a separation of the kinds \pref{it:fpt-mc} $\nRightarrow$ \pref{it:fpt-mcmc} or \pref{it:m-nip} $\nRightarrow$ \pref{it:m-nip-mc} on unordered binary structures.

The equivalence \pref{it:bd-tww} $\Leftrightarrow$ \pref{it:strict} gives an elegant characterization of bounded twin-width for ordered binary structures.
As we open our introduction with, it is reminiscent of the ties between treewidth and graph minors.
As bounded treewidth can be characterized as avoiding one planar graph as a minor, bounded twin-width can be characterized for matrix classes over finite alphabets as avoiding one matrix as a~\lm.
Finally \pref{it:exp-speed} $\Rightarrow$ \pref{it:exp-speed-mc} is a bit of a curiosity.
It says that one cannot leave the realm of small ordered classes (starting from a hereditary class) by means of a~\fomc-transduction.

One might dislike in the characterization \pref{it:bd-tww} $\Leftrightarrow$ \pref{it:strict} the seemingly arbitrary multipliers inherent to the definition of a weighted sum. 
By~\cref{obs:f2} for $0,1$-matrices, bounded twin-width can be characterized more neatly with the mere parity minors.

\begin{theorem}\label{thm:characterization}
  Let $\mathcal M$ be a matrix class over $\mathbb F_2$.
  $\mathcal M$ has bounded twin-width if and only if it is \strict.
\end{theorem}

The previous statement might hold more generally in any prime field.
For non-prime fields, an easy counterexample is the set of all the $0,x$-matrices, where $x \neq 0_{\mathbb F}$ has order in $(\mathbb F,+)$ strictly smaller than $|\mathbb F|$.

\section{Fixed-parameter algorithms for matrix division problems}\label{sec:divisions}

We show how to use the results of~\cref{sec:fomc} and the approximation algorithm of matrix twin-width~\cite[Theorem 2]{twin-width4}, to decide matrix problems involving a division (of a submatrix) with some \fomc-definable properties in fixed-parameter time.
This is based on a win-win argument generalizing the algorithmic scheme of Guillemot and Marx~\cite{Guillemot14} to solve \textsc{Permutation Pattern}, and somewhat resembling the bidimensionality technique~\cite{bidim}.
It allows for instance to detect a $k$-grid minor or a $k$-mixed minor in an $n \times n$ matrix, or to decide if a $k \times k$ matrix is a parity or \lm of an $n \times n$ matrix in time $f(k)n^{\Oo(1)}$. 

We recall three results from the previous paper of the series.
The first one is an approximation algorithm for the twin-width of matrices on finite alphabets.
It outputs a large rich division if the twin-width is too high.

\begin{theorem}[\cite{twin-width4}]\label{thm:approx-tww}
Given as input an $n \times m$ matrix $M$ over a finite field $\mathbb F$, and an integer~$k$, there is a $2^{2^{\Oo(k^2 \log k)}}(n+m)^{\Oo(1)}$-time algorithm which returns
\begin{compactitem}
\item either a $2k(k+1)$-rich division of $M$, certifying that $\tww(M) > k$, 
\item or a contraction sequence certifying that $\tww(M) = 2^{\Oo(k^4)}$.
\end{compactitem}
\end{theorem}

The second result is the fact that huge rich divisions contain large rank divisions. 

\begin{theorem}[\cite{twin-width4}]\label{thm:rd-to-gr}
  Let $\mathbb F$ be a finite field and $K$ be equal to $|\mathbb F|^{|\mathbb F|^k \mt(k|\mathbb F|^k)}$.
  Every $\mathbb F$-matrix $M$ with a $K$-rich division has a~rank-$k$ division.
\end{theorem}

The third result turns rank divisions into the more structured rank Latin divisions.

\begin{lemma}[\cite{twin-width4}]\label{lem:rd-rld}
 Let $\mathbb F$ be a finite field.
 There is a computable function $f: \mathbb N \to \mathbb N$ such that every $\mathbb F$-matrix with a rank-$f(k)$ division has a submatrix with a rank-$k$ Latin division. 
\end{lemma}

We derive the following convenient corollary. 

\begin{theorem}\label{thm:approx-tww2}
Given as input an $n \times m$ matrix $M$ over a fixed finite field $\mathbb F$, and an integer~$k$, there is an $f(k)(n+m)^{\Oo(1)}$-time algorithm which returns
\begin{compactitem}
\item either a rank-$k$ Latin division of a submatrix of $M$,
\item or a contraction sequence certifying that $\tww(M) \leqslant g(k)$.
\end{compactitem}
where $f$ and $g$ are computable functions.
\end{theorem}
\begin{proof}
  This is a direct consequence of \cref{thm:approx-tww,thm:rd-to-gr,lem:rd-rld}.
  The proof of \cref{thm:rd-to-gr} is effective and can readily be turned into an \FPT~algorithm.
  Indeed besides some rank computations and vector comparisons, it mainly uses the Marcus-Tardos theorem~\cref{thm:marcus-tardos}, which is effective (see~\cite[Appendix A]{GuillemotM13}).
  The proof of~\cref{lem:rd-rld} consists of successive Ramsey extractions, which can be done in polynomial time.
\end{proof}

A parameterized problem $\Pi$, taking as input a matrix over a fixed finite field and a non-negative integer, is \emph{\fomc-definable} if, there is a computable function $f$, and for every non-negative integer~$k$, there is a $\fomc[\tau]$ sentence $\varphi_{\Pi,k}$ of size $f(k)$ such that $(M,k)$ is a YES-instance of $\Pi$ if and only if $M \models \varphi_{\Pi,k}$.
A parameterized problem $\Pi$, taking as input a matrix $M$ and a non-negative integer $k$, is said \emph{rank-bidimensional} if, for some computable functions $f$ and $g$, the existence of a rank-$f(k)$ Latin division in $M$ (i.e., a~submatrix of $M$ has a rank-$f(k)$ Latin division) permits to decide $\Pi$ in time $g(k)|M|^{\Oo(1)}$.
We can now state the main observation of this section.

\begin{theorem}\label{thm:bidim}
Every \fomc-definable rank-bidimensional problem is in \FPT.
\end{theorem}
\begin{proof}
  Let $\Pi$ be a rank-bidimensional problem, with computable functions $f$ and $g$.
  Let $(M,k)$ be an input of $\Pi$.
  We run the algorithm of~\cref{thm:approx-tww2} with parameter $f(k)$.
  In time $f'(k)|M|^{\Oo(1)}$, this either yields a rank-$f(k)$ division of a submatrix of $M$, and we can decide $(M,k)$ in further time $g(k)|M|^{\Oo(1)}$ (since $\Pi$ is rank-bidimensional), or a $g'(k)$-sequence of $M$, and we can conclude by~\cref{thm:fomc-mc} (since $\Pi$ is \fomc-definable).
\end{proof}

As a corollary of \cref{thm:bidim}, we obtain for instance that deciding if $N$ is a parity minor of $M$ is fixed-parameter tractable in the size of $N$. 

\begin{theorem}\label{thm:parity-minor-test}
  Let $N$ be a $k \times k$ matrix, and $M$ be an $n \times m$ matrix, both over~$\mathbb F_2$.
  One can decide $N \infpar M$ in time $f(k)(n+m)^{\Oo(1)}$ for some computable function~$f$.
\end{theorem}
\begin{proof}
  Let us call this problem parameterized by $k$, \pmc.
  By the proof of~\cref{lem:strict-01}, \pmc is rank-bidimensional.
  We shall then prove that \pmc is \fomc-definable, and conclude by~\cref{thm:bidim}. 
  This is not a mere consequence of the proof of~\cref{lem:min-clos-transduction} since we can no longer rely on the non-deterministic augmentation by unary relations to only keep the desired rows and columns.

  Instead, we claim that any $k \times k$ parity minor can be realized (after deletions) by a~division where every row part and every column part have size at most $k+1$.
  Indeed, let us associate to a row or column $x$ the parity vector $p(x) \in \mathbb F_2^k$ corresponding to the parity of its number of 1-entries in each of the $k$ cells it intersects.
  Since $\text{dim}(\mathbb F_2^k)=k$, if a part contains at least $k+2$ rows (or columns), then one can find a set of at least one and at most $k+1$ rows summing (in $\mathbb F_2$) to the zero vector.
  This non-empty set of rows can be deleted without changing the parity minor.
  Importantly, at most $k+1$ rows are deleted, so at least one row is remaining (so the parity minor is well formed).
  It is noteworthy that we actually do not need modular counting (contrary to \cref{lem:min-clos-transduction}).
  We can express $N \infpar M$ by an {\rm FO}-sentence and conclude by invoking the algorithm of~\cite[Section 7]{twin-width1}.

  The presence of $N$ as a parity minor can be defined by a disjunction $\psi_N$ for every tuple $(s(1),\ldots,s(k),t(1),\ldots,t(k)) \in [1,k+1]^{2k}$ of the FO-sentence 
  
  $$\exists x^1_1, \ldots, \exists x^1_{s(1)}, \ldots, \exists x^k_1, \ldots, \exists x^k_{s(k)}, \exists y^1_1, \ldots, \exists y^1_{t(1)}, \ldots, \exists y^k_1, \ldots, \exists y^k_{t(k)}$$
  
  $$\bigwedge_{(i,a) \neq (i',a')} x^i_a \neq x^{i'}_{a'}~\wedge~\bigwedge_{(j,b) \neq (j',b')} y^j_b \neq y^{j'}_{b'}$$

  $$\wedge~\bigwedge_{1 \leqslant i,j \leqslant k}  \text{ODD}(x^i_1,\ldots,x^i_{s(i)},y^j_1,\ldots,y^j_{t(j)}) \leftrightarrow E^N(i,j),~\text{where}$$

  $\text{ODD}(x^i_1,\ldots,x^i_{s(i)},y^j_1,\ldots,y^j_{t(j)})$ is the formula
  
  $$\bigvee_{\substack{S \subseteq \{x^i_1,\ldots,x^i_{s(i)}\} \times \{y^j_1,\ldots,y^j_{t(j)}\} \\ |S|~\text{is odd}}}~\bigwedge_{(x^i_a,y^j_b) \in S} E^M(x^i_a,y^j_b)~\wedge~\bigwedge_{(x^i_a,y^j_b) \notin S} \neg E^M(x^i_a,y^j_b).$$

  And $M \models \psi_N$ holds if and only if $N$ is a parity minor of $M$.
\end{proof}

Similarly we can invoke \cref{thm:bidim} for the following problems.
In the following theorem, \textsc{Linear Minor Containment} is the problem of deciding $N \inflin M$, given two matrices $N, M$ over a fixed finite field, parameterized by $|N|$.

\begin{theorem}\label{thm:other-pb}
  \textsc{Linear Minor Containment}, \imc, \gmc, \mmc, \textsc{Rank-$k$ Division}, \textsc{Rank-$k$ Latin Division}, \textsc{Permutation Pattern} are fixed-parameter tractable.
\end{theorem}
\begin{proof}
It is not difficult to adapt the previous proof to show that all these problems are \fo-definable and rank-bidimensional.
\end{proof}

It is noteworthy that while we only know how to \emph{approximate} efficiently the twin-width of matrices, we can efficiently compute their maximum grid or mixed minor \emph{exactly}.  
We only included~\textsc{Permutation Pattern} in the previous theorem to highlight the fact that the framework of~\cref{thm:bidim} generalizes the fixed-parameter algorithm of Guillemot and Marx~\cite{Guillemot14}.
Their algorithm specifically targets~\textsc{Permutation Pattern} and has of course a much better running time.

Let us insist that the problems of~\cref{thm:other-pb} are solved for \emph{general} matrices (of possibly unbounded twin-width) thanks to the twin-width theory.
As witnessed by the fact that the fixed-parameter algorithm of Guillemot and Marx~\cite{Guillemot14} is a surprising result and a remarkable breakthrough, it is unlikely that there is an easy alternative proof of that theorem.

\cref{thm:approx-tww2} may also be useful for $\W[1]$-hard problems like \textsc{$k$-Biclique}~\cite{Lin18}, that are \fo-definable but \emph{not} rank-bidimensional.
Atminas et~al.~\cite{Atminas12} showed that \textsc{$k$-Biclique} has a fixed-parameter algorithm in the combined parameter \emph{$k$ plus length of the longest induced path}.
In particular, this means that \textsc{$k$-Biclique} is in \FPT~on $P_t$-free graphs. 
We also give a fixed-parameter algorithm for \textsc{$k$-Biclique} in a subclass of structures, but we phrase our result in terms of matrices rather than graphs. 
The color coding technique~\cite{Alon16} provides an \FPT~reduction from \textsc{$k$-Biclique} on bipartite graphs to \textsc{$k$-Biclique} (on general graphs).
Thus the former problem is~$\W[1]$-hard, by the breakthrough of Lin~\cite{Lin18}.
In the language of $0,1$-matrices, this can be equivalently phrased as finding a $k \times k$ submatrix full of 1-entries.

\begin{theorem}\label{thm:biclique}
  Let $\mathcal M$ be a set of $0,1$-matrices not containing every permutation matrix.
  Then \textsc{$k$-Biclique} on $\mathcal M$ is in \FPT. 
\end{theorem}

\begin{proof}
  Let $s$ be the dimension of a permutation matrix \emph{not} in $\mathcal M$, and let $k' := \max(k,s)$.
  We call the algorithm of~\cref{thm:approx-tww2} with the input matrix $M$ and parameter $k'$.
  If this yields a $g(k')$-sequence, we conclude with the algorithm of~\cite[Section 7]{twin-width1} since \textsc{$k$-Biclique} is \fo-definable:
  $$\text{BICLIQUE}_k = \exists r_1, \ldots, \exists r_k, \exists c_1, \ldots, \exists c_k~\bigwedge_{i \neq j \in [k]}\neg (r_i = r_j)~\wedge~\bigwedge_{i \neq j \in [k]}\neg (c_i = c_j)$$
  $$\wedge~\bigwedge_{i \in [k]} R(r_i)~\wedge~\bigwedge_{j \in [k]} \neg R(c_j)~\wedge~\bigwedge_{i,j \in [k]}E(r_i,c_j).$$

  If instead we get a rank-$k'$ Latin division of a submatrix of $M$, we observe that at least one of the constant $k' \times k'$ submatrices has to be full with 1-entries.
  If not, we claim that all the $k' \times k'$ permutation matrices belong to $\mathcal M$.
  Indeed, given any $k'$-permutation $\sigma$, one can, for each rank-$k'$ submatrix of the rank Latin division $\mathcal D$, keep a pair of row and column intersecting at a 1-entry in the $(i,\sigma(i))$-cell of $\mathcal D$ for every $i \in [k']$, and at a 0-entry in all the other cells.
  Thus $M$ contains every $k' \times k'$ permutation matrix as a submatrix, and we conclude since a matrix class is, by definition, submatrix-closed.
  Since $k' \geqslant s$, this contradicts the assumption on $\mathcal M$.

  We therefore get a $k' \times k'$ submatrix full of 1-entries, and can conclude that the instance is positive since $k' \geqslant k$.
\end{proof}

\cref{thm:approx-tww2} permits to show that finding a half-graph of height $k$ is in \FPT~on any matrix class missing one permutation matrix and one complement of permutation matrix.
It may for instance also be utilized to generalize the known result that \textsc{$k$-Induced Matching} is fixed-parameter tractable on bipartite $K_{t,t}$-free graphs~\cite{Dabrowski13}.
Let \textsc{$k$-Induced Permutation} be the problem which takes a $0,1$-matrix $M$ and a $k \times k$ permutation matrix $P$ in input, and asks whether $P$ is a submatrix of $M$.
This problem also generalizes \textsc{Permutation Pattern} for which $M$ would be a permutation matrix, too.
It extends \textsc{$k$-Induced Matching} in bipartite graphs to the ordered setting: Not only one wants to find $k$ mutually induced edges, but they should realize a given permutation. 

\begin{theorem}\label{thm:ind-perm}
  \textsc{$k$-Induced Permutation} is in \FPT~on any $K_{t,t}$-free matrix class.
\end{theorem}
\begin{proof}
  The proof follows~\cref{thm:biclique} by switching the roles of bicliques and permutation patterns.
  More generally the problem is fixed-parameter tractable with respect to $k+t$.
\end{proof}

The following, together with the \fo-definability of the problems at hand (and the \FPT~algorithm for FO model checking on ordered binary structures of bounded twin-width~\cite{twin-width4}), summarizes the previously drawn algorithmic consequences. 
\begin{proposition}\label{prop:bdtww-case}
  Every $0,1$-matrix class missing either one of the following combinations has bounded twin-width:
  \begin{compactitem}
  \item a constant-1 matrix and a permutation matrix, or
  \item a constant-0 matrix and a complement of permutation matrix, or
  \item a constant-1 triangular matrix, a permutation matrix, and a complement of a permutation matrix.
  \end{compactitem}
\end{proposition}

\section{Products of bounded twin-width matrices over a finite field}\label{sec:product}

Perhaps somewhat disappointingly, none of the algorithmic applications of~\cref{sec:divisions} actually requires the expressive power of \fomc and~\cref{thm:fomc-mc}.
We will now see some interesting consequences of \cref{sec:fomc} for matrix multiplication, that do require modular counting. 
First we show that if a matrix class $\mathcal M$ over a finite field has bounded twin-width, then so does its set of squares $\mathcal M^2 = \{AB :  \text{A and B are two conformal} $ $\text{matrices of}~\mathcal M\}$.

\begin{theorem}\label{thm:matrix-mult}
  There is a function $f: \mathbb N^3 \to \mathbb N$ such that for every conformal matrices $A$ and $B$ over~$\mathbb F_q$,
  The product $AB$ (over $\mathbb F_q$) has twin-width at most $f(\tww(A),\tww(B),q)$.
\end{theorem}
\begin{proof}
  Since $\left( \begin{matrix} 0&A \\ B&0 \end{matrix} \right) \cdot \left( \begin{matrix} 0&A \\ B&0 \end{matrix} \right) = \left( \begin{matrix} AB&0 \\ 0&BA \end{matrix} \right)$ and
  $$\tww\left( \left( \begin{matrix} 0&X \\ Y&0 \end{matrix} \right) \right) = \tww\left( \left( \begin{matrix} X&0 \\ 0&Y \end{matrix} \right) \right)=\max(\tww(X),\tww(Y),2),$$
  we shall just prove that there is a function $g$ such that $M^2$ has twin-width at most $g(\tww(M))$, for every $n \in \mathbb N$ and $M \in \mathcal M_n(\mathbb F_q)$.
  Identifying a matrix of $\mathcal M_{\square}(\mathbb F_q)$ and the corresponding $\tau_q$-structure, there is a simple \fomc-interpretation $\mathsf S$ such that $\mathsf S(M)=M^2$ for every $M \in \mathcal M_{\square}(\mathbb F_q)$.

  Indeed one can keep the relations $R$ and $\prec$ as in $M$, and express $E^{M^2}_i(x,y)$ as
  $$\bigvee_{\substack{a: [q-1]^2 \to [0,q-1] \\ \sum\limits_{j,k \in [q-1]^2} a(j,k) \cdot (\tilde{j} \tilde{k}) = \tilde i}}~\bigwedge_{j,k \in [q-1]}~\exists^{a(j,k)[q]}z~E_j^M(x,z)~\wedge~E^M_k(z,y).$$
  As previously, we wrote $\tilde i$ for the element of $\mathbb F_q$ corresponding to relation $E_i$.
  The expression $\tilde{j} \tilde{k}$ is a product in $\mathbb F_q$, while $a(j,k) \cdot (\tilde{j} \tilde{k})$ is the sum of $a(j,k)$ occurrences of $\tilde{j} \tilde{k}$.
  As every element of $(\mathbb F_q,+)$ has an order dividing $q$, it is enough to count the number of pairs $(\tilde j, \tilde k)=(M_{x,z},M_{z,y})$ modulo $q$, which the formula does.
  We do not assume that $\tilde j \tilde k = \tilde k \tilde j$ (although it does hold), so our formula would also work in non-commutative rings.

  We finally invoke~\cref{thm:fomc-closure} to conclude that $\tww(M^2)$ is bounded by a function of $\tww(M)$ and~$q$.
\end{proof}

\begin{theorem}\label{thm:matrix-product}
  Let $q$ be a prime power, and $d$ be a natural.
  Let $A, B$ be two $n \times n$ matrices over $\mathbb F_q$, both of twin-width at most~$d$.
  One can compute the product $AB$ in time $\Oo_{d,q}(n^2 \log n)$.
\end{theorem}
\begin{proof}
By~\cref{thm:approx-tww-quad}, we compute an $\Oo_{d,q}(1)$-sequence for $$\left( \begin{matrix} 0&A \\ B&0 \end{matrix} \right),~\text{in time}~\Oo_{d,q}(n^2 \log n).$$
We conclude either by turning this contraction sequence into a twin-decomposition in time $\Oo_{d,q}(n^2)$, by~\cref{thm:to-tww-dec}, and invoking the upcoming practical matrix squaring of~\cref{thm: matmult}, or by combining the \fomc-interpretation of~\cref{thm:matrix-mult} (and~\cref{thm:fomc-closure}) with the efficient algorithm of Gajarský et al.~\cite{Gajarsky22} (see~\cref{thm:gaj-quasilinear}) to compute the interpretations of bounded twin-width structures.
We can finally read off the top-left block $AB$ in $$\left( \begin{matrix} 0&A \\ B&0 \end{matrix} \right)^2,$$
in time $\Oo_{d,q}(n^2)$.

If we chose the former approach, we now have a~twin-decomposition $(\mathcal T, \mathcal B)$ of $AB$.
We can initialize an $n \times n$ matrix to all 0 entries, and for each edge of $\mathcal B$ labeled $\ell$, fill the corresponding entries with $\ell$.
This takes quadratic time since we access each matrix entry at most once.
If we instead went with the latter approach, we shall simply make $(q-1)n^2$ constant-time queries to build $AB$, $q-1$ for each entry of $AB$.
\end{proof}

The bottleneck of~\cref{thm:matrix-product} is to compute the contraction sequence.
Should this step be improved to run in $\Oo_{d,q}(n^2)$ time, one would get an overall quadratic algorithm to multiply two matrices of bounded twin-width.



\section{Efficient square computation given a twin-decomposition}\label{sec:efficient}

For the sake of convenience, we will now use the language of graphs rather than of matrices.
Recall that the \emph{square} of a graph $G$ is the graph $G^2$ with the same vertex set and whose edge set is given by: 
$$E(G^2):= \sg{uv~:~\text{ there is a path of length at most $2$ from $u$ to $v$ in $G$}}.$$
We similarly define the \emph{modular square} of $G$ denoted $G^{[2]}$ whose vertex set is $V(G)$ and edge set is given by: 
$$E(G^{[2]}):=\sg{uv~:~|N(u)\cap N(v)| = 1 \pmod 2}.$$
Note that if $M$ is an adjacency matrix of $G$ in $\mathbb F_2$, then the adjacency matrix of $G^{[2]}$ is exactly $M^2$.
More generally, for every $q\geq 2$, if $G$ is a complete graph with edges labeled by a function $\nu: E(G)\to \mathbb F_q$, 
we let $G^{[q]}$ be the complete graph with vertex set $V(G)$ and edges labeled by $\lambda: E(G^{[q]})\to \mathbb F_q$ defined for every $uv\in E(G^{[q]})$ by:
$$\lambda(uv):=\left(\sum_{w\in V\backslash\sg{u,v}}\nu(uw)\nu(wv) \right) \in \mathbb F_q.$$
If $M$ is an adjacency matrix of the labeled graph $G$ the corresponding adjacency matrix of $G^{[q]}$. For every prime power $q=p ^{\alpha}$, we let $m(q)$ denote the cost of basic arithmetic computations in $\mathbb F_q$; here only addition, subtraction, and multiplication are needed. Our main result in this section is the following:

\begin{theorem}
 \label{thm: matmult}
 For every prime power $q\geq 2$, there is an $\mathcal{O}(m(q)d^2q^{2d}n)$-time algorithm that, given a twin-decomposition $(\tree, \bb)$  of width $d$ of a graph $G$ with $n$ vertices whose edges are labeled in $\mathbb F_q$, outputs a twin-decomposition of width $\Oo(d^2q^d)$ of $G^{[q]}$. 
\end{theorem}

Note that in practice, if $q=p^{\alpha}$ with $p$ prime, one can choose $m(q)=\mathcal O(\log_p(q)^2\log(p)^2)$ (see for example \cite[Table 2.8]{HAC}).

\begin{figure}[h!]
  \centering
  \begin{tikzpicture}[scale=0.9]
    \foreach \i in {1,...,7}{
      \pgfmathsetmacro{\x}{( 360/7 )*(2-\i)}        
      \node[draw,circle,minimum size=0.2cm] (\i) at (\x:1.5) {$\i$} ;
    }

    \foreach \i/\j in {1/2,1/3,1/5,1/6,1/7,2/4,2/5,3/4,3/5,6/7}{
       \draw (\i) -- (\j) ;
    }

    \begin{scope}[xshift=5cm]
    \foreach \i in {1,...,7}{
      \pgfmathsetmacro{\x}{( 360/7 )*(2-\i)}        
      \node[draw,circle,minimum size=0.2cm] (\i) at (\x:1.5) {$\i$} ;
    }

    \foreach \i/\j in {1/2,1/3,1/6,1/7,2/3,2/5,2/6,2/7,3/5,3/6,3/7,5/6,5/7,6/7}{
       \draw (\i) -- (\j) ;
    }
    \end{scope}
\end{tikzpicture}
  \caption{Left: a graph $G$. Right: the graph $G^{[2]}$.}
\label{fig:G2}
\end{figure}
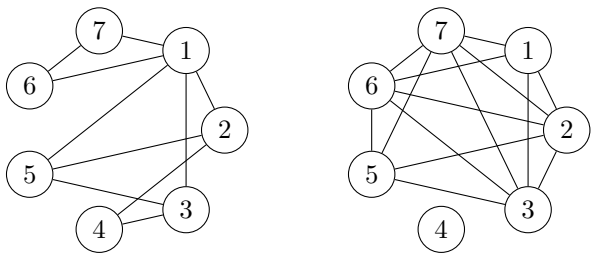

\paragraph*{Overview}

We now describe a way of computing a twin-decomposition $(\mathcal{T'}, \mathcal{B'})$ of $G^{[q]}$ when we we are given as input a twin-decomposition $(\mathcal T, \mathcal B)$ of $G$.
For the sake of clarity, we will only give the proof of \cref{thm: matmult} for the case $q=2$, and later explain (at the end of this section) how to generalize it for greater values of $q$.
From now on we assume that $G$ is a graph of twin-width at most $d\geq 0$.
We will first explain in~\cref{lem: refinement} how to refine a $d$-sequence of $G$ in order to get an~$\Oo(d^22^d)$-sequence for $G^{[2]}$.
Then we will show how to compute the associated contraction tree $\mathcal{T'}$ from $\mathcal{T}$ using dynamic programming.
The next step is to compute $\mathcal{B'}$ a~set of transversal edges such that $(\mathcal{T'}, \mathcal{B'})$ is a twin-decomposition of $G^{[2]}$.
This is the most technical part.
For this, we will first show how to construct a~set of labeled edges $\mathcal{B}_1$, with labels in $\sg{0,1}$, again using dynamic programming.
It will follow from our construction that $\mathcal{B}_1$ have size $\Oo(d^24^dn)$, and that for every two vertices $u,v\in V(G)$, we can determine whether or not $u$ and $v$ are adjacent in $G^{[2]}$ by computing the parity of the sum of the labels over the labeled edges of $\mathcal{B}_1$ we meet between the branch corresponding to $u$ and the one corresponding to $v$ in $\mathcal{T'}$, plus an additional term that we will introduce later.
Eventually we will prove how to compute the desired set $\mathcal{B'}$ from $\mathcal{B}_1$, using dynamic programming once more.
We claim that the first two steps of our computation (i.e., the computation of $\mathcal{T'}$ and $\mathcal{B}_1$) can be done at the same time, but for the sake of clarity we will present them separately and assume that $\mathcal{T'}$ is already known when we compute $\mathcal{B}_1$.

We will adopt the following convention: We denote with lowercase characters the vertices of the original graph $G$ and with uppercase characters the vertices of the graphs $G_i$ and of the trees involved in the twin-decompositions, which we always identify to subsets of $V(G)$.
Hence a vertex $U\in V(\mathcal T)$ is identified to the set of leaves of the subtree of $\mathcal T$ rooted at $U$.

\paragraph*{How to refine a $d$-sequence for $G$ to a $\Oo(d^22^d)$-sequence for $G^{[2]}$}

Assume that $G$ admits a $d$-sequence $G_n, \ldots,G_1$ and denote $\mathcal{P}_{n}, \ldots, \mathcal{P}_1$ the associated partitions of $V(G)$.
We define the following refinement of this sequence, denoted by $\mathcal P'_n, \ldots, \mathcal P'_1$, where for each $i\in [n]$, if the red neighbors of $U\in \mathcal{P}_i$ in $G_i$ are the sets $W_1, \ldots, W_k \in \mathcal{P}_i$ (with $k\leq d$), then we partition $U$ into the at most $2^{k+1}$ sets $(U_{(p, b_1, \ldots, b_k)})_{p, b_1, \ldots, b_k \in \sg{0,1}}$ defined for every $p, b_1, \ldots, b_k \in \sg{0,1}$ by: 
$$U_{(p,b_1, \ldots, b_k)} := \sg{x\in U~:~\forall i \in [k], |N(x)\cap W_i| = b_i \pmod 2 \text{ and } \mathrm{deg}_U(x) = p \pmod 2}.$$

In other words, we obtain $\mathcal P'_i$ from $\mathcal{P}_i$ by grouping together vertices from each $U\in \mathcal{P}_i$ according to the parity of their degree in $U$ and in each of the red neighbors of $U$ in $G_i$.
Note that some of the sets $U_{(p,b_1, \ldots, b_k)}$ may be empty. Nevertheless, at the end of the contraction sequence, $\mathcal{P}_1 = \sg{V(G)}$, hence $|\mathcal P'_1|\leq 2$, as it simply corresponds to the $2$-partition of $V(G)$ between vertices of even and odd degree. Thus to get a complete contraction sequence, we can eventually complete arbitrarily the sequence $\mathcal P'_n, \ldots, \mathcal P'_1$ into a contraction sequence of $G$.

\begin{lemma}
\label{lem: refinement}
 If $\mathcal{P}_n, \ldots, \mathcal{P}_1$ is the sequence of partitions of $V(G)$ associated to a $d$-sequence $G_n, \ldots, G_1$ of $G$, then the refinement $\mathcal P'_n, \ldots, \mathcal P'_1$ described above can be arbitrarily completed into a $D$-sequence for $G^{[2]}$, where $D:= (d^2 + d +1)2^{d+1} - 1$. 

\end{lemma}

\begin{proof}
 We first check that the sequence $\mathcal P'_n, \ldots, \mathcal P'_1$ can be completed into a contraction sequence.
 In other words, we show that for every $i\in [2,n]$ and for every $W_{(p,b_1,\ldots,b_k)}\in \mathcal P'_i$, there exists some $W'_{(p',b'_1,\ldots, b'_l)}\in \mathcal P'_{i-1}$ such that $W_{(p,b_1,\ldots,b_k)}\subseteq W'_{(p',b'_1,\ldots,b'_l)}$. 
 
 Assume that between $G_i$ and $G_{i-1}$, the contraction described by $(U,V,Z)$ is done, with $U,V\in \mathcal{P}_i$ and $Z:=U\cup V\in \mathcal{P}'_{i-1}$. 
 We let $W_1, \ldots, W_k$ and $W'_1, \ldots, W'_{k'}$ be respectively the red neighbors of $U$ and $V$ in $G_i$, with $k\leq d$ and $k'\leq d$.  
 Observe that the red neighbors of $Z$ in $G_{i-1}$ are the $W_j$s and the $W'_j$s distinct of $U$ and $V$, and possibly some other subsets of vertices from $\mathcal{P}_{i-1}$. We denote all the red neighbors of $Z$ in $G_i$ with $X_1, \ldots, X_l$. Let $p, b_1, \ldots, b_k \in \sg{0,1}$. We show that there exist $p',b'_1, \ldots, b'_l \in \sg{0,1}$ such that $U_{(p,b_1, \ldots, b_k)}\subseteq Z_{(p',b'_1, \ldots, b'_l)}$. For this, for every $j \in [l]$, we let $b'_j := b_r$ if $X_j = W_r$ for some $r\in [k]$. Otherwise if $X_j \neq W_r$ for every $r\in [k]$, then it means that $X_j$ is not a red neighbor of $U$ in $G_i$, and thus we let $b'_j:= 1$ if there is a complete biclique between $U$ and $X_j$ and $|X_j|$ is odd, and
 $b'_j:= 0$ otherwise. It remains to properly define $p'$. If $UV$ forms a black edge in $G_i$, then we simply let $p':= p + |V| \pmod 2$.
 If $UV$ forms a non edge in $G_i$, then we let $p':= p$.
 Finally if $UV$ forms a red edge in $G_i$, then we have $V= W_j$ for some $j \in [k]$.
 In particular, every vertex in $U_{(p,b_1, \ldots, b_k)}$ is adjacent to $b_j$ vertices of $V$ modulo $2$.
 Hence every vertex of $U_{(p,b_1, \ldots, b_k)}$ is adjacent modulo $2$ to
 $p' := p+b_j \pmod 2$ 
 other vertices of $Z=U\cup V$.
 By setting this value for $p'$ we get the desired inclusion $U_{(p,b_1, \ldots, b_k)}\subseteq Z_{(p',b'_1, \ldots, b'_l)}$.
 
 To prove that that the $\mathcal P'_i$s can be completed into a contraction sequence, it remains to show that if $W\in \mathcal{P}_i\backslash\sg{U,V}$ has red neighbors $X_1, \ldots, X_k$ in $G_i$ with $k \leq d$ and red neighbors $Y_1, \ldots, Y_l$ in $G_{i-1}$ with $l \leq d$, then for every $p, b_1, \ldots, b_k\in \sg{0,1}$ there is $p', b'_1, \ldots, b'_l\in \sg{0,1}$ such that $W_{(p,b_1, \ldots, b_k)}\subseteq W_{(p',b'_1, \ldots, b'_l)}$. For this observe that the red neighbors of $W$ in $G_{i-1}$ are exactly the red neighbors of $W$ in $G_{i}$ that are distinct of $U$ and $V$ plus possibly the set~$Z$. Thus for every $j \in [l]$, if $Y_j=X_i$ for some $i \in [k]$ we let $b'_j :=b_i$.
 Moreover, we let $p':=p$. If there exists some $j \in [k]$ such that $Y_j$ is not one of the $X_i$s, then we must have $Y_j=Z$. 
 Observe that by definition of $W_{(p, b_1, \ldots, b_k)}$, every vertex of $W_{(p, b_1, \ldots, b_k)}$ has the same number $a$ of neighbors modulo $2$ in $U$. Similarly, every vertex of $W_{(p, b_1, \ldots, b_k)}$ has the same number $b$ of neighbors modulo $2$ in $V$. Hence if we let $b'_j :=  (a+b) \pmod 2$, we have the desired inclusion: $W_{(p, b_1, \ldots, b_k)}\subseteq W_{(p', b'_1, \ldots, b'_l)}$. Thus we proved that the sequence $\mathcal P'_n, \ldots, \mathcal P'_1$ can be completed into a contraction sequence.

 Now we need to show that this sequence is a $D$-contraction sequence for $G^{[2]}$. We let $G^{[2]}_i$ denote the trigraph obtained from $G^{[2]}$ when the vertex set is contracted according to the partition $\mathcal P'_i$.  
 Let $i \in [n]$ and $U \in \mathcal{P}_i$. We show how to bound the red degree of every $U_{(p, b_1, \ldots, b_k)}$ in $G^{[2]}_i$. Let $W_1, \ldots, W_k$ be the red neighbors of $U$ in $G_i$, with $k\leq d$. Let $p, b_1, \ldots, b_k\in \sg{0,1}$. 
 \begin{claim}
  \label{clm: redeg}
    The only red neighbors of $U_{(p, b_1, \ldots, b_k)}$ in $G^{[2]}_i$ are either other subsets of $U$, or subsets of some sets $V\in \mathcal{P}_i$ such that the distance between $U$ and $V$ in the red graph of $G_i$ is at most $2$.
 \end{claim}
\begin{proofofclaim}
 Let $V\in \mathcal{P}_i\backslash \sg{U}$ be such that every red path from $U$ to $V$ in $G_i$ has length at least $3$. Let $W'_1, \ldots, W'_{k'}$ be the red neighbors of $V$ in $G_i$ with $k'\leq d$, and let $p', b'_1, \ldots, b'_{k'}\in \sg{0,1}$. Our goal is to show that $U_{(p, b_1, \ldots, b_k)} V_{(p', b'_1, \ldots, b'_{k'})}$ is not a~red edge in $G^{[2]}_i$.
 In other words, we will show that for every $u \in  U_{(p, b_1, \ldots, b_k)}$ and $v \in V_{(p', b'_1, \ldots, b'_{k'})}$, the value $|N(u) \cap N(v)| \pmod 2$ does not depend on the choice of $u$ and $v$. 
 Observe first that by hypothesis, $UV$ is either a black edge or a non-edge in $G_i$. We may assume that we are in the first case, as the second one is easier to deal with. Let $u \in  U_{(p, b_1, \ldots, b_k)}$ and $v \in V_{(p', b'_1, \ldots, b'_{k'})}$. As $\mathrm{deg}_V(v) = p' \pmod 2$, we have: 
 $$|N(u)\cap N(v) \cap V| = p' -1 \pmod 2.$$ 
 By symmetry we also have: 
 $$|N(u) \cap N(v) \cap U| = p-1 \pmod 2.$$
 It remains to count the number of common vertices of $u$ and $v$ which are neither in $U$ nor in~$V$.
 There are three cases. 
 
 First if $W \in \mathcal{P}_i \backslash \sg{U,V}$ is such that both $UW$ and $WV$ form black edges of $G_i$, then we have: 
 $$|N(u) \cap N(v) \cap W| = |W|.$$
 
 Assume now that $W\in \mathcal{P}_i \backslash \sg{U,V}$ is such that $UW$ is a red edge in $G_i$ and $WV$ is a black edge in $G_i$.
 Then $W = W_j$ for some $j \in [k]$.
 In particular, by definition of $U_{(p, b_1, \ldots, b_k)}$, $u$ has $b_j$ neighbors in $W$ modulo $2$, and each of them is a neighbor of $v$. Thus we have: 
 $$|N(u) \cap N(v) \cap W_j| = b_j \pmod 2.$$
 
 The last kind of common neighbors of $u$ and $v$ we can find are vertices from some $W\in \mathcal{P}_i \backslash \sg{U,V}$ where $UW$ is a black edge in $G_i$ and $WV$ is a red edge in $G_i$.
 By the same symmetric reasoning, we get that:
 $$|N(u) \cap N(v) \cap W'_j| = b'_j \pmod 2.$$
 
 Thus we have: 
 $$|N(u)\cap N(v)| = \left(p + p' + \sum_{j=1}^k b_j + \sum_{j=1}^{k'} b'_j\right) \pmod 2,$$
 
 a value that does not depend on the choice of $u$ and $v$. This is equivalent to say that $U_{(p, b_1, \ldots, b_k)} V_{(p', b'_1, \ldots, b'_{k'})}$ does not form a red edge in $G^{[2]}_i$.  
\end{proofofclaim}
 
 By~\cref{clm: redeg} we get that the only red neighbors of $U_{(p, b_1, \ldots, b_k)}$ are either other subsets of~$U$, or subsets of some $V \in \mathcal{P}_i \backslash \sg{U}$ such that $UV$ is a red edge in $G_i$, or subsets of some $V \in \mathcal{P}_i \backslash \sg{U}$ such that there exists some $W \in \mathcal{P}_i \backslash \sg{U,V}$ so that both $UW$ and $VW$ form red edges in $G_i$. In total this gives us at most:

 $$2^{d+1}-1 + d2^{d+1} + d^2 2^{d+1} 
 =(d^2 + d +1)2^{d+1} - 1$$
 red neighbors of $U$ in $G^{[2]}_i$, which is the desired result.   
 \end{proof}

\begin{remark}
\label{rem: G2}
 Observe that if we replace the condition '$|N(x)\cap W_i| \text{ is odd}$' by '$|N(x)\cap W_i| \geq 1$' in the definition of the refinement $\mathcal P'_n, \ldots, \mathcal P'_1$, then using similar arguments than the one from the previous proof, we get a $\Oo(d^22^d)$-sequence for $G^2$. 
 \end{remark}
 
\paragraph*{Computation of $\mathcal{T'}$}
 
Assume that we are given a twin-decomposition $(\mathcal T, \mathcal B)$ of $G$, associated to a $d$-sequence $G_n, \ldots, G_1$ of $G$ with associated partition sequence $\mathcal{P}_n, \ldots, \mathcal{P}_1$. We describe a way to compute a tree $\mathcal{T'}$ corresponding to a $D$-contraction sequence of $G^{[2]}$ of the form described in~\cref{lem: refinement}, with $D:=(d^2+d+1)2^{d+1}-1$. The tree $\mathcal{T'}$ will be constructed in $n-1$ steps, and step $i$ (for $i$ going from $n-1$ down to $1$) is based on $G_i$ (that is computed on the fly).
We give in \cref{fig:tree} an example of the construction of $\mathcal{T'}$ for the graph of~\cref{fig:G2}. 
 

We assume that $\mathcal B$ is lifted up.
Recall this is always possible up to a preprocessing of time $\Oo(dn)$ by~\cref{rem: high}. 
In time $\mathcal{O}(n)$, we compute for each node $U\in V(\mathcal{T})$ a parity bit $s_U \in \sg{0,1}$ such that $s_U = |U|\pmod 2$ (recall that $U$ is naturally identified to a subset of $V(G)$).


Recall that by \cref{lem: listeadj}, we can dynamically compute in time $\Oo(dn)$ the lists $L_U^i$ of the red neighborhoods of vertices $U\in V(G_i)$. 
Though we present it separately, the algorithm computing $\mathcal T'$ and the the set of partial edges $\mathcal{B}_1$ will be dynamic, and all the computations in step $i$ are done at the same time, so we assume that at step $i$ that we always know the lists $L^i_U$ and the parity bits $s_U$.

Our goal is now to find a way to compute the refinement we described previously.
For this we define for every $i\in [2,n]$, every $U\in V(G_i)$ such that $L^i_U=\sg{W_1, \ldots, W_k}$ and every $p, b_1, \ldots, b_k\in \sg{0,1}$ the at most $2^{d+1}$ pointers $q^i[U,p,b_1,\ldots,b_k]$ by:
$$q^i[U,p,b_1,\ldots,b_k]:=(V,p', b_1', \ldots, b_l'),$$  
where $(v,p', b_1', \ldots, b_l')$ is the (unique, if it exists) tuple such that 
$U_{(p, b_1, \ldots, b_k)}\subseteq V_{(p',b_1', \ldots, b_l')}$, with $U\in V(G_i)$ and $V\in V(G_{i-1})$. Observe that $U_{(p, b_1, \ldots, b_k)}$ may be possibly empty, in which case we can set $q^i[U,p,b_1,\ldots,b_k]:=\emptyset$.

We show that these pointers can be dynamically computed: 
\begin{lemma}
 \label{lem: pointers}
 The values $q^i[U,p,b_1,\ldots,b_k]$ can be dynamically computed for decreasing values of $i$ from $n$ down to 2. Moreover, each step takes time $\mathcal O(d 2^d)$.
\end{lemma}

Here, \cref{lem: pointers} has to be understood as \cref{lem: listeadj}: we are able to visit the nodes of $\mathcal T$ for decreasing values of their labels $i\in [n-1]$, and update accordingly $\mathcal{O}(n)$ pointers $q[U,p,b_1, \ldots, b_k]$ such that at step $i$, the value of $q[U,p,b_1, \ldots, b_k]$ is exactly $q^i[U,p,b_1, \ldots, b_k]$ for every $U \in V(G_i)$ and $p,b_1, \ldots, b_k\in \sg{0,1}$. Hence we always keep only a linear number of values $q[U,p,b_1,\ldots,b_k]$, and show that, at each step, we only change a bounded number of them (namely $\Oo(d2^d)$).

\begin{proof}
 As we will define $q^{i-1}$ from $q^{i}$, we will also need to define for the initialization  $q^{n+1}$. For this, as
 $G_n=G$, $\mathcal{P}_n$ is the partition of $V(G)$ into singletons and there are no red edges, we set for every $v\in V(G)$ and $p\in \sg{0,1}$:
 $$q^{n+1}[\sg{v},p]:=\begin{cases}
	(\sg{v},0)&\text{if $p=0$}\\
	\emptyset&\text{if $p=1$}	
\end{cases} $$
This will help us in what follows to only deal with nonempty sets $U_{(p,b_1,\ldots,b_k)}$, as the empty ones will be the ones having no antecedent by $q^i$.

Now let $i\in [2, n]$ and assume that the contraction between $G_i$ and $G_{i-1}$ is $(U,V,Z)$, with $U,V\in V(G_i)$ and $Z=U\cup V\in V(G_{i-1})$.
Let $W\in V(G_i)$. We assume that $L^i_W=\sg{W_1, \ldots, W_k}$ is known. 
If $W\notin\sg{U,V}$, then we distinguish four different cases. 
If $WU\in E(G_i)$ and $WV\in E(G_i)$ or if $WU\notin E(G_i)\cup R(G_i)$ and $WV\notin E(G_i)\cup R(G_i)$, then the adjacencies (or non-adjacencies) of $W$ in $G_{i+1}$ are the same as in $G_i$, thus for every $p, b_1,\ldots,b_k \in \sg{0,1}$, we set $q^{i-1}[W,p, b_1,\ldots,b_k]:=q^i[W,p, b_1,\ldots,b_k]$. 
In fact this implies that the only pointers we will update are those related to vertices of $V(G_{i-1})$ which are either linked by a red edge to $U$ or $V$, or $U$ and $V$ themselves.

If $WU\in E(G_i)$ and $WV\notin E(G_i)\cup R(G_i)$, then we have $L^{i-1}_W=L^i_W\cup \sg{Z}=\sg{W_1, \ldots, W_k, W_{k+1}:=Z}$, so we set for every $p, b_1,\ldots,b_k \in \sg{0,1}$: 
$$q^i[W,p, b_1,\ldots,b_k]:=
	(W, p, b_1,\ldots,b_k,s_U).$$
By hypothesis on the placement of the transversal edges of $\mathcal B$, we can detect this case, which happens when there is an edge of $\mathcal B$ between the node of $\mathcal T$ corresponding to $U$ and $W$, and where $V$ is neither in $L^i_W$, nor linked to $W$ by an edge of $\mathcal B$.
Of course the case where $UW\notin E(G_i)$ and $UV\in E(G_i)$ is symmetric.  

If $UW\notin E(G_i)\cup R(G_i)$ and $UV\in R(G_i)$, then assume without loss of generality that $W_k=V$ and let $W'_k:=Z$.
$L^{i-1}_W=(L^i_W\backslash\sg{V})\cup \sg{Z}$ and we set for every $p, b_1,\ldots,b_k \in \sg{0,1}$:
$$q^i[W,p, b_1,\ldots,b_k]:=(W,p, b_1,\ldots, b_k).$$
 
The case when $UW\notin E(G_i)\cup R(G_i)$ and $UV\in R(G_i)$ is similar, as we also have $L^{i-1}_W=(L^i_W\backslash\sg{W_k})\cup \sg{W'_k}$, so if again we assume that $W_k=V$, we just need to set:
$$q^i[W,p, b_1,\ldots,b_k]:=(W,p, b_1,\ldots, b_k+p_U).$$

The last case satisfying $W\notin \sg{U,V}$ is when we have both $UW\in R(G_i)$ and $VW\in R(G_i)$. Assume for simplicity that $W_{k-1}=U$ and $W_k=V$, and set $W'_{k-1}:=Z$. Then we have $L_W^{i-1}=(L_W^i\backslash \sg{W_{k-1}, W_k})\cup \sg{W'_{k-1}}$ so we let: 
$$q^i[W,p, b_1,\ldots, b_{k-2},b_{k-1} ,b_k]:=(W,p, b_1,\ldots, b_{k-2}, b_{k-1}+b_k).$$

Again, all these configurations are easy to detect thanks to our assumptions on $\mathcal B$ and the fact we know $L^i_W$.

Assume now that $W=U$, and let $L_U^i=\sg{W_1, \ldots, W_k}$ and $p, b_1, \ldots, b_k\in \sg{0,1}$. We explain now how to compute the value $q[U, p, b_1, \ldots, b_k]$.
Observe first that if a vertex of $U$ has degree $p$ in $U$ modulo $2$, then we can compute $p'$, its degree modulo $2$ in $Z=U\cup V$ by setting: 
$$p'=\begin{cases}
	p&\text{if $UV \notin E(G_i)\cup R(G_i)$}\\
	p+s_V&\text{if $UV\in E(G_i)$}\\
	p+b_l &\text{if $UV\in R(G_i)$ and $V=W_l$}\\
\end{cases}.$$
We now describe the different types of red adjacencies that can exist between $Z$ and other vertices in $G_{i-1}$, and show for each of them how the associated parity bit $b'$ can be computed. More precisely, if $u\in U_{(p, b_1, \ldots, b_k)}$, we show how to compute $b':=deg_Z(u)$. 
Let $Y\in V(G_{i-1})$ be a red neighbor of $Z$ in $G_i$. Then $Y$ can be of five different types, and we show how to compute the red parity $b'$ modulo $2$ of the different vertices of $U\subseteq Z$ in $G_{i-1}$ according to the one they had in $G_i$.
First, observe that if $Y$ was already a red neighbor of $U$ in $G_i$, i.e., if $Y=W_j\in L_U^i \cap L_Z^{i-1}$ for some $j$, then we simply let $b':=b_j$.
Assume now that $Y\in L_V^i\backslash L_U^i$. Then either there is no edge between $Y$ and $U$ in $G_i$, in which case we let $b':=0$, or $YU\in E(G_i)$ and we let $b':=s_Y$. The last two cases are when a red edge is created between $Z$ and $Y$ in $G_{i-1}$, but $Y$ is neither a red neighbor of $U$ nor of $V$. Then again either $UY\notin E(G_i)$ and $VY\in E(G_i)$, so we let $b':=0$, or $UY\in E(G_i)$ and $VY\notin E(G_i)$ so we let $b':=s_Y$. 

We can identify all these cases as we know $L_U^i$, $L_V^i$, $L_Z^{i-1}$, and by the hypotheses we put on $\mathcal B$.
Now for the complexity at step $i$, as the only $W\in V(G_i)\backslash\sg{U,V}$ such that some $q[W,p,b_1,\ldots,b_k]$ are updated are the ones such that $W$ is a red neighbor of $U$ or $V$, it means that at most $2d2^{d+1}$ pointers $q[W,p,b_1,\ldots,b_k]$ are modified at step $i$. One can check that each of these modifications can be done in time $\Oo_d(1)$, as well as the listing of all such $W$. Hence we deduce the desired overall complexity for the running time at each step.
\end{proof}

Now observe by \cref{lem: listeadj,lem: pointers} that we have everything in hand to compute $\mathcal{T'}$. Indeed, we can compute dynamically the pointers $q[U,p, b_1, \ldots, b_k]$, and keep only the one associated to sets $U_{(p,b_1,\ldots,b_k)}$ that are nonempty. The nodes of $\mathcal{T'}$ will be the set of every nonempty set $U_{(p,b_1,\ldots,b_k)}$ we encounter. One has to take care as the set $U_{(p,b_1,\ldots,b_k)}$ can be equal to some other set $U'_{(p',b'_1 ,\ldots, b'_{k'})}$. However they will correspond to the same node of $\mathcal{T'}$. If at some step $i$ we have $q^i[U,p,b_1,\ldots,b_k]=q^i[p',b'_1 ,\ldots, b'_{k'}]=(Z,p_Z,c_1,\ldots,c_t)$, then we choose the node $Z_{(p_Z,c_1,\ldots,c_t)}$ as the common parent of $U_{(p,b_1,\ldots,b_k)}$ and $U'_{(p',b'_1 ,\ldots, b'_{k'})}$.
To keep a binary tree, one can replace the edges from $Z_{(p_Z,c_1,\ldots,c_t)}$ to its (more than 2) children by a binary tree whose leaves are these children and root is $Z_{(p_Z,c_1,\ldots,c_t)}$.
This creates a bounded number of additional nodes in $\mathcal{T'}$. 

In the end of the last step when $i=1$, we may end with two disjoint rooted trees whose leaves correspond to the nodes of even degree and to the nodes of odd degree.
In this case, we add an additional node that corresponds to $V(G)$ to get a rooted tree.
It remains to label the nodes of $\mathcal{T'}$.
By~\cref{lem: refinement}, one can take any total order which is consistent with the one of the creation of the nodes in the proof (i.e., if $i<j$, we label every node associated to a set that was created in step $i$ with a label smaller than those associated to a set created in step $j$).

\begin{figure}[h!]
  \centering
  \begin{tikzpicture}[scale=0.9]

    \begin{scope}[xshift=0cm, yshift=-1cm, scale=0.7]
      \foreach \i/\l in {1/6,2/7,3/1,4/5,5/2,6/3,7/4}{
        \node[draw,circle,inner sep=0.04cm,minimum size=0.5cm] (l\l) at (\i,0) {$\l$} ;
      }
      \foreach \i/\j/\l in {1.5/1.1/6,3.5/1.1/5,5.5/1.1/4, 4.5/2.2/3, 6.5/3.3/2, 4.5/4.4/1}{
        \node[draw,rectangle,inner sep=0.04cm,minimum size=0.5cm] (v\l) at (\i,\j) {$\l$} ;
      }
     \foreach \i/\j in {6/6,7/6, 1/5,5/5, 2/4,3/4, 4/2}{
       \draw[<-] (l\i) -- (v\j) ;
     }
     \foreach \i/\j in {6/1,5/3,4/3,3/2,2/1}{
       \draw[<-] (v\i) -- (v\j) ;
     }
     \foreach \i/\j in {6/7,1/5}{
       \draw[line width=0.05cm,blue] (l\i) to [bend right] (l\j) ;
     }
     \foreach \i/\j in {6/1,4/4}{
       \draw[line width=0.05cm,blue] (v\i) to [bend left] (l\j) ;
     }
     \draw[line width=0.05cm,blue] (v5) to [bend right] (v4) ;     
     \end{scope}

     \begin{scope}[xshift=6cm, yshift=-1cm]
      \foreach \i/\l in {1/6,2/7,3/1,4/5,5/2,6/3,7/4}{
        \node[draw,circle,inner sep=0.04cm,minimum size=0.5cm] (l\l) at (\i,0) {$\l$} ;
      }
      \foreach \i/\j/\l/\n in {1.5/1.1/$67_{\textcolor{green!60!black}{1}}^{\textcolor{red}{6}}$/1,3.5/1.1/$15_{\textcolor{green!60!black}{1},0}^{\textcolor{red}{5}}$/2,5.5/1.1/$23_{\textcolor{green!60!black}{0}}^{\textcolor{red}{4}}$/3, 3.5/2.2/$1235_{\textcolor{green!60!black}{1},0,0}$/4, 5.5/2.2/$1235_{\textcolor{green!60!black}{0},0,1}$/5, 4.5/3.3/$12345_{\textcolor{green!60!black}{1},0}^{\textcolor{red}{3}}$/6, 7/3.3/$12345_{\textcolor{green!60!black}{0},0}$/7, 3.5/4.4/$1234567_{\textcolor{green!60!black}{0}}^{\textcolor{red}{2}}$/8, 5.5/4.4/$1234567_{\textcolor{green!60!black}{1}}$/9, 4.5/5.5/$r^{\textcolor{red}{1}}$/10}{
        \node[draw,rectangle,inner sep=0.04cm,minimum size=0.5cm] (v\n) at (\i,\j) {\l} ;
      }
     \foreach \i/\j in {6/1,7/1, 1/2,5/2, 2/3,3/3, 4/7}{
       \draw[<-] (l\i) -- (v\j) ;
     }
     \foreach \i/\j in {2/4,3/5,4/6,5/6,1/8,7/8,8/10,6/9,9/10}{
       \draw[<-] (v\i) -- (v\j) ;
     }     
     \end{scope}
     
  \end{tikzpicture}
  \caption{Left: A twin-decomposition of the graph $G$ of \cref{fig:G2}. Right: The tree $\mathcal{T'}$ associated to $(\mathcal T, \mathcal B)$. We subdivided some edges of $\mathcal{T'}$ to make clear how the refinement works. The labels of the nodes of $\mathcal{T'}$ are given by the red exponents. The first bit (in green) of the tuples in subscript correspond to the bits $p$, and the remaining ones to the $b_i$s.}
\label{fig:tree}
\end{figure}
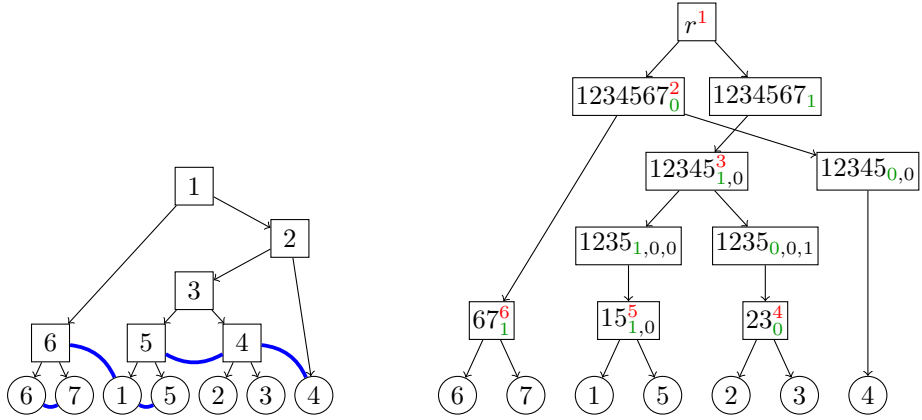

\paragraph*{Computation of the set of labeled transversal edges $\mathcal{B}_1$}

We now compute exactly the adjacencies in $G^{[2]}$.  
As previously announced, we first need to compute a set $\mathcal{B}_1$ of transversal edges of $\mathcal{T'}$ with linear size in $n$, together with for every edge $e=UV\in \mathcal{B}_1$ a parity label $\beta_e\in \sg{0,1}$. 
We will ensure that $\bb_1$ has linear size in $n$.
Moreover we will also compute for each node $U\in V(\tree')$ a parity bit $\alpha_U\in \sg{0,1}$.
In the end, we wish to recover, with $\bb_1$ and the bits $\alpha_U$, all the adjacencies of $G^{[2]}$.
More precisely, we show that for every $u,v\in V(G)$, $uv\in E(G^{[2]})$ if and only if the sum of the labels $\beta_e$ over all edges $e\in \mathcal{B}_1$ between the branches of $\mathcal{T'}$ having $\sg{u}$ and $\sg{v}$ as children plus the sum of the parity bits $\alpha_U$ over all nodes $U$ that are both predecessor of $\sg{u}$ and $\sg{v}$ is odd. 
%

We again visit the nodes of $\mathcal{T'}$ by decreasing values of their labels.
In fact this step can be done at the same time as the previous one (construction of $\mathcal{T'}$), so we assume that whenever the node we are visiting was created at step $i$ of the contraction sequence of $\mathcal{T}$, we have access to the lists of red neighbors $L^i_U$ for every $U\in V(G_i)$ and to the correspondence between the nodes of $\mathcal{T'}$ and the partition of each $U\in V(G_i)$ into $U_{(p,b_1,\ldots,b_k)}$.
We progressively form $\mathcal{B}_1$, starting from the empty set.
In the end, if we denote with $\prec_{\tree'}$ the partial order over the nodes of $\tree'$ defined by $U'\prec_{\tree'}U$ if and only if $U'$ is an ancestor of $U$ in $\tree'$, then the adjacency between every two vertices $u$ and $v$ in $G^{[2]}$ is determined by the parity of the sum:
\begin{equation}\label{prop:b1} 
S_{uv} = \sum_{\substack{W\in V(\tree')\\
W\prec_{\tree'} \sg{u} \\
           W\prec_{\tree'} \sg{v}}}\alpha_W 
+\sum_{\substack{UV\in \mathcal{B}_1 \\
           U\prec_{\tree'} \sg{u} \\
           V\prec_{\tree'} \sg{v}}}\beta_{UV}.
\end{equation}



See \cref{fig:B1} for an illustration of the algorithm explained below.

\begin{lemma}
 \label{lem: partial-edges}
 One can dynamically compute in time $\Oo(d^2 4^dn)$ a set of labeled edges $\mathcal{B}_1$ of size $\Oo(d^2 4^dn)$ and parity bits $\alpha_U$  
 such that \cref{prop:b1} holds.
\end{lemma}


\begin{proof}
 As mentioned above, we assume that we have access to the red adjacencies in $\mathcal{T}$ through the lists $L_U^i$, as well as the parity bits $s_U:=|U|\pmod 2$ and the correspondence between the nodes of $\mathcal{T'}$ and the associated subsets in the refinement described in~\cref{lem: refinement}.
 We still assume that $\mathcal{B}$ is lifted up.
 Moreover, observe that it is easy to compute for every node $U'$ of $\tree'$ the parity of its cardinality $s'_{U'}:=|U'|\pmod 2$.
 
 Our main goal is to count (modulo $2$) the number of common neighbors between every pair of distinct vertices $u,v\in V(G)$.
 For such a pair, we distinguish different types of common vertices, and show that we take every of them into account either with a bit $\alpha_U$ or with a label $\beta_e$ for some $e\in \mathcal{B}_1$.
 We visit the nodes of $\mathcal{T}$ by decreasing order of their labels.
 We start with $\mathcal{B}_1:=\emptyset$ and for each node $U'\in \tree'$, $\alpha_{U'}:=0$.
 The general idea is that at step $i$, we have a partial knowledge of $G_i$: we know all the red adjacencies, thanks to the lists $L^i_U$, and we know some of the black edges thanks to $\bb$.
 At step $i$, we consider every black edge $e$ of $G_i$ that ``disappears'' in $G_{i-1}$, and compute the contribution of every path of length 2 in $G$ containing at least one of the edges of $e$ (recall that we can see $e\in E(G_i)$ as a set of edges of $E(G)$).
 Doing so, we consider each edge of $G$ only once, which prevents us from double counting.

Assume that the $i$-th contraction of the $d$-sequence described by $\tree$ is corresponds to the triple $(U,V,Z)$ with $U,V\in V(G_i)$ and $Z:=U\cup V\in V(G_{i-1})$. Then for every $W\in V(G_i)$ ($W$ can possibly be $V$) such that $UW\in \bb$, we do the following:

\begin{enumerate}[$(i)$]
 \item\label{it:case1} For every node $U_{(p,b_1, \ldots, b_k)}\subseteq U$, we replace the value of $\alpha_{U_{(p,b_1, \ldots, b_k)}}$ by 
 $\alpha_{U_{(p,b_1, \ldots, b_k)}}+s_W$. Similarly, for every node $W_{(p,b_1,\ldots,b_k)}\subseteq W$ we replace $\alpha_{W_{(p,b_1,\ldots,b_k)}}$ by $\alpha_{W_{(p,b_1,\ldots,b_k)}}+s_U$.

 \item\label{it:case2} We add in $\mathcal{B}_1$ between every two distinct nodes $U_{(p,b_1, \ldots, b_k)}\subseteq U$ and $U_{(p',b'_1, \ldots, b'_k)}\subseteq U$ an edge $e$ labeled by $\beta_e:=s_{W}$. Similarly we add in $\mathcal{B}_1$ between every two distinct nodes $W_{(p,b_1, \ldots, b_k)}$ and $W_{(p,b_1, \ldots, b_k)}$ an edge $e$ labeled by $\beta_e:=s_{U}$.
 (Together with the previous item, this is when we count the contribution of the paths of the form $U-W-U$ or $W-U-W$.)
 
 \item\label{it:case3} We add in $\mathcal{B}_1$ between every two distinct nodes $U_{(p,b_1, \ldots, b_k)}\subseteq Z$ and $W_{(p',b'_1, \ldots, b'_{k'})}\subseteq W$ an edge $e$ labeled by $\beta_e:=p+p' \pmod 2$.
   Informally these edges count the contribution of the paths of the form $U-W-W$ or $U-U-W$.
 
 \item\label{it:case4} For every $W'\in V(G_i)$ with $W'\neq W$ such that $UW' \in \bb$ ($W'$ is possibly equal to $V$), we add in $\bb_1$ between every two distinct nodes $W_{(p,b_1, \ldots, b_k)}\subseteq W$ and $W'_{(p',b'_1, \ldots, b'_{k'})}\subseteq W'$ an edge $e$ labeled by $\beta_e:=s_U$.
   These edges count the contribution of the paths of the form $W-U-W'$.
   (Together with the four previous cases, we counted so far the contribution of every $2$-path that takes only black edges of $G_i$ and ``disappearing'' in $G_{i-1}$.)

 
 \item\label{it:case5} For every $X\in L^i_U$, we add in $\bb_1$ between every two distinct nodes $X_{(p,b_1, \ldots, b_k)}\subseteq X$ and $U_{(p',b'_1, \ldots, b'_{k'})}\subseteq U$ an edge $e$ labeled by $\beta_e:=b_j$, where $j$ is the index of $U$ in $L^i_X$.
   These edges count the contribution of the paths of the form $X\textcolor{red}{\sim}U-W$.

 \item\label{it:case6} For every $X\in L^i_W$, we add in $\bb_1$ between every two distinct nodes $U_{(p,b_1, \ldots, b_k)}\subseteq U$ and $X_{(p',b'_1, \ldots, b'_{k'})}\subseteq X$ an edge $e$ labeled by $\beta_e:=b'_j$, where $j$ is the index of $W$ in $L^i_X$.
   These edges count the contribution of the paths of the form $U-W\textcolor{red}{\sim}X$. 

%
%
%
\end{enumerate}

Note that in this procedure, we can add multiple labeled edges for a same pair of nodes of $\tree'$. In this case, we replace all of them by a single edge whose label is the sum of all their labels.
Observe that at step $i$, we add at most $(1 + (d+1) + d(d+1) + d + d)2^{2d}=\Oo(d^2 4^d)$ labeled edges (the number of candidates for $W$ is at most $d+1$), hence $\mathcal{B}_1$ eventually has linear size in $n$.

We show that when the procedure terminates, the sum $S_{uv}$ written above is odd for every pair of adjacent vertices in $G^{[2]}$, and even for the other pairs. For this, consider a pair $u,v\in V(G)$ of distinct vertices of $G$. Our goal is to show that every $w\in N(u)\cap N(v)$ has a contribution of exactly $1$ in the sum $S_{uv}$.
In other words, we show that every path of length 2 of the form $u-w-v$ is counted exactly once in $S_{uv}$.
We say that the edge $UV$ \emph{disappears at step $i$} if $U, V\in V(G_i)$, $UV\in E(G_i)$ and $U'V'\notin E(G_{i-1})$, where $U', V' \in V(G_{i-1})$ are such that $U\subseteq U'$ and $V\subseteq V'$ (in particular it is the case when $U\cup V\in V(G_{i-1})$). Let $i\in [n]$, $w\in N(u)\cap N(v)$ and $U, V$ and $W$ be the sets of $V(G_i)$ that respectively contain $u,v$ and $w$ (they may be equal).  
Observe that for $i=n$ the edges $UW$ and $WV$ are in $E(G_n)$, and then disappear at some step $i$. Assume without loss of generality that $UW$ disappears after $VW$ does, and take $i$ minimal such that $UW\in E(G_i)$. We distinguish the following disjoint cases:
\begin{itemize}
 \item Assume first that both $UW$ and $VW$ disappear at step $i$, i.e., we also have $VW\in E(G_i)$. 
 \begin{itemize}
 \item Assume that $U=V$. As $UW$ disappears at step $i$, it means that $UV\in \bb$, either because $U$ is merged with some other set, or because $W$ does. T
   hen at step $i$ of our algorithm, we count the paths $u-w-v$ either in \pref{it:case1}, or in~\pref{it:case2} according to whether or not $u$ and $v$ lie in the same node of $\tree'$.
 \item If $U\neq V$, then it means that $W$ is merged with another set in $G_{i-1}$.
   Let $W'\in V(G_{i-1})$ be the set that contains $W$ as a subset.
  Then we have $UW'\in R(G_{i-1})$ and $VW'\in R(G_{i-1})$, and $UW, VW\in \bb$.
  In particular this means that we counted $u-w-v$ in \pref{it:case4} (with $W'$ playing the role of $Z$).
 \end{itemize}
 \item Now assume that $V=W$, i.e., that $VW$ disappeared because $VW$ were merged together.
 \begin{itemize}
  \item Assume that $UW$ disappears because $U$ and $W$ are merged together in $G_{i-1}$. Then the path $u-w-v$ is counted in \pref{it:case2}.
  \item Assume now that $U$ is merged with another set and contained in $U'\in V(G_{i-1})$.
    Then again, we claim that $u-w-v$ is counted in \pref{it:case2}.
    The case when $W$ is merged with another set in $G_{i-1}$ is symmetric. 
 \end{itemize}

 \item Now we assume that $U,V$ and $W$ are pairwise disjoint, and that $VW\in R(G_i)$.
 \begin{itemize}
  \item Assume first that $U$ and $W$ are merged into $Z:=U\cup W$ in $G_{i-1}$. Then as $V$ is a red neighbor of $W$ in $G_i$, observe that the path $u-w-v$ is counted in \pref{it:case6}.
  \item Assume now that $V$ and $W$ are merged into $Z:=V\cup W$ in $G_{i-1}$.
    Then as $V$ is a red neighbor of $W$ in $G_i$, observe that the path $u-w-v$ is counted in \pref{it:case1}.
  \item If instead $U$ and $V$ are merged into $Z:=U\cup V$ in $G_{i-1}$.
  Then again, the path $u-w-v$ is counted in \pref{it:case6}. 
  \item Eventually assume that the sets containing $u,v$ and $w$ are still disjoint in $G_{i-1}$, and call them respectively $U', V'$ and $W'$. Then $U'W', V'W' \in R(G_{i-1})$. Observe that as $UW$ disappeared, we have $V'=V$ and exactly one of the two equalities $W=W'$ and $U=U'$ holds. If $U\neq U'$ and $W=W'$, then observe that the path $u-w-v$ was counted in \pref{it:case6} (with $U'$ playing the role of $Z$), while if $U=U'$ and $W\neq W'$, it was counted in \pref{it:case5}. 
 \end{itemize}
\end{itemize}

Now it only remains to show that $w$ is counted in only one item of our algorithm. For this, observe that there is a unique $i$ such that the edge $uw$ is included in a black edge $UW\in \bb$.
We are done as every path $u-w-v$ is only counted in our algorithm at the step $i$ when $UW$ disappears (if it disappears before $VW$). 
\end{proof}

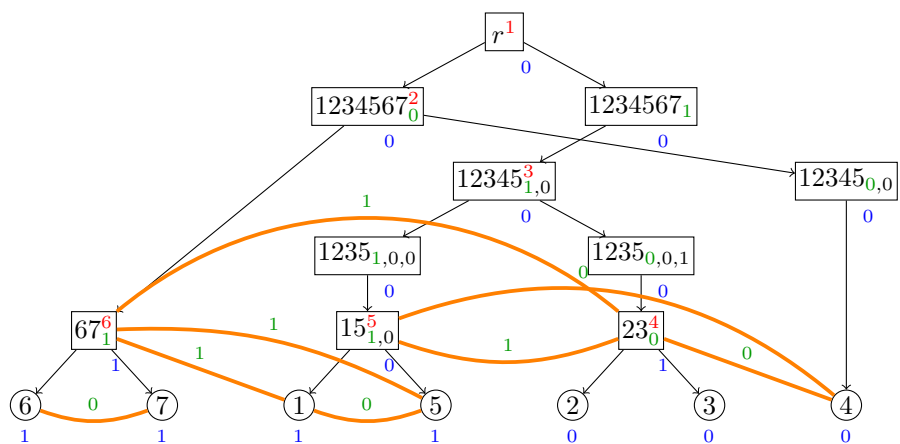
\begin{figure}[h!]
  \centering
  \begin{tikzpicture}[scale=0.9]
      \foreach \i/\l/\p in {1/6/1,2/7/1,3/1/1,4/5/1,5/2/0,6/3/0,7/4/0}{
        \node[draw,circle,inner sep=0.04cm,minimum size=0.4cm, label={[label distance=0.001cm]below:$_{\textcolor{blue}{\p}}$}] (l\l) at (\i*2,0) {$\l$} ;
      }
      \foreach \i/\j/\l/\n/\p in {1.5/1.1/$67_{\textcolor{green!60!black}{1}}^{\textcolor{red}{6}}$/1/1,3.5/1.1/$15_{\textcolor{green!60!black}{1},0}^{\textcolor{red}{5}}$/2/0,5.5/1.1/$23_{\textcolor{green!60!black}{0}}^{\textcolor{red}{4}}$/3/1, 3.5/2.2/$1235_{\textcolor{green!60!black}{1},0,0}$/4/0, 5.5/2.2/$1235_{\textcolor{green!60!black}{0},0,1}$/5/0, 4.5/3.3/$12345_{\textcolor{green!60!black}{1},0}^{\textcolor{red}{3}}$/6/0, 7/3.3/$12345_{\textcolor{green!60!black}{0},0}$/7/0, 3.5/4.4/$1234567_{\textcolor{green!60!black}{0}}^{\textcolor{red}{2}}$/8/0, 5.5/4.4/$1234567_{\textcolor{green!60!black}{1}}$/9/0, 4.5/5.5/$r^{\textcolor{red}{1}}$/10/0}{
        \node[draw,rectangle,inner sep=0.04cm,minimum size=0.5cm,label={[label distance=0.001cm]-70:$_{\textcolor{blue}{\p}}$}] (v\n) at (\i*2,\j) {\l} ;
      }
     \foreach \i/\j in {6/1,7/1, 1/2,5/2, 2/3,3/3, 4/7}{
       \draw[<-] (l\i) -- (v\j) ;
     }
     \foreach \i/\j in {2/4,3/5,4/6,5/6,1/8,7/8,8/10,6/9,9/10}{
       \draw[<-] (v\i) -- (v\j) ;
     }     
     
     \foreach \i/\j in {6/7,1/5}{
       \draw[line width=0.05cm,orange] (l\i) to [bend right=20] node[midway, below, label={above=1mm:$_{\textcolor{green!60!black}{0}}$}]{} (l\j) ;
     }
     
     \foreach \i/\j/\l/\e in {1/5/1/15,1/1/1/0,3/4/0/0}{
       \draw[line width=0.05cm,orange] (v\i) to [bend left=\e] node[midway, below, label={ $_{\textcolor{green!60!black}{\l}}$}]{} (l\j) ;
     }

     \foreach \i/\j/\l/\d/\e in {1/3/1/left/40,2/3/1/right/20}{
       \draw[line width=0.05cm,orange] (v\i) to [bend \d=\e] node[midway, below, label={ $_{\textcolor{green!60!black}{\l}}$}]{} (v\j) ;
     }     
     \draw[line width=0.05cm,orange] (v2) to [bend left] node[pos=0.4, below, label={ $_{\textcolor{green!60!black}{0}}$}]{} (l4) ;
     
  \end{tikzpicture}
  \caption{The result of the algorithm from \cref{lem: partial-edges} applied to the twin-decomposition of \cref{fig:tree}. The blue bits are the values of the $\alpha_U$s and the orange labeled edges represent the set $\bb_1$. One can check that the sums $S_{uv}$ give the adjacencies of $G^{[2]}$ pictured in \cref{fig:G2}.}
\label{fig:B1}
\end{figure}

\paragraph*{Computation of the set $\mathcal{B}'$ of transversal edges of $\mathcal{T'}$}

In the previous sections, we built a tree $\tree'$ together with a set of labeled $0,1$-edges~$\bb_1$ whose vertices are labeled by Booleans $\alpha_U$ for every $U\in V(\tree')$, such that for every two vertices $u,v\in V(G)$, $uv$ form an edge of $G^{[2]}$ if and only if the sum of the labels of the edges of $\bb_1$ joining the branch having $\sg{u}$ and the branch having $\sg{v}$ as leaves plus the sum of the labels $\alpha_W$ of the internal nodes $W$ that are common ancestors of $\sg{u}$ and $\sg{v}$ in $\tree'$ is odd.
We denote by $G_n^{[2]}, \ldots, G_1^{[2]}$ the contraction sequence of $G^{[2]}$ that corresponds to the tree $\mathcal{T'}$.

We now compute a set of transversal edges $\bb'$ such that $(\tree', \bb')$ is a twin-decomposition of $G^{[2]}$.
We proceed in two steps. 
\begin{enumerate}[$(1)$]
\item First, we visit the nodes of $\tree'$ by decreasing values of their labels (i.e., from bottom to top), and try to lift up the highest we can the information they bear. This way, we modify $\bb_1$, by removing the edges which were ``too high'' in $\tree'$ to be lifted up, and construct little by little another set of labeled transversal edges $\bb_2$ of $\tree'$ corresponding to the edges we ``blocked'', i.e., that could not be pushed.
  We will prove that these edges correspond exactly to the black and non-edges of $G^{[2]}_i$ that disappear in $G^{[2]}_{i-1}$.
  Thus they correspond to the only possible place for the edges of $\bb'$.
  It remains to decide their status in $G^{[2]}_i$ (i.e., whether they correspond to edges or non-edges in $G^{[2]}$).
\item In the last step, we visit the nodes of $\tree'$ this time by increasing values of their labels (i.e., from top to bottom).
  The goal is to communicate the contribution of the labels $\alpha_U$ and of the edges that were not removed yet from $\bb_1$ to the lower strata of the tree $\tree'$, in order to decide the status of each edge of $\bb_2$.
\end{enumerate}
Again we will only add or remove a constant number of edges from $\bb_1$ and $\bb_2$ at each step of (1) and (2), which ensures that both these sets always have linear size in $n$, and that our total computation time is also linear.

\begin{lemma}
\label{lemma: b1tob2}
One can dynamically compute bits $\alpha'_U$ for each $U\in \mathcal{T'}$ together with a set $\bb_2$ of labeled edges and remove edges of $\bb_1$ 
such that we eventually have the following properties:
 \begin{itemize}
  \item For every pair of distinct vertices $u,v\in V(G)$, the sum $S_{uv}$ is equal to the sum $S'_{uv}$ defined by:
  $$S'_{uv}:= \sum_{\substack{W\in V(\tree')\\
W\prec_{\tree'} \sg{u} \\
           W\prec_{\tree'} \sg{v}}}\alpha_W
+\sum_{\substack{UV\in \mathcal{B}_1 \\
           U\prec_{\tree'} \sg{u} \\
           V\prec_{\tree'} \sg{v}}}\beta_{UV}
+\sum_{\substack{UV\in \bb_2 \\
           U\prec_{\tree'} \sg{u} \\
           V\prec_{\tree'} \sg{v}}}\beta_{UV}.$$
  \item For every edge $e=UV\in \bb_2$ between $U,V\in V(G_i)$, $UV\notin R(G^{[2]}_i)$. Moreover, at step $i-1$, either $U$ and $V$ are merged together, or if $U'$ and $V'$ denote the disjoint sets of $V(G_{i-1})$ that respectively include $U$ and $V$, we have $U'V'\in R(G^{[2]}_{i-1})$.
  \item For every remaining edge $e=UV\in \bb_1$, we must have $UV\in R(G^{[2]}_i)$.
  \item For every $U,V,Z\in V(\mathcal{T'})$ such that $U$ and $V$ are the two children of $Z$ in $\mathcal{T'}$, if $\ell(Z)=i-1$ then $\alpha'_Z=1$ if and only if $UV\in R(G_i)$.
 \end{itemize}
Furthermore, in the end, $\bb_1$ and $\bb_2$ have size at most $\Oo(d^24^dn)$.
\end{lemma}

\begin{proof}
  We will keep, at step $i$, $n-i+1$ lists $J^i_U$ for each node $U\in \bb_i(\tree')=V(G_i^{[2]})$, updated in time $\Oo(1)$ at each step.
  We will see that the list $J^i_U$ corresponds exactly to the red neighborhood of $U$ in $G^{[2]}_i$. 
 
 We start ($i=n$) with $J^i_{\sg{v}}:=\emptyset$ for each $u\in V(G)$ and $\alpha'_U:=0$ for every $U\in V(\mathcal T')$.
 Now let $i\in [n]$ and assume that the contraction between $G^{[2]}_i$ and $G^{[2]}_{i-1}$ is described by the triple $(U,V,Z)$ where $U,V\in V(G^{[2]}_i)$ and $Z\in V(G^{[2]}_{i-1})$. We also assume that the lists $J^i_U$ were already computed for every $U\in V(G^{[2]}_i)$.
 Again, we keep the convention for $\bb_2$ that if at some point, there are multiple edges of $\bb_2$ between two nodes $U,V\in V(\tree')$, then we replace them by a single edge whose label is the sum modulo $2$ of all their labels.
 At step $i$ of our algorithm we do the following:
 For every $W\in V(G^{[2]})\backslash\sg{U,V}$, we do not touch to the list $J^i_W$. We create the list $J^i_Z:=J^i_U\cup J^i_V\backslash \sg{U,V}$. These lists may be modified in what follows.
 \begin{enumerate}[$(i)$]
 \item\label{it2:case1} If $V\notin J_U^i$, and if $e=UV\in \bb_1$ with label $\beta_e$, then we remove $e$ from $\bb_1$ and add it in~$\bb_2$.
   Similarly, if $V\notin J^i_U$ and there is no edge between $U$ and $V$ in $\bb_1$, we add an edge $e=UV$ labeled with $\beta_e:=0$ in $\bb_2$. Intuitively, we mark there the position of every edge or non-edge of $G^{[2]}_i$ that disappears in $G^{[2]}_{i-1}$ because its two endpoints are merged together.
 \item\label{it2:case2} For every $W\in V(G^{[2]}_i)\backslash \sg{U,V}$ such that $W\notin J^i_U \cup J^i_V$, if $e_1:=UW$ and $e_2:=VW$ are in $\bb_1$ and $\beta_{e_1}=\beta_{e_2}$, then we remove $e_1$ and $e_2$ from $\bb_1$ and add the edge $e:=ZW$ in $\bb_2$ with label $\beta_e:=\beta_{e_1}$. Similarly, if $W\notin J^i_U \cup J^i_V$ and only one of the edges between $e_1$ and $e_2$ is in $\bb_1$, say $e_1$ with label $\beta_{e_1}=0$, then we remove it from $\bb_1$ and add in $\bb_2$ the edge $e:=ZW$ with label $\beta_e:=0$.
   Intuitively this corresponds to the case when the partial information we kept until now of the relations between $U$ and $W$ and between $V$ and $W$ is the same, so we ``lift it up'' in $\tree'$.
  \item\label{it2:case3} For every $W\in V(G^{[2]}_i)\backslash \sg{U,V}$ such that $W\notin J^i_U \cup J^i_V$, $e_1=UW\in \bb_1$ and $e_2=VW\in \bb_1$ have different labels $\beta_{e_1}\neq \beta_{e_2}$, we remove $e_1$ and $e_2$ from $\bb_1$ and add them in $\bb_2$. Moreover we add $W$ to the list $J^{i-1}_Z$ and reciprocally we add $Z$ to the list $J^{i-1}_W$. Similarly if we have $W\notin J^i_U \cup J^i_V$ and if $e_1=UW \in \bb_1$ but there is no edge between $V$ and $W$ in $\bb_1$, then we remove $e_1$ from $\bb_1$ and add in $\bb_2$ the edges $e_1$ and $e_2:=VW$ with label $\beta_{e_2}:=0$. We also apply the same changes to $J^{i-1}_W$ and $J^{i-1}_Z$. Intuitively this is where we mark the position of every edge or non-edge of $G^{[2]}_i$ that disappears to create a red edge where there was no error before.
  \item\label{it2:case4} For every $W\in V(G^{[2]}_i)\backslash \sg{U,V}$ such that $W\in J^i_U$ and $W\notin J^i_V$, if the edge $e=VW$ is in $\bb_1$, then we remove it from $\bb_1$ and add it in $\bb_2$. Similarly if there is no edge in $\bb_1$ between $V$ and $W$, we add in $\bb_2$ the edge $e:=VW$ labeled by $\beta_e:=0$. In both cases, we also add $Z$ in the list $J^{i-1}_W$.
  Intuitively this is where we mark the position of every edge or non-edge of $G^{[2]}_i$ that disappears because it is merged with some previous red edge of $G^{[2]}_i$.
  \item\label{it2:case5} If $U\in J_V^i$ then we set $\alpha'_Z:=1$.
 \end{enumerate}

That for every $u,v\in V(G), S_{uv}=S'_{uv}$ is immediate, as we only moved edges from $\bb_1$ to $\bb_2$, and merged some of the edges of $\bb_1$ in a way that does not change the value of $S_{uv}$.

However, it is not clear yet that the lists $J^i_U$ have bounded size at each step.
We show that indeed correspond to the red neighborhood of each vertex of $G^{[2]}_i$.
\begin{claim}
\label{clm: vrairouge}
 The list $J^i_U$ is exactly the red neighborhood of $U$ in $G^{[2]}_i$.
\end{claim}
\begin{proofofclaim}
 We show the claim by induction over decreasing values of $i\in [n]$. If $i=n$, then the claim is immediate, as there is no red edge in the partition of $G^{[2]}$ into singletons, and we initialized the lists $J^i_U$ to be empty.
 
 Assume now that the claim holds for some $i \in [2,n]$, and that the contraction between $G^{[2]}_i$ and $G^{[2]}_{i-1}$ is described by the triple $(U,V,Z)$ with $U,V\in V(G^{[2]}_i)$ and $Z\in V(G^{[2]}_{i-1})$.
 Then a vertex $W\in V(G^{[2]}_i)\backslash\sg{U,V}$ is a red neighbor of $Z$ in $G^{[2]}_{i-1}$ if and only if it was already a red neighbor of $U$ or $V$ in $G^{[2]}_i$, or if it has not the same relation with $U$ and $V$ in $G^{[2]}_i$. In the first case, $Z$ is indeed in $J^{i-1}_W$ thanks to~\pref{it2:case4}.
 Assume now without loss of generality that $UW\in E(G^{[2]}_i)$ while $VW\notin E(G^{[2]}_i)$. By~\cref{lem: partial-edges}, this means that if we take any $u\in U, v\in V, W\in W$, we have: $1=S'_{uw}\neq S'_{vw}=0$. Now observe that every common ancestor of $u$ and $w$ in $\tree'$ is also a common ancestor of $v$ and $w$, and vice versa.
 Hence we must have: 
 $$\sum_{\substack{U'W'\in \bb_1 \\
           U\prec U'\prec_{\tree'} \sg{u} \\
           W'\prec W'\prec_{\tree'} \sg{w}}}\beta_{U'W'}
+\sum_{\substack{U'W'\in \bb_2 \\
           U\prec U'\prec_{\tree'} \sg{u} \\
           W'\prec W'\prec_{\tree'} \sg{w}}}\beta_{U'W'}
           \neq 
\sum_{\substack{V'W'\in \bb_1 \\
           V\prec V'\prec_{\tree'} \sg{v} \\
           W'\prec W'\prec_{\tree'} \sg{w}}}\beta_{V'W'}
+\sum_{\substack{U'W'\in \bb_2 \\
           V\prec V'\prec_{\tree'} \sg{v} \\
           W'\prec W'\prec_{\tree'} \sg{w}}}\beta_{V'W'}.
$$

It is easy to see in the algorithm that as $W\notin J^i_U$, then there is no edge of $\bb_1$ between the subtree of $\tree'$ rooted in $W$ and the one rooted in $U$.
The same holds between $W$ and $V$.
Hence in the above inequality, the sums over edges of $\bb_1$ vanish.
One can also observe that the only edge in $\bb_2$ that may exist between the subtree rooted in $W$ and the one rooted in $U$ is $e_1:=UW$.
The same way, the only edge in $\bb_2$ between the subtree rooted in $W$ and the one rooted in $V$ is $e_1:=VW$.
Hence at least one of these edges must have been created in the algorithm.
It could not be in \pref{it2:case1} or \pref{it2:case4}, as we are not in these cases, so it must have been in \pref{it2:case4}.
Thus \pref{it2:case3} was executed according to $W$, so the lists $J_Z^{i-1}$ and $J_W^{i-1}$ have been correctly updated.
\end{proofofclaim}

By~\cref{clm: vrairouge}, we observe that the creation of an edge in $\bb_2$ corresponds exactly to the disappearance of the edges or non-edges of $G^{[2]}_i$, which proves the second item of the lemma. 
Moreover, observe that in the algorithm, every edge of $e=UV \in \bb_1$ such that $V\notin J^i_U$ is removed at one point from $\bb_1$, so by~\cref{clm: vrairouge} we are also done with the third item of the lemma.

Now by~\cref{clm: vrairouge} and the fact that $\tree'$ represents a contraction sequence as described in~\cref{lem: refinement}, we get that the size of each $L^i_U$ is $\Oo(d^22^d)$.
Hence, each vertex $U\in V(G^{[i]})$ has at most $\Oo(d^22^d)$ edges in $\bb_2$ linking it to another node of $\tree'$ with a higher label.
Hence at each step, $|\bb_1|+|\bb_2|$ increases by at most $\Oo(d^22^d)$, and eventually $|\bb_1|+|\bb_2|$ has the same order than at the start, i.e., $|\bb_1|+|\bb_2|=\Oo(d^24^dn)$.
The same way, we get that the execution time at each step of our algorithm is $\Oo(d^22^d)$.
So the total running time of the procedure is $\Oo(d^22^dn)$.
\end{proof}

Now we have everything in hand to conclude.
The set $\bb_2$ indicates the locations of the edges and non-edges of the graphs $G_i$ when they ``disappear.''
It only remains to give them the appropriate label: $1$ if it corresponds to an edge, and $0$ is it corresponds to a non-edge.
By construction, no edge of $\bb_1$ are placed ``below'' an edge of $\bb_2$, so we only need to communicate the information of the edges and the bits of the nodes that are ``above'' them.

\begin{lemma}
 \label{lem: b2tob}
 Given $\tree',\bb_1,\bb_2$ computed as before, we can construct in time $\Oo(d^24^dn)$ a set $\bb'$ of size $\Oo(d^22^dn)$ such that $(\tree', \bb')$ is a twin-decomposition of $G^{[2]}$ associated to the contraction sequence described in~\cref{lem: refinement}.
\end{lemma}
\begin{proof}
 We want each node to communicate its parity bit $\alpha_U$ to its children, and to transfer similarly the information of the edges that are still in $\bb_1$ to the edges of $\bb_2$ placed below them in $\tree'$. 
 
 Hence we now visit the nodes of $\tree'$ by \emph{increasing} values of their labels (i.e., we consider the associated contraction sequence of $G^{[2]}$ starting at its end).
 Assume that the contraction done between $G^{[2]}_i$ and $G^{[2]}_{i-1}$ is described by the triple $(U,V,Z)$ where $U,V\in V(G^{[2]}_i)$ and $Z\in V(G^{[2]}_{i-1})$.
 Again, if at some point there are multiple edges of $\bb_1$ between two nodes, we just replace them by a single one and sum their labels.
 At step $i$ we do the following:
 \begin{enumerate}[$(i)$]
  \item Increment the values $\alpha_U$ and $\alpha_V$ by $\alpha_Z$. 
  \item If $e=UV\in \bb_2$, then increment its label by $\alpha_Z$. If its new label is $0$, remove it from $\bb_2$. Otherwise keep it in $\bb_2$.
  \item If $e=UV\notin \bb_2$ and if $\alpha'_Z=0$, then
  add in $\bb_2$ the edge $e:=UV$ whose label is $\beta_{e}:=\alpha_{Z}$.
  \item If $e=UV\notin \bb_2$ and $\alpha'_Z=1$, then
  add in $\bb_1$ the edge $e:=UV$ whose label is $\beta_{e}:=\alpha_{Z}$.
\item For every $W\in V(G_{i-1})\backslash Z$ such that $e=ZW\in \bb_1$, if $e'=UW\in \bb_2$, increment the label of $e'$ by $\beta_e$. Otherwise if $UW\notin \bb_2$, add the edge $e':=UW$ in $\bb_1$ with label $\beta_{e'}:=\beta_e$.
  Do the same according to $V$.
  Finally remove $e$ from $\bb_1$.
 \end{enumerate}
 
 At step $i$, we keep the property that if there is an edge $e=UV\in \bb_1$ between $U,V\in V(G^{[2]}_i)$, then $UV\in R(G^{[2]}_i)$.
 Thus we ensure that every edge $\bb_1$ will ``encounter enough edges of $\bb_2$ below it to make it disappear''. Hence when the algorithm terminates, we have $\bb_1=\emptyset$.
 Moreover as the red degrees in $ G^{[2]}_i$ are bounded by $\Oo(d^22^d)$, we also have at each step $|\bb_1| = \Oo(d^22^dn)$.
 Similarly the total running time is $\mathcal{O}(d^22^dn)$.
 
 Now observe that for every two distinct vertices $u,v\in V(G)$, the sum $S'_{uv}$ does not change at each step.
 Hence, in the end, this sum is 1 if and only if the (unique) edge $e$ of $\bb_2$ that connects an ancestor of $\sg{u}$ with an ancestor of $\sg{v}$ in $\tree'$ has label 1.
 If we set :$\bb':=\sg{e\in \bb_2, \beta_e=1}$, then by the first item of~\cref{lemma: b1tob2}, $(\tree',\bb')$ is indeed a twin-decomposition of $G^{[2]}$.  
\end{proof}

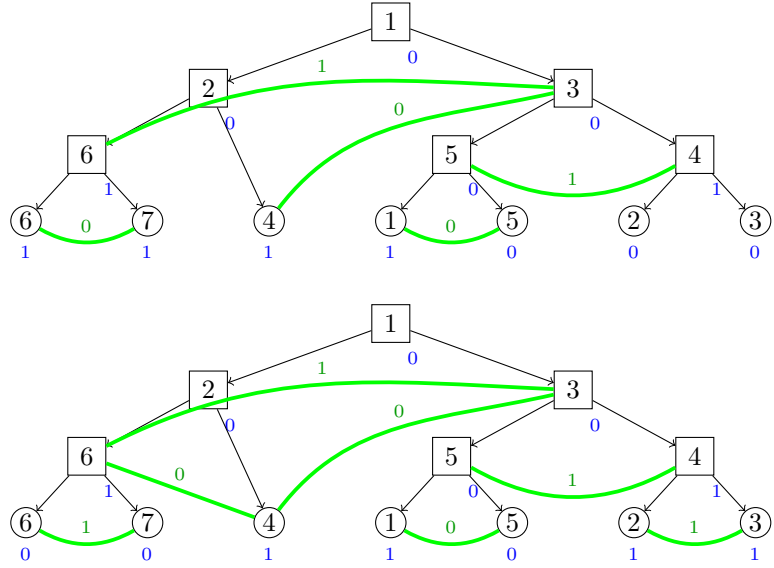
\begin{figure}[h!]
  \centering
  \begin{tikzpicture}[scale=0.8]
  \begin{scope}
   \foreach \i/\l/\p in {1/6/1,2/7/1,3/4/1,4/1/1,5/5/0,6/2/0,7/3/0}{
        \node[draw,circle,inner sep=0.04cm,minimum size=0.4cm, label={[label distance=0.001cm]below:$_{\textcolor{blue}{\p}}$}] (l\l) at (\i*2,0) {$\l$} ;
      }
      \foreach \i/\j/\l/\n/\p in {1.5/1.1/6/1/1, 4.5/1.1/5/2/0, 6.5/1.1/4/3/1, 2.5/2.2/2/4/0,  5.5/2.2/3/5/0, 4/3.3/1/6/0}{
        \node[draw,rectangle,inner sep=0.04cm,minimum size=0.5cm,label={[label distance=0.001cm]-70:$_{\textcolor{blue}{\p}}$}] (v\n) at (\i*2,\j) {\l} ;
      }
     \foreach \i/\j in {6/1,7/1, 4/4, 1/2, 5/2, 2/3, 3/3}{
       \draw[<-] (l\i) -- (v\j) ;
     }
     \foreach \i/\j in {1/4, 4/6, 2/5, 3/5, 5/6}{
       \draw[<-] (v\i) -- (v\j) ;
     }     
     
     \foreach \i/\j in {6/7,1/5}{
        \draw[line width=0.05cm,green] (l\i) to [bend right] node[midway, below, label={above=1mm:$_{\textcolor{green!60!black}{0}}$}]{} (l\j) ;
    }
      
      \foreach \i/\j/\l in {5/4/0}{
        \draw[line width=0.05cm,green] (v\i) to [bend right, out=-10] node[midway, below, label={ $_{\textcolor{green!60!black}{\l}}$}]{} (l\j) ;
      }

      \foreach \i/\j/\l in {5/1/1}{
        \draw[line width=0.05cm,green] (v\i) to [bend right, in=200, out=-10] node[midway, below, label={ $_{\textcolor{green!60!black}{\l}}$}]{} (v\j) ;
      }      

      \foreach \i/\j/\l in {2/3/1}{
        \draw[line width=0.05cm,green] (v\i) to [bend right] node[midway, below, label={ $_{\textcolor{green!60!black}{\l}}$}]{} (v\j) ;
      }      
      
  \end{scope}
  \begin{scope}[yshift=-5cm]
   \foreach \i/\l/\p in {1/6/0,2/7/0,3/4/1,4/1/1,5/5/0,6/2/1,7/3/1}{
        \node[draw,circle,inner sep=0.04cm,minimum size=0.4cm, label={[label distance=0.001cm]below:$_{\textcolor{blue}{\p}}$}] (l\l) at (\i*2,0) {$\l$} ;
      }
      \foreach \i/\j/\l/\n/\p in {1.5/1.1/6/1/1, 4.5/1.1/5/2/0, 6.5/1.1/4/3/1, 2.5/2.2/2/4/0,  5.5/2.2/3/5/0, 4/3.3/1/6/0}{
        \node[draw,rectangle,inner sep=0.04cm,minimum size=0.5cm,label={[label distance=0.001cm]-70:$_{\textcolor{blue}{\p}}$}] (v\n) at (\i*2,\j) {\l} ;
      }
     \foreach \i/\j in {6/1,7/1, 4/4, 1/2, 5/2, 2/3, 3/3}{
       \draw[<-] (l\i) -- (v\j) ;
     }
     \foreach \i/\j in {1/4, 4/6, 2/5, 3/5, 5/6}{
       \draw[<-] (v\i) -- (v\j) ;
     }
     \foreach \i/\j/\l in {6/7/1,1/5/0,2/3/1}{
        \draw[line width=0.05cm,green] (l\i) to [bend right] node[midway, below, label={above=1mm:$_{\textcolor{green!60!black}{\l}}$}]{} (l\j) ;
    }
     
    \draw[line width=0.05cm,green] (v1) to [] node[midway, below, label={ $_{\textcolor{green!60!black}{0}}$}]{} (l4) ;     

      \foreach \i/\j/\l in {5/4/0}{
        \draw[line width=0.05cm,green] (v\i) to [bend right, out=-10] node[midway, below, label={ $_{\textcolor{green!60!black}{\l}}$}]{} (l\j) ;
      }

      \foreach \i/\j/\l in {5/1/1}{
        \draw[line width=0.05cm,green] (v\i) to [bend right, in=200, out=-10] node[midway, below, label={ $_{\textcolor{green!60!black}{\l}}$}]{} (v\j) ;
      }      

      \foreach \i/\j/\l in {2/3/1}{
        \draw[line width=0.05cm,green] (v\i) to [bend right] node[midway, below, label={ $_{\textcolor{green!60!black}{\l}}$}]{} (v\j) ;
      }      
     
  \end{scope}

  \end{tikzpicture}
  \caption{Top: The result of the algorithm from \cref{lemma: b1tob2} applied to the couple $(\mathcal T', \bb_1)$ from \cref{fig:B1}. The green edges represent the set $\bb_2$. For the sake of clarity, we did not add the bits $\alpha'_U$. 
  Bottom: The result of the algorithm from \cref{lem: b2tob} applied to the triple $(\mathcal T', \bb_1, \bb_2)$ to the left.}
\label{fig:B2}
\end{figure}

Now if we put together \cref{lem: refinement,lem: listeadj,lem: pointers,lem: partial-edges,lemma: b1tob2,lem: b2tob} we get the following result which is exactly \cref{thm: matmult} when $q=2$:

\begin{theorem}
 \label{thm: matmult2}
 There exists a $\mathcal{O}(d^24^dn)$-time algorithm that, given a twin-decomposition $(\tree, \bb)$  of width $d$ of a graph $G$ with $n$ vertices, outputs a twin-decomposition of width $\Oo(d^22^d)$ of $G^{[2]}$. 
\end{theorem}

\begin{remark}
 The previous algorithm can also be adapted to compute a twin-decomposition of $G^2$ associated to the refinement of \cref{rem: G2} in the same complexity.
\end{remark}

\paragraph*{The algorithm for general values of $q\geq2$}

As mentioned in the overview of the previous proof, there are only few changes to do in order to prove \cref{thm: matmult} in its full generality.
We now explain which parts have to be modified.

We let $G$ be a graph with edges labeled by a function $\nu:E(G)\to \mathbb F_q$.
For every $l\in \mathbb F_q$ we let $G_l$ be the graph with vertex set $V(G)$ and edge set $\nu^{-1}(l)$.
For every $u\in V(G)$ and $U\subseteq V(G)$, we let $\degr^{(l)}_U(u):= |\sg{v\in U~:~\nu(uv)=l}|$.
Again let $(\mathcal T, \bb)$ be a $d$-sequence for $G$, with associated partition sequence $\mathcal P_n, \ldots, \mathcal P_1$.
The twin-decomposition $(\mathcal{T'}, \bb')$ we want to compute corresponds to the refinement of $\mathcal P_n, \ldots, \mathcal P_1$ where for every $i\in [n], U\in \mathcal P_i$, if $W_1, \ldots, W_k$ denote the at most $d$ red neighbors of $U$ in $G_i$, then 
$\mathcal P'_i$ contains all the following subsets of $U$ which are nonempty:

$$U_{(p,b_1, \ldots, b_k)} := \{x\in U~:~\forall i \in [k], \sum_{l\in \mathbb F_q} \degr^{(l)}_{W_i}(x)\cdot l = b_i $$
$$\text{and}~\sum_{l\in \mathbb F_q}\mathrm{deg}^{(l)}_U(x)\cdot l = p \},$$

for every value of $(p,b_1,\ldots,b_k)\in \mathbb{F}_q^{k+1}$, where for each $n\in \mathbb N$ and each $l\in \mathbb F_q$ we let $n\cdot l$ denote the sum of $l$ with itself $n$ times (which is not the same thing that the multiplication in $\mathbb F_q$). In particular, note that for each $n\in \mathbb N$ and $l\in \mathbb F_q$, $n\cdot l = (n \pmod q) \cdot l$.
As previously, note that each set $U$ is partitioned into at most $q^{d+1}$ subsets $U_{(p,b_1,\ldots,b_k)}$. The reader can check that the proof of \cref{lem: refinement} generalizes, i.e., that we can complete arbitrarily this refinement and obtain a $D$-sequence for $G^{[q]}$ with $D:=(d^2+d+1)q^{d+1}-1$. 
We claim that the proof of \cref{lem: pointers} also generalizes, with the additional factor $m(q)$ in the time complexity, as one needs to make a (constant) number of elementary computations in $\mathbb{F}_q$ at each update of the pointers $q[U,p,b_1,\ldots,b_k]$. The only change is that we must pre-compute the values $s_U:=|U|\pmod q$ for each $U\in V(\mathcal T)$.

We let $\lambda : E(G^{[q]})\to \mathbb F_q$ denote the labeling function of the edges of $G^{[q]}$.
We still need to compute some values $\alpha_U\in \mathbb F_q$ for $U\in V(\mathcal{T'})$ and a
set of labeled transversal edges $\bb_1$ of $\mathcal{T'}$, where for each $e\in \bb_1$, $\beta_e\in \mathbb F_q$ denotes the label of $e$ such that for every $u,v\in V(G^{[q]})$: 
$$\lambda(uv)=\sum_{\substack{W\in V(\tree')\\
W\prec_{\tree'} \sg{u} \\
           W\prec_{\tree'} \sg{v}}}\alpha_W 
+\sum_{\substack{UV\in \mathcal{B}_1 \\
           U\prec_{\tree'} \sg{u} \\
           V\prec_{\tree'} \sg{v}}}\beta_{UV}.
$$

The algorithm described in \cref{lem: partial-edges} can be generalized as desired, and we get a set $\bb_1$ as desired of size $\mathcal O(d^2 q^{2d}n)$.
Its running time becomes $\mathcal O(m(q)d^2q^{2d})$.
Eventually, \cref{lemma: b1tob2,lem: b2tob} generalize in a straightforward way, up to a multiplicative cost $m(q)$ for the time complexity.
Hence the overall complexity is $\Oo(m(q)d^2 q^{2d}n)$ so we are done with \cref{thm: matmult}.


\end{document}